\documentclass[11pt]{article}
\usepackage[utf8]{inputenc}
\usepackage{amsmath}
\usepackage{amssymb}
\usepackage{amsthm}
\usepackage{xcolor}
\usepackage{hyperref}
\usepackage{fullpage}
\usepackage{url}
\usepackage{authblk}
\usepackage{xfrac}
\usepackage{appendix}
\usepackage{makecell}

\usepackage[inline]{enumitem}

    \usepackage{caption}
    \usepackage{setspace}
\usepackage[left=1in,right=1in,top=1.35in,bottom=1.4in]{geometry}
\usepackage{algorithm}
\usepackage{algpseudocode}
\usepackage{natbib}
\hypersetup{
	colorlinks=true,       
	linkcolor=blue,          
	citecolor=blue,        
	filecolor=blue,      
	urlcolor=blue           
}

\numberwithin{equation}{section}
\numberwithin{figure}{section}

\allowdisplaybreaks
\theoremstyle{plain}
\newtheorem{lemma}{Lemma}[section]

\newtheorem{corollary}[lemma]{Corollary}

\newtheorem{theorem}[lemma]{Theorem}

\theoremstyle{definition}

\newtheorem{remark}[lemma]{Remark}

\newtheorem{example}[lemma]{Example}

\newcommand{\Rbb}{\mathbb{R}}

\newcommand{\Nbb}{\mathbb{N}}

\newcommand{\Pbb}{\mathbb{P}}
\newcommand{\Ebb}{\mathbb{E}}
\newcommand{\Var}{\textnormal{Var}}
\newcommand{\Cov}{\textnormal{Cov}}

\newcommand{\Pcal}{\mathcal{P}}
\newcommand{\diff}{\textnormal{d}}
\newcommand{\trace}{\textnormal{tr}}

\newcommand{\covspace}{\mathbb{K}}
\newcommand{\symspace}{\mathbb{S}}

\newcommand{\id}{\textnormal{id}}
\newcommand{\Hcal}{\mathcal{H}}

\newcommand{\Ncal}{\mathcal{N}}
\newcommand{\Supp}{\textnormal{supp}}

\DeclareMathOperator{\TV}{TV}

\newcommand{\Perp}{\perp\!\!\!\!\perp} 

\newcommand{\B}{\mathcal{B}}

\newcommand{\distrcstr}{\mathcal{DCB}}
\newcommand{\varcstr}{\mathcal{VCB}}
\newcommand{\gencstr}{\mathcal{GCB}}

\newcommand{\transport}{\mathbf{t}}
\newcommand{\BW}{\Pi}

\begin{document}

\title{Constrained Denoising, Empirical Bayes, and Optimal Transport}

\date{\today}
\author[1]{Adam Quinn Jaffe\thanks{\texttt{a.q.jaffe@columbia.edu}}}
\author[2]{Nikolaos Ignatiadis\thanks{NI gratefully acknowledges support from NSF (DMS-2443410).}}
\author[1]{Bodhisattva Sen\thanks{BS gratefully acknowledges support from NSF (DMS-2311062)}}

\affil[1]{Department of Statistics, Columbia University, New York, NY}
\affil[2]{Department of Statistics and Data Science Institute, University of Chicago, Chicago, IL}

	

\maketitle

\begin{abstract}
    In latent variables models, two important goals are denoising and deconvolution: denoising aims to estimate the latent variables, whereas deconvolution aims to estimate the distribution of the latent variables.
    As has been recognized in the literature over the last century, these two goals are fundamentally in tension, since denoising yields a poor estimate of the distribution of the latent variables due to shrinkage, and deconvolution yields a distribution-valued estimate that carries no unit-specific information.
    In this paper, we provide a systematic study of denoisers, and empirical Bayes approximations thereof, which attain optimal denoising error subject to the constraint that the distribution of the denoised data matches, in some sense, the distribution of the latent variables.
    Our insight is that optimal transport allows practitioners to navigate the tension between denoising and deconvolution.
    More precisely, we propose a modular methodology that combines any suitable unconstrained empirical Bayes denoiser (arising, e.g., via $F$-modeling, $G$-modeling, conjugate-parametric models) with any suitable information about the distribution of the latent variables (e.g., its moments, support, or an approximation of the entire distribution via deconvolution) into a single denoised data set.
    We prove explicit rates of convergence for our proposed methodologies, and we apply the resulting methods in applications in astronomy, baseball analytics, and marketing.
\end{abstract}

\medskip

\small

\vspace{-0.4cm}

\medskip


\normalsize


\section{Introduction}
\label{sec:introduction}

\subsection{Problem Statement and Motivation}\label{subsec:prob-statement}

This work is focused on the statistical problem of denoising, which concerns
the following model for the joint distribution for a pair of random variables
$(\Theta,Z)\in\Rbb^m\times\Rbb^d$ for $m\ge1$ and $d\ge 1$:
\begin{equation}\label{eq:Mdl}
\Theta \sim G, \qquad \mbox{and} \qquad (Z \mid \Theta = \theta) \sim P_\theta.
\end{equation}
Here, $G$ is an unknown distribution in $\Rbb^m$ and $\{P_{\theta}\}_{\theta}$
is a known family of probability distributions on $\Rbb^d$, and we write $F$ for the marginal distribution of $Z$. We assume that
there exist $(\Theta_1,Z_1),\ldots, (\Theta_n,Z_n)$ which are independent
identically distributed (i.i.d.) pairs from the above model; we refer to
$\Theta_1,\ldots,\Theta_n$ as the \textit{latent variables}, we refer to
$Z_1,\ldots, Z_n$ as the \textit{observations}, and \textit{denoising} is the
task of estimating/predicting the latent variables from the observations. Each index $i$ refers to a unit which is of interest in
its own right: in our applications (see Section~\ref{sec:app}), unit $i$ is a star in a large
astronomical catalog (and $\Theta_i$ its latent chemical abundances), a
baseball player in his rookie season (and $\Theta_i$ his latent batting
skill), or a customer segment in a
marketing experiment (and $\Theta_i$ the effect of an intervention on that
segment); elsewhere, unit $i$ may be a site in a multisite trial (and
$\Theta_i$ the treatment effect at that site), or a small area in a survey (and $\Theta_i$ the average wage in that area). Denoising thus amounts to reporting a point estimate of $\Theta_i$ for each unit $i$.

More precisely, the goal of denoising is to construct a function $\delta:\Rbb^d\to\Rbb^m$ such that the risk $\Ebb[\|\delta(Z)-\Theta\|^2]$ is small, with expectation taken over the joint distribution of $(\Theta,Z)$.
If the distribution $G$ is known, then the best choice is the posterior mean $\delta_{\B}(z) := \Ebb[\Theta\,|\,Z = z]$, called the \textit{Bayes denoiser} (with respect to the squared error loss).
If the distribution $G$ is not known, then it is often still possible to construct some function $\hat{\delta}_{\B}$, called an \textit{empirical Bayes (EB) denoiser} \citep{Robbins1956,efron2019bayes}, which achieves a risk that is nearly equal to that of the Bayes denoiser.

The fact that $\hat{\delta}_{\B}$ achieves nearly the optimal risk, however, does not imply that the distribution of $\hat{\delta}_{\B}(Z_1),\ldots, \hat{\delta}_{\B}(Z_n)$ must be close to the distribution of $\Theta_1,\ldots,\Theta_n$. Instead, it often implies that the former is ``overshrunk'' compared to the latter.
This can be seen already at the population level, since $\hat{\delta}_{\B}$ targets $\delta_{\B}$, and the law of total covariance yields the following, where $\preceq$ refers to the positive semi-definite order:
\begin{equation}\label{eqn:M-prec-A}
    \Cov(\delta_{\B}(Z)) = \Cov(\Theta) -\Ebb[\Cov(\Theta\,|\,Z)] \preceq \Cov(\Theta).
\end{equation}
As a concrete example of this phenomenon, consider the simple setting where $\Theta\sim\Ncal(0,1)$ and $P_{\theta} = \Ncal(\theta,1)$, and note that we have $\delta_{\B}(Z)\equiv Z/2\sim\Ncal(0,\sfrac{1}{2})\neq \Ncal(0,1)$. This
phenomenon has been noted by many authors over the last hundred years,
beginning with \citet{eddington1940correction}; we review the relevant
literature in Section~\ref{subsec:related-lit}.

\begin{figure}
    \centering
\includegraphics[width=\linewidth]{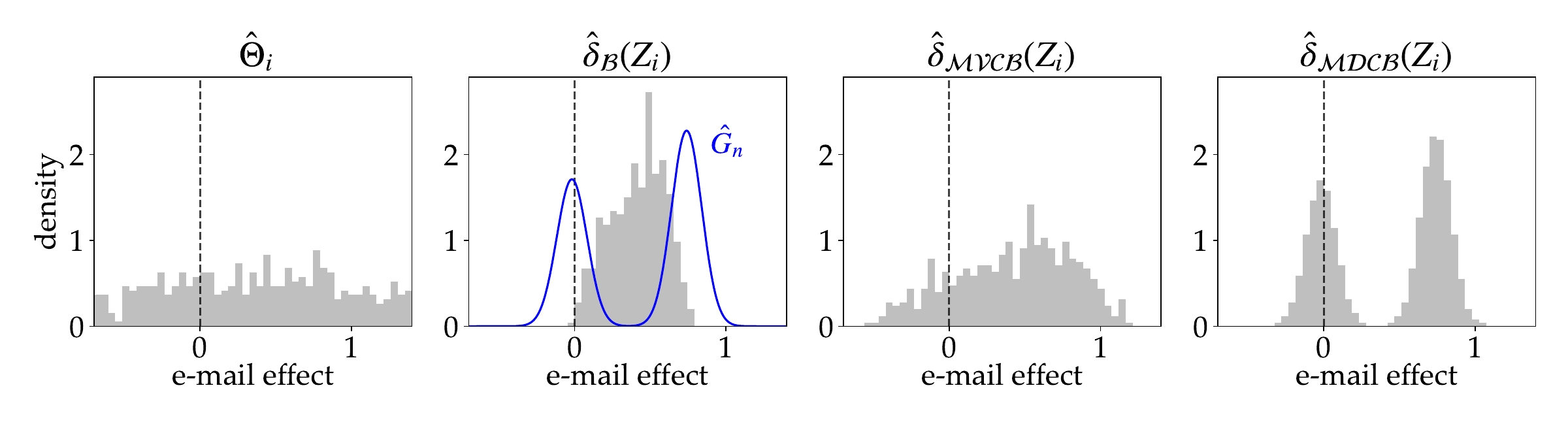}
\caption{Denoising the effect of marketing e-mails across customer segments in the randomized experiment of Subsection~\ref{subsec:marketing}.
We plot the set of natural frequentist estimates which do not share information across segments (first), along with the unconstrained EB denoised data and a deconvolution estimate of the latent variable distribution (second);
the EB denoiser overshrinks and gives the false impression that marketing e-mails have a positive effect on every customer segment, while the deconvolution estimate reveals that the effects come from two distinct components, one of which is centered at zero.
Then, we plot the marginal variance-constrained EB denoised data (third) which gives the correct dispersion, and the marginal distribution-constrained EB denoised data (fourth) which denoises the customer segments while placing them into two distinct components.}
\label{fig:hillstrom}
\end{figure}

Our formulation, which builds on a long tradition of constrained Bayes estimation \citep{LouisI, Ghosh}, is a
practical compromise between denoising and \textit{deconvolution} (i.e.,
the estimation of $G$ itself) and its informal goal is reliable scientific
communication~\citep{shen1998triplegoal}. On the one hand, an accurate point estimate should be reported for each unit $i$, since each unit carries a specific meaning for some particular stakeholders. On the other hand, the reported estimates will also be examined
as an ensemble: a researcher may plot a histogram of the estimates, for
instance to situate a unit of interest within the population, and, owing to
the underdispersion in~\eqref{eqn:M-prec-A}, such a plot is easily
misinterpreted.

Figure~\ref{fig:hillstrom} illustrates this tension through a data set from a marketing
experiment \citep{Hillstrom2008MineThatData}, in which unit $i$ is a
customer segment and $\Theta_i$ is the log-odds ratio for a website visit between customers in segment $i$ who did and did not receive a promotional e-mail
(see Subsection~\ref{subsec:marketing} for details).
The natural frequentist baseline is the sample log-odds ratio which leads to an overdispersed set of estimates.
The empirical Bayes estimates
$\hat{\delta}_{\B}(Z_1),\ldots,\hat{\delta}_{\B}(Z_n)$ in the second panel are positive for almost every segment, so a
stakeholder inspecting this histogram could wrongly conclude that the promotional
e-mail raises the odds of a visit in every segment.
A separate goal is to estimate $G$ directly, for which a  deconvolution estimate $\hat{G}_n$ (also shown in the second panel) indeed suggests that a sizable fraction of segments have
negative log-odds ratios.
Reconciling these apparent contradictions between
\emph{denoising} and \emph{deconvolution} for a non-expert audience is an
obstacle to reliable communication; moreover, an individual point estimate
of $\Theta_i$ has, in general, no meaningful position within a separately estimated $\hat{G}_n$. Constrained
denoising keeps a point estimate for each unit while ensuring that, as an
ensemble, the estimates agree with the deconvolution estimate of $G$, at a modest cost in denoising accuracy.

The goal of this work is to construct fully-data-driven denoisers $\hat{\delta}$ whose outputs $\hat{\delta}(Z_1),\ldots, \hat{\delta}(Z_n)$ are distributed, in a suitable sense, like the latent variables $\Theta_1,\ldots, \Theta_n$, while achieving a risk near that of the optimal constrained denoiser.
The panels of Figure~\ref{fig:hillstrom} preview two such denoisers we study.
If we trust the smooth deconvolution estimate $\hat{G}_n$, there appear to be two groups of segments: a large group for which the e-mail clearly raises the odds of a visit, and a second group centered at a log-odds ratio of zero but spread across both signs, so that the e-mail has a negative effect for some of these segments.
The \emph{distribution-constrained} denoiser in the fourth panel requires that the reported estimates have the same distribution as our estimated distribution of the latent log-odds ratios, and it recovers the two groups of segments. Because the distribution-constrained denoiser relies on full deconvolution, which is a difficult statistical problem, we also study constrained denoisers that require estimating only targeted properties of $G$.
The \emph{variance-constrained} denoiser in the third panel requires only that the estimates match the first two moments of $G$; it restores the dispersion of the effects, and already this is enough to avoid the misleading impression that no segment has a negative effect. 
We also propose \emph{general-constrained} denoisers (not shown in the figure)  that match any prescribed collection of linear functionals of $G$ and so can interpolate between distribution- and variance-constrained denoisers.

\subsection{Classical Denoising and Three Constrained Denoising Problems}
\label{subsec:decisiontheory}

We now describe the classical denoising problems and three constrained variants that are of interest in this work.
First, we consider the (unconstrained) Bayes denoising problem,
\begin{equation}\label{eqn:Bayes}
    \underset{\delta:\Rbb^d\to\Rbb^m}{\textnormal{minimize}}\quad\Ebb\left[\|\delta(Z)-\Theta\|^2\right],
\tag{$\B$}
\end{equation}
whose solution is the posterior mean $\delta_{\B}(z) =\Ebb[\Theta\,|\,Z=z]$, and we next add suitable constraints.

\medskip

\noindent \textbf{Variance-Constrained Denoising.} 
The following problem has been studied by~\citet{LouisI, Ghosh}, and several others, and is precisely given as
\begin{equation}\label{eqn:pop-varcstr}
    \underset{\delta:\Rbb^d\to\Rbb^m}{\textnormal{minimize}}\quad\Ebb\left[\|\delta(Z)-\Theta\|^2\right]\qquad\textnormal{s.t.}\qquad \Ebb[\delta(Z)] = \Ebb[\Theta] \quad\textnormal{and}\quad \Cov(\delta(Z)) = \Cov(\Theta).
    \tag{$\varcstr$}
\end{equation}
Problem~\eqref{eqn:pop-varcstr} aims to find a denoiser $\delta$ that minimizes the Bayes risk subject to the constraint that the first two moments of $\delta$ match those of the latent variable $\Theta$, thus mitigating the shrinkage of~\eqref{eqn:M-prec-A}.

\medskip

\noindent \textbf{Distribution-Constrained Denoising.}
The second problem of interest has been recently studied by~\citet{GarciaTrillosSen}, and it is given as
\begin{equation}\label{eqn:pop-distrcstr}
    \underset{\delta:\Rbb^d\to\Rbb^m}{\textnormal{minimize}}\quad\Ebb\left[\|\delta(Z)-\Theta\|^2\right]\qquad\textnormal{s.t.}\qquad \delta(Z) \overset{\mathcal{D}}{=} \Theta.\tag{$\distrcstr$}
\end{equation}
Here, rather than enforcing constraints on the first two moments, we require that the distribution of $\delta(Z)$ matches exactly the distribution of the latent variable $\Theta$.
(As we explain later, it is also possible to require that the distribution of $\delta(Z)$ equals some $\tilde{G}$ which need not coincide with $G$.)

\medskip

\noindent \textbf{General-Constrained Denoising.} 
We may accommodate general constraints via the problem
\begin{equation}\label{eqn:pop-gencstr}
    \underset{\delta:\Rbb^d\to\Rbb^m}{\textnormal{minimize}}\quad\Ebb\left[\|\delta(Z)-\Theta\|^2\right]\qquad\textnormal{s.t.}\qquad \Ebb[\psi_{\ell}(\delta(Z))] = \Ebb[\psi_{\ell}(\Theta)] \quad \textnormal{ for all } 1\le \ell \le k
    \tag{$\gencstr$}
\end{equation}
where $\psi_1,\ldots,\psi_k:\Rbb^m\to\Rbb$ is some prescribed set of measurable functions.
This formulation is general enough to cover a wide variety of problems; for example, it encompasses~\eqref{eqn:pop-varcstr} when $m=1$ by taking $\psi_1(z):=z$ and $\psi_2(z):=z^2$, it also encompasses~\eqref{eqn:pop-distrcstr} in a limiting sense when $\psi_1,\psi_2,\ldots$ are taken to be a suitable sequence of bounded measurable functions. The added generality is useful in some novel applications and substantially extends the scope of constrained denoising.

\medskip

The reason why the problems \eqref{eqn:pop-varcstr}, \eqref{eqn:pop-distrcstr}, and \eqref{eqn:pop-gencstr} are difficult to analyze is that distributional constraints are non-convex in terms of the denoiser. However, it is one of the fundamental insights of the theory of optimal transport (OT) \citep{Villani} that problems involving such non-convex distributional constraints are in many settings equivalent to a suitable linear program over a space of couplings, and that analysis of the latter provides theoretical and computational advantages for the former.
This connection has been previously made by \cite{GarciaTrillosSen} at the population-level, and the present work extends this paradigm to the EB setting. This allows us to comprehensively study all the problems above, and to derive new results even for problem \eqref{eqn:pop-varcstr} which has been extensively studied in previous literature.

\subsection{Summary of Results}\label{subsec:results}

The main results of this paper concern the various constrained Bayes problems above (Section~\ref{sec:pop}), their EB approximations (Section~\ref{sec:emp}), and various extensions and applications.
We now describe these points in more detail.

\medskip

\noindent \textbf{Variance-Constrained Denoising.} First, we provide a detailed study of problem~\eqref{eqn:pop-varcstr}, which has been introduced in \citet{LouisI} and studied in many subsequent works (e.g., \citet{Ghosh, ghosh1999adjusted}).
At the population level, we provide a closed-form solution for the variance-constrained Bayes denoiser $\delta_{\varcstr}$ (Theorem~\ref{thm:VCB-solution}) which shows that $\delta_{\varcstr}$ is an affine function of the Bayes denoiser $\delta_{\B}$; crucially, this characterization arises by casting the denoising problem in terms of the geometry of the Bures-Wasserstein space of Gaussian OT, for which many explicit formulas are known.
We further show that the price, in terms of risk, that one needs to pay for matching the first two moments of the denoisers to that of the unknown latent variables is modest; that is, the variance-constrained Bayes risk $R_{\varcstr}$ (the infimum of \eqref{eqn:pop-varcstr}) is upper bounded by twice the optimal Bayes risk (the infimum of~\eqref{eqn:Bayes}) (Corollary~\ref{cor:risks}) and upper bounded by the risk of any unbiased denoiser (Lemma~\ref{lem:risk-comparison}).
At the empirical level, we provide an EB denoising procedure (Algorithm~\ref{alg:Gaussian-mod-NP-prior}) which computes a denoiser $\hat{\delta}_{\varcstr}$ whose distance to the oracle $\delta_{\varcstr}$ exhibits an explicit rate of convergence (Theorem~\ref{thm:varcstr-gm-npp}).

\medskip

\noindent \textbf{Distribution-Constrained Denoising.} Next, we consider the distribution-constrained Bayes problem~\eqref{eqn:pop-distrcstr}.
At the population level, we review conditions from \cite{GarciaTrillosSen} guaranteeing the existence of a uniquely defined distribution-constrained Bayes denoiser $\delta_{\distrcstr}$ (Theorem~\ref{thm:pop-distrcstr}); at the empirical level, we provide an EB denoising scheme (Algorithm~\ref{alg:distr-cstr}) which produces a denoiser $\hat{\delta}_{\distrcstr}$ that exhibits an explicit rate of convergence to the oracle constrained denoiser (Theorem~\ref{thm:emp-distrcstr}).
This work extends recent results in \citet{GarciaTrillosSen} concerning the population-level problem; we provide rigorous statistical guarantees for an EB procedure which was stated (but not rigorously studied) in \citet[Appendix~E]{GarciaTrillosSen}.

\medskip

\noindent \textbf{General-Constrained Denoising.} Third, we consider the general-constrained Bayes problem~\eqref{eqn:pop-gencstr}, which encapsulates a wide variety of practical constraints (for instance, the fourth-moment constraints of \cite{armstrong2022robust} or support constraints such as nonnegativity).

At the population level, we give conditions under which there is a uniquely-defined general-constrained Bayes denoiser $\delta_{\gencstr}$ (Theorem~\ref{thm:gen-cstr-num}); at the empirical level, we provide an EB denoising scheme which produces a denoiser $\hat{\delta}_{\gencstr}$ that exhibits an explicit rate of convergence to the oracle constrained denoiser (Theorem~\ref{thm:gencstr-informal}).

\medskip

A fundamental flexibility of our paradigm is that it allows one to incorporate the~\eqref{eqn:pop-varcstr},~\eqref{eqn:pop-distrcstr}, and~\eqref{eqn:pop-gencstr} constraints into an arbitrary unconstrained EB denoiser; in fact, the constrained denoisers inherit explicit rates of convergence from the underlying unconstrained denoisers and the other elements of the problem being estimated.
The unconstrained denoiser may arise via $G$-modeling, $F$-modeling, parametric EB, smoothed versions thereof, hybrids thereof, and more.
This allows us to determine the rates of convergence for our constrained EB denoising procedures in many different settings of interest (e.g., light- and heavy-tailed $G$, Gaussian and Poisson likelihood $\{P_{\theta}\}_{\theta}$, conjugate parametric models), and we give a detailed analysis of the resulting rates of convergence in Appendix~\ref{app:rates}.
A notable example is the case of a nonparametric latent variable distribution and a Gaussian likelihood where we use nonparametric maximum likelihood (NPMLE) for $G$-modeling; our results imply that the rate of convergence of $\hat{\delta}_{\varcstr}$ is nearly-parametric, and the rate of convergence of $\hat{\delta}_{\distrcstr}$ is slow (due to the slow minimax rate of nonparametric deconvolution, i.e., \citet{CarollHallDeconvolution, ZhangDeconvolution, FanDeconvolution}).

Our EB denoisers exhibit a rate of convergence with respect to the corresponding oracle in terms of a particular notion of distance which allows us to control both an excess denoising error (closely related to the excess risk) and an excess deconvolution error.
Although the risks of the oracle constrained denoisers are necessarily suboptimal, we provide some evidence towards such constraints providing a balance between denoising and deconvolution; see Corollary~\ref{cor:risks}, Lemma~\ref{lem:risk-comparison}, and the numerical simulations of Section~\ref{sec:sim}.
In all cases, we observe that constrained EB denoisers have only modestly larger risk than the (unconstrained) oracle Bayes denoiser, while achieving substantially smaller deconvolution error.

We also illustrate our methodology in three applications (Section~\ref{sec:app}).
In an application to astronomy (Section~\ref{subsec:astro}), we use our results to denoise the distribution of latent chemical abundances in a large catalog of observed stars, furthering the analyses of \cite{Ratcliffe_2020, Soloff}.
In an application to baseball (Section~\ref{subsec:baseball}), we use our results to denoise the distribution of latent batting skill for players who are promoted from minor- to major-league play.
In the previously-discussed application to marketing (Section~\ref{subsec:marketing}), we denoise the effect of promotional email interventions on consumer behavior.
Because these applications involve some degree of heterogeneity in the observations (e.g., heteroskedasticity in the Gaussian case) we include a detailed discussion (Section~\ref{sec:heterosked-body}) of some considerations for extensions to the heterogeneous case.


\subsection{Related Literature}\label{subsec:related-lit}

We build upon a long tradition in the EB literature~\citep{Robbins1956, efron2019bayes} of constructing nonparametric denoisers that can provably perform nearly as well as oracle denoisers that operate under knowledge of the true latent variable distribution. The statistical properties of unconstrained denoisers that mimic the Bayes denoisers in~\eqref{eqn:Bayes} have been increasingly well understood in recent years~\citep{zhang1997empirical, brown2009nonparametric, JiangZhang, efron2011tweedie, SahaGuntuboyina,PolyanskiyWuPoisson, ShenWu, ignatiadis2023empiricala, barbehenn2023nonparametric, Soloff, ghosh2025steins}. Here we lift existing results on unconstrained denoisers and provide further theory on EB approximations to constrained denoisers.

The fallacy of applying downstream statistical analyses to EB-denoised data sets has been noted by many authors over the last hundred years, e.g.,~\citet{LouisI, Ghosh,bloom2017using, agarwal2020data,chen2025empirical}, and in textbooks [\citealp[p.~88]{raudenbush2002hierarchical}; \citealp[p.~316]{rao2015small}].
In fact, this fallacy has been recognized since the early developments of EB in astronomy.
To wit, Sir Arthur Eddington introduced what is now known as Tweedie's formula~\citep{efron2011tweedie} to denoise noisy parallax observations; this method was communicated by Sir Frank Dyson~\citeyearpar{dyson1926method}. Later,~\citet{eddington1940correction} lamented that his formula ``is not infrequently employed in an illegitimate way'':
Astronomers would use the empirical distribution of denoised estimates $\widehat{\Theta}_1,\dotsc,\widehat{\Theta}_n$ and interpret it as an improved estimate of the distribution of $\Theta_1,\dotsc,\Theta_n$. Eddington explains that this practice is ``fallacious,'' since the distribution of $\widehat{\Theta}_i$ ``deviates as far from the true distribution [of $\Theta_i$] in the direction of reduced spread as the observed distribution [of $Z_i$] does in the direction of increased spread.'' Despite Eddington's paper (which is well-known in astronomy), astronomers still often use shrunk values in the way Eddington warned against~\citep{loredo2007analyzing}.

The variance-constrained denoiser we study, traditionally called the constrained Bayes estimator, has been examined and used in applications by several authors~\citep{LouisI, lahiri1990adjusted, Ghosh, devine1994constrained, ghosh1999adjusted, FreyCressie, bloom2017using}. While some results like~\citet[Theorem 1]{Ghosh} apply generally to any Bayes estimator (paralleling our Theorem~\ref{thm:VCB-solution}), existing applications predominantly pair constrained Bayes denoisers with parametric EB methods. To the best of our knowledge, only few authors~\citep{shen2000triplegoal, lockwood2018flexible, lee2024improving} have explored more flexible nonparametric specifications alongside constrained Bayes denoising. Our work advances this literature in two ways: first, by deriving explicit estimation error rates that directly inherit from the original Bayes denoiser's error rates; second, by building on OT theory to provide a new perspective and to extend variance constrained denoisers to multidimensional parameters. Though~\cite{ghosh1999adjusted} previously tackled multidimensional extensions, we identify and correct an oversight in their approach (see Remark~\ref{rem:Ghosh-comparison}).


Our proposed methodology is tied to practical desiderata in a number of application areas. In this paper, we specifically consider applications to astronomy (Subsection~\ref{subsec:astro}), baseball (Subsection~\ref{subsec:baseball}), and marketing (Subsection~\ref{subsec:marketing}). 
Regarding astronomy, \citet{loredo2007analyzing} has argued that distributional constraints on denoising procedures are a fundamental feature of astrostatistics. More broadly, our methodology is directly relevant to domains that have recognized the need for constrained denoising, including multisite experiments in education research~\citep{raudenbush2015learning, bloom2017using, lee2024improving},
small-area estimation \citep{DevineLouisHalloran, GhoshRao, LeylandDavies}, epidemiology~\citep{lyles1997prediction, moore2010empirical}, and errors-in-variables regression \citep{freedman2004new}. In other fields, empirical Bayes methods are well established but constrained denoising remains less explored: our baseball application, while novel, complements a long tradition of empirical Bayes methods in sport analytics \citep{EfronMorris, brown2008inseason, jiang2010empirical, gu2017empirical}; empirical Bayes methods have been applied to A/B testing problems \citep{azevedo2019empirical, coey2019improving} closely related to our marketing application; and we anticipate further applications in areas such as those in labor economics reviewed by \citet{walters2024empirical}.

\section{Preliminaries on Optimal Transport}\label{sec:prelim}





We review some aspects of optimal transport (OT), which provides the mathematical foundation for our main results.
While OT was initially studied as a concrete problem in applied mathematics \citep{Villani, RachevRdorf}, there has been an explosion of recent interest in applications to statistics \citep{StatisticsWassersteinSpace, StatisticalOT}.

We begin with Monge's formulation of the OT problem, which aims to send each point to a given location in order to transport some source measure onto some target measure in such a way that minimizes the average cost.
(For simplicity, we consider the case of quadratic cost function and of probability measures on the same space, but these can both be relaxed.)
That is, for probability measures $\mu,\mu'$ on $\Rbb^k$, we consider the problem
\begin{equation}\label{eqn:Monge}
        \underset{T:\Rbb^k\to\Rbb^k}{\textnormal{minimize}} \int_{\Rbb^k}\|x-T(x)\|^2\diff \mu(x)\qquad
        \textnormal{s.t.}\qquad T_{\#}\mu = \mu'.
\end{equation}
Here $T_{\#}\mu$ denotes the \textit{pushforward} of $\mu$ by $T$, defined as the probability measure on $\Rbb^k$ given by the distribution of $T(X)$ when $X\sim \mu$.
Problem~\eqref{eqn:Monge} is difficult to solve directly, since the constraint $T_{\#}\mu = \mu'$ is non-convex.

Another perspective is Kantorovich's formulation, which aims to send each point to a distribution over locations in order to transport some source measure onto some target measure in a way that minimizes the average cost.
That is, for probability measures $\mu,\mu'$ on $\Rbb^k$, we consider
\begin{equation}\label{eqn:Kant}
    \underset{\pi\in\Gamma(\mu;\mu')}{\textnormal{minimize}} \int_{\Rbb^k}\|x-x'\|^2\diff \pi(x,x')
\end{equation}
Here, $\Gamma(\mu;\mu')$ denotes the set of all \textit{couplings} of $\mu,\mu'$, meaning the set of all joint distributions $\pi$ such that if $(X,X')\sim\pi$ then $X\sim \mu$ and $X'\sim \mu'$.
Problem~\eqref{eqn:Kant} is a linear program (possibly infinite dimensional) so it is often simple to analyze and implement.

One of the fundamental insights of OT is that problems \eqref{eqn:Monge} and \eqref{eqn:Kant} are closely related.
For instance, it is easy to show that  \eqref{eqn:Kant} is a \textit{convex relaxation} of \eqref{eqn:Monge} in the sense that it is a convex optimization problem whose domain naturally contains the domain of the former and its optimal value lower bounds the optimal value of the former. Remarkably, a result of \citet{Brenier1991} provides sufficient conditions for \eqref{eqn:Kant} to be a \textit{tight convex relaxation} of \eqref{eqn:Monge} in the sense that their optimal values agree, and every solution of the latter can be naturally mapped onto a solution of the former. More precisely, if $\mu$ and $\mu'$ both have finite second moment and if $\mu$ has a density with respect to Lebesgue measure, then \eqref{eqn:Monge} admits a solution $\delta$ which is unique up to $\mu$-equivalence, \eqref{eqn:Kant} admits a unique solution $\pi$, and these solutions are related via $\pi = (\id,\delta)_{\#}\mu$; moreover, the optimal $\delta$ for \eqref{eqn:Monge} must be equal to the gradient of some real-valued convex function defined on $\Rbb^m$.
In other words, the unique solution to the Kantorovich problem is a coupling supported on the graph of a gradient of a convex function which is a solution to the Monge problem.

In practice, a common approach to solving \eqref{eqn:Monge} is as follows.
First, find a solution $\pi$ to \eqref{eqn:Kant}.
Second, define the function $\delta:\Rbb^k\to\Rbb^k$ via
$    \delta(x')
    := \int_{\Rbb^k}x\, \pi(\diff x,x')\big/\int_{\Rbb^k}\pi(\diff x,x')
$
called the \textit{barycentric projection} of $\pi$; note that this is just the conditional expectation $\Ebb[X'\,|\,X=x]$ when $(X,X')\sim \pi$.
If the convex relaxation above is tight, then $\delta$ must be solution to \eqref{eqn:Monge}.
Otherwise, $\delta$  can be thought of as an approximate solution to \eqref{eqn:Monge}.

We next review some aspects of the geometry of the Wasserstein space.
We write $\Pcal_2(\Rbb^k)$ for the space of probability measures $\mu$ on $\Rbb^k$ satisfying $\int_{\Rbb^k}\|x\|^2\diff\mu(x)<\infty$. For any probability measures $\mu,\mu'\in\Pcal_2(\Rbb^k)$, we write $W_2^2(\mu,\mu')$ for the optimal value of problem~\eqref{eqn:Kant} above, which is called the (square of the) \textit{2-Wasserstein distance} between $\mu$ and $\mu'$.
(If $\mu$ additionally has a density with respect to Lebesgue measure, then we know that this is equal to the optimal value of problem~\eqref{eqn:Monge}.)
The metric space $(\Pcal_2(\Rbb^k),W_2)$ is called the \textit{2-Wasserstein space}.

These notions simplify significantly in the case of centered Gaussian measures, which will prove useful in our study of the variance-constrained denoising problem~\eqref{eqn:pop-varcstr}.
For $m\in\Nbb$, we write $\covspace(m)$ for the space of (strictly) positive definite $m\times m$ matrices.
For any $\Sigma,\Sigma'\in\covspace(m)$, we may consider problem~\eqref{eqn:Monge} for probability measures $\mu=\Ncal(0,\Sigma)$ and $\mu' = \Ncal(0,\Sigma')$; it turns out that this is an instance of a generalized Procrustes problem (see \cite{ZemelPanaretosProcrustes,Masarotto} for further detail) and that its (necessarily unique) solution given by the linear map $\delta(x) := \transport_{\Sigma}^{\Sigma'}x$, where $\transport_{\Sigma}^{\Sigma'} := \Sigma^{-\sfrac{1}{2}}(\Sigma^{\sfrac{1}{2}}\Sigma' \Sigma^{\sfrac{1}{2}})^{\sfrac{1}{2}}{\Sigma}^{-\sfrac{1}{2}}$.
Here, we write $A^{\sfrac{1}{2}}$ for the unique positive definite square root of a positive definite matrix $A\in\covspace(m)$, and we write $\preceq$ and $\prec$ for the usual Loewner order on $\covspace(m)$, that is, the order generated by the positive semi-definite (PSD) cone.
The space $\covspace(m)$ endowed with the metric inherited from $W_2$, under the natural identification, is referred to as the \textit{Bures-Wasserstein space}.

\section{Oracle Constrained Bayes Denoising}\label{sec:pop}

In this section, we consider the oracle versions (i.e., assuming that the latent variable distribution $G$ is known) of our constrained denoising problems~\eqref{eqn:pop-varcstr},~\eqref{eqn:pop-distrcstr} and~\eqref{eqn:pop-gencstr}. 
We address the case of variance constraints in Section~\ref{subsec:pop-varcstr}, the case of distributional constraints in Section~\ref{subsec:pop-distrcstr}, and the case of general constraints in Section~\ref{subsec:pop-gencstr}. The proofs of all the results stated in this section are given in Appendix~\ref{app:oracle}.

\subsection{Variance-Constrained Denoising}\label{subsec:pop-varcstr}

We begin with a detailed study of problem~\eqref{eqn:pop-varcstr} in which our results can be made rather explicit.
To give context for this problem, we mention the following motivating applications.

\begin{example}[Mitigating underdispersion by ``unshrinking'' a bit]
Graphical distributional summaries (e.g., histograms, density plots) of Bayes denoisers inevitably exhibit underdispersion relative to the true distribution of the unobserved latent variables $\Theta_1,\dotsc,\Theta_n$ (as directly follows from~\eqref{eqn:M-prec-A}). At an intuitive level, these graphical summaries would more faithfully represent the true heterogeneity in the population if the denoisers could be ``unshrunk'' to match the variance of the latent variables. Problem~\eqref{eqn:pop-varcstr} captures precisely this intuition through its variance-matching constraint, as Theorem~\ref{thm:VCB-solution} below demonstrates. Thus, problem~\eqref{eqn:pop-varcstr} provides a principled and useful strategy for presenting more faithful empirical summaries of heterogeneity by at least capturing the correct second moment structure. We note that this motivation was central to the work of \citet{LouisI} and \citet{Ghosh}.
\end{example}

\begin{example}[Errors-in-variables regression: moment reconstruction]
\label{example:moment_reconstruction}
Consider a regression setting with covariates $X \in \mathbb R^k$ and response $Y \in \mathbb R$; interest focuses on the relationship between $X$ and $Y$. We observe the covariates with noise, that is, we observe $W = X + E$ where $E$ has a known distribution. In regression calibration~\citep{carroll1990approximate} we perform the regression analysis on the pairs $(\mathbb E[X \mid W], Y)$; see~\citet{chen2025empirical} for a critical discussion of the approach. By contrast, in moment reconstruction~\citep{freedman2004new}, we seek a denoiser of $X$, which may possibly depend on both $W$ and $Y$, that satisfies
$$
\Ebb[\delta(W,Y)] = \Ebb[X],\qquad \Cov(\delta(W,Y)) = \Cov(X), \qquad \Ebb[\delta(W,Y) Y] = \Ebb[XY],
$$
and apply the regression analysis on pairs $(\delta(W,Y),Y)$.
The idea is that we would like our denoised covariates to preserve cross-moments with the response as these are fundamental for understanding the $X,Y$ relationship.
A moment reconstruction denoiser is provided by optimization problem~\eqref{eqn:pop-varcstr} by taking $Z = (W,Y)$ and $\Theta = (X,Y)$.
\end{example}

Our analysis of problem~\eqref{eqn:pop-varcstr} requires a few assumptions, which we now introduce and discuss.
Because we will impose covariance constraints on the denoising problem, we need some conditions to ensure that the constrained denoising problem admits a solution.
The following is necessary:
    \begin{equation}\label{eqn:2M}
        \int_{\Rbb^m}\|\theta\|^2\diff G(\theta) < \infty.
        \tag{2M}
    \end{equation}
We require the second moment condition for $G$ but not for $F$ (the marginal distribution of $Z$).
We also make the following assumption on the covariance of the Bayes denoiser $\delta_{\B}$:
\begin{equation}\label{eqn:pos-def}
            \Cov(\delta_\B(Z))\textnormal{ is strictly positive definite}. \tag{PD}
    \end{equation}
This condition is usually easy to verify; for instance, if $\Theta$ has Gaussian distribution $G=\Ncal(\mu,A)$ for $A\succ 0$ and $P_{\theta}=\Ncal(\theta,\Sigma)$, then we have $ \Cov(\delta_\B(Z)) = A(A+\Sigma)^{-1}A \succ 0$ for any $\Sigma\succeq 0$. In this form, the assumption precludes the case that $\Theta$ is concentrated on a proper subspace of $\Rbb^m$

Some additional assumptions provide sufficient regularity in order to ensure that the variance-constrained denoising problem admits a unique solution. Specifically, we assume
        \begin{equation}\label{eqn:Z-dens}
            \textnormal{the distribution of } Z \textnormal{ has a density with respect to Lebesgue measure, and}
                \tag{ZC}
        \end{equation}        \begin{equation}\label{eqn:Bayes-dens}
            \textnormal{the distribution of } \delta_\B(Z) \textnormal{ has a density with respect to Lebesgue measure.}
            \tag{BC}
        \end{equation}
        While condition~\eqref{eqn:Z-dens} holds whenever all elements of the model $\{P_{\theta}\}_{\theta}$ have a density with respect to Lebesgue measure (for example, in the Gaussian case), it does not hold in some discrete cases (for example, in the Poisson case).
        Under assumption~\eqref{eqn:Z-dens}, \citet[Remark~2.2]{GarciaTrillosSen} contains some sufficient conditions to establish assumption~\eqref{eqn:Bayes-dens}.
    We are ready to state the first of our main results, for which the univariate case is effectively due to \citet{LouisI, Ghosh}.
    Here, $R_{\B}$ denotes the minimum value of problem~\eqref{eqn:Bayes} and $R_{\varcstr}$ the minimum value of problem~\eqref{eqn:pop-varcstr}.
    Also, recall the definition of $\transport$ given in Section~\ref{sec:prelim}.
	\begin{theorem}\label{thm:VCB-solution}
        Under assumption~\eqref{eqn:2M} and \eqref{eqn:pos-def}, problem~\eqref{eqn:pop-varcstr} has solution 
	\begin{equation}\label{eq:varcstr}
			\delta_{\varcstr}(z) = \transport_{\Cov(\delta_\B(Z))}^{\Cov(\Theta)}(\delta_\B(z)-\Ebb[\Theta]) + \Ebb[\Theta].
		\end{equation}
	   and its risk is given by $R_{\varcstr} = R_{\B} + W_2^2(\Ncal(0,\Cov(\Theta)),\Ncal(0,\Cov(\delta_\B(Z))))$.
            Under the further assumptions~\eqref{eqn:Z-dens} and \eqref{eqn:Bayes-dens}, $\delta_{\varcstr}$ is the $F$-a.s. unique solution to problem~\eqref{eqn:pop-varcstr}.
	\end{theorem}
The interpretation of the result is simple: The variance-constrained Bayes denoiser $\delta_{\varcstr}$ is an affine function of the unconstrained Bayes denoiser $\delta_{\B}$, and the requisite rescaling is exactly the one that minimizes the OT cost from the covariance of the unconstrained Bayes denoiser $\delta_{\B}(Z)$ to the covariance of the latent variable $\Theta$.

        \begin{remark}[Comparison with \citet{ghosh1999adjusted}]\label{rem:Ghosh-comparison}
        The multivariate problem of optimal denoising subject to constraints on the covariance matrix has been previously studied by \citet{ghosh1999adjusted}. Their proposed denoiser is
        $\delta(z) = (\Cov(\Theta))^{\sfrac{1}{2}}(\Cov(\delta_{\B}(Z)))^{-\sfrac{1}{2}}(\delta_{\B}(z)-\Ebb[\Theta])+\Ebb[\Theta]$.
        Note that this differs from the optimal denoiser given in Theorem~\ref{thm:VCB-solution}, but both are of the form
        \begin{equation}\label{eqn:VCB-linear}
            \delta(z) = T(\delta_{\B}(z)-\Ebb[\Theta])+\Ebb[\Theta]
        \end{equation}
        for a suitable matrix $T\in\Rbb^{m\times m}$; that is, both are an affine transformation of the Bayes denoiser.
        We reconcile this difference by observing that an additional condition is needed in the discussion following \citet[equation~(40)]{ghosh1999adjusted}:
        \citet{ghosh1999adjusted} identify that an optimal denoiser must be of the form \eqref{eqn:VCB-linear}, but omit the condition that $T$ be symmetric.
        Since all subsequent applications (e.g., to the dataset in~\citet{DevineLouisHalloran}) involve matrices $\Cov(\Theta)$ and $\Cov(\delta_{\B}(Z))$ that commute, it holds that
        $
            \transport_{\Cov(\delta_{\B}(Z))}^{\Cov(\Theta)} = (\Cov(\Theta))^{\sfrac{1}{2}}(\Cov(\delta_{\B}(Z)))^{-\sfrac{1}{2}},
        $
        and so, in this case, the denoiser in~\citet{ghosh1999adjusted} is identical to the denoiser of Theorem~\ref{thm:VCB-solution}. Stating the result in the general case when $\Cov(\Theta)$ and $\Cov(\delta_{\B}(Z))$ do not commute requires the machinery of the Bures-Wasserstein geometry as summarized in Section~\ref{sec:prelim}.
        \end{remark}

    Next, we discuss the risk of $\delta_{\varcstr}$ compared to other baselines; to this end, define $R(\delta):=\Ebb[\|\delta(Z)-\Theta\|^2]$ for any denoiser $\delta$, so that $R_{\B} = R(\delta_{\B})$ and $R_{\varcstr} = R(\delta_{\varcstr})$.
    The first result in this direction is that the statistical cost of imposing the variance constraint is modest.

\begin{corollary}\label{cor:risks}
		Under assumption~\eqref{eqn:2M} and \eqref{eqn:pos-def}, we have $R_{\B}\le R_{\varcstr}\le 2R_{\B}$.
\end{corollary}

The second result is that $\delta_{\varcstr}$ has better risk than any unbiased baseline. 
Specifically, we consider denoisers $\delta_0:\Rbb\to\Rbb$ satisfying $\Ebb[\delta_0(Z)\,|\,\Theta]=\Theta$ almost surely, which coincides with the MLE $\delta_0(z)=z$ in Gaussian and Poisson models.

\begin{lemma}\label{lem:risk-comparison}
    Under assumption~\eqref{eqn:2M} and \eqref{eqn:pos-def}, if $m=d=1$ and $\delta_0:\Rbb\to\Rbb$ satisfies $\Ebb[\delta_0(Z)\,|\,\Theta]=\Theta$, then $R_{\varcstr} \le R(\delta_0)$.
\end{lemma}

The previous two results show that imposing the variance constraint increases the risk beyond the Bayes risk, but neither beyond twice the Bayes risk nor the risk of any unbiased baseline.

    \subsection{Distribution-Constrained Denoising}\label{subsec:pop-distrcstr}
Next we study the distribution-constrained Bayes problem~\eqref{eqn:pop-distrcstr}. This has recently been studied in \cite{GarciaTrillosSen}, and in this subsection we review some of their results and state them in a form that is amenable to our subsequent analyses. We keep in mind the following motivations throughout this subsection.
\begin{example}[Distortion-perception tradeoff]
    The problem of image restoration in computer vision involves a ground truth image $\Theta\in\Rbb^m$ and an observed image $Z\in\Rbb^d$, where the observation has much lower resolution compared to the ground truth or has been corrupted with noise.
    \citet{BlauMichaeli} refer to $\|\delta(Z)-\Theta\|^2$ as \textit{distortion} and $\rho(\delta_{\#}F,G)$ as  \textit{perception} for some metric or divergence $\rho$ on the space of probability measures, and they show that simultaneous minimization of distortion and perception is not possible.
    \citet{DistPercTrade} proposed to take $\rho = W_2$ and to try to minimize the distortion subject to a constraint on the perception, which corresponds to problem~\eqref{eqn:pop-distrcstr} when we impose the constraint of perfect perception.
\end{example}

\begin{example}[Errors-in-variables regression continued]
\label{example:errors_in_variables_distribution}
In Example~\ref{example:moment_reconstruction}, we introduced moment reconstruction. \citet{freedman2004new} propose moment reconstruction with the following motivation. Suppose we can come up with a denoiser $\delta(W,Y)$ of $X$ such that,
$
(\delta(W,Y),Y) \overset{\mathcal{D}}{=} (X,Y),
$
then we could learn about the relationship between $(X,Y)$ by studying $(\delta(W,Y),Y)$. Note that the constraint above corresponds to the distribution-constrained Bayes problem in~\eqref{eqn:pop-distrcstr}. 
\citet{freedman2004new} instead propose the moment reconstruction approach described in Example~\ref{example:moment_reconstruction}, arguing that enforcing the full distributional constraint ``is a very difficult problem, but if we content ourselves with the lesser aim of matching just the first two moments of the joint distribution, then a simple solution can be obtained.''
\end{example}

The main result in this subsection is Theorem~\ref{thm:pop-distrcstr}, which is a restatement of earlier work in \cite{GarciaTrillosSen}; it states that problem~\eqref{eqn:pop-distrcstr} admits a tight convex relaxation.
To understand this result, we recall the general argument from \cite{GarciaTrillosSen}: Let $F(\diff z) = \int_{\Rbb^m}P_{\theta}(\diff z)\diff G(\theta)$ denote the marginal distribution of $Z$, and for any denoiser $\delta$, make the following orthogonal decomposition
    \begin{equation}\label{eqn:bias-variance}
		\Ebb\left[\|\delta(Z)-\Theta\|^2\right] =\Ebb\left[\|\delta(Z)-\delta_{\B}(Z)\|^2\right]+ \Ebb\left[\|\delta_{\B}(Z)-\Theta\|^2\right],
	\end{equation}
    and observe that the second term is the irreducible Bayes risk $R_{\mathcal{B}}$ and the first term depends on $\delta$.
    As such, problem~\eqref{eqn:pop-distrcstr} is naturally related to an optimal transport problem with the following (non-standard) cost function $c_{G}: \Rbb^d \times \Rbb^m \to [0,\infty)$ 
    \begin{equation}\label{eqn:costG}
        c_{G}(z,\eta):=\|\eta - \delta_{\B}(z)\|^2,
    \end{equation}
    which, as we emphasize in the notation, depends on the unknown distribution $G$ (and also the  known likelihood $P=\{P_{\theta}\}_{\theta}$) via the conditional expectation $\delta_{\B}(z) = \Ebb[\Theta\,|\,Z=z]$.)
    We then have the following, which is a re-statement of \citet[Theorem~2.4]{GarciaTrillosSen}, and bears some similarity to \citet[Theorem~1]{DistPercTrade}.
    Recall from Section~\ref{sec:prelim} that $\Gamma(F;G)$ represents the space of all couplings of $F$ and $G$, and let us write $R_{\distrcstr}$ for the minimum value of problem~\eqref{eqn:pop-distrcstr}.
    
    \begin{theorem}[Theorem~2.4 of \cite{GarciaTrillosSen}]\label{thm:pop-distrcstr}
        Under assumptions \eqref{eqn:2M}, \eqref{eqn:Z-dens}, and \eqref{eqn:Bayes-dens}, the problem~
        \begin{equation}\label{eqn:pop-distrcstr-Monge}
    	\underset{\pi\in\Gamma(F;G)}{\textnormal{minimize}} \int_{\Rbb^d\times\Rbb^m}c_G(z,\eta)\, \diff \pi(z,\eta),
        \end{equation}
    has a unique solution $\pi_{\distrcstr}$. This solution is concentrated on the graph of a function which we denote by $\delta_{\distrcstr}:\Rbb^d\to\Rbb^m$. This function $\delta_{\distrcstr}$ can be written as $\delta_{\distrcstr} = \nabla \phi \circ \delta_{\B}$ for some convex function $\phi:\Rbb^m\to\Rbb \cup \{+\infty\}$, and we have $R_{\distrcstr} = R_{\B} + W_2^2(G,(\delta_{\B})_{\#}F)$. Consequently, \eqref{eqn:pop-distrcstr-Monge} is a tight convex relaxation of~\eqref{eqn:pop-distrcstr}, and $\delta_{\distrcstr}$ is the $F$-a.s. unique solution to \eqref{eqn:pop-distrcstr}.
    \end{theorem}
The above result says that, under appropriate conditions, problem~\eqref{eqn:pop-distrcstr} indeed has a unique solution which can be obtained by solving the linear program~\eqref{eqn:pop-distrcstr-Monge}.
Further, the risk of the distribution-constrained denoiser $\delta_{\distrcstr}(Z)$ is the sum of the Bayes risk $R_{\B}$ and the squared Wasserstein distance between the distributions of $\Theta$ and $\delta_{\B}(Z)$. Contrast this with Theorem~\ref{thm:VCB-solution}, where we showed that the risk of the variance-constrained denoiser $\delta_{\varcstr}(Z)$ is the sum of the Bayes risk and the squared Bures-Wasserstein distance between the covariance matrices of $\Theta$ and $\delta_\B(Z)$.

\begin{remark}[Wasserstein v. Bures-Wasserstein distance]
    Since the constraint of~\eqref{eqn:pop-distrcstr} is stronger than that of~\eqref{eqn:pop-varcstr}, the excess risk in Theorem~\ref{thm:pop-distrcstr} must be larger than the excess risk in Theorem~\ref{thm:VCB-solution}.
    This can also be seen directly from the following well-known comparison between the Wasserstein and Bures-Wasserstein distances (see \cite[Proposition~1.6.5]{StatisticsWassersteinSpace}), where random variables $X,X'$ have laws $\mu,\mu'$ respectively:
    \begin{equation*}
        W_2^2(\mu,\mu') \ge \left\|\Ebb[X]-\Ebb[X']\right\|^2 + \left|\sqrt{\trace(\Cov(X))}-\sqrt{\trace(\Cov(X'))}\right|^2.
    \end{equation*}
    This partially explains why matching the first and second moments of $\delta(Z)$ to those of $\Theta$ often reduces the deconvolution error: both underdispersion (e.g., $\delta=\delta_{\B}$) and overdispersion (e.g., $\delta(z)=z$) provide lower bounds on the distance between the distribution of $\delta(Z)$ and $\Theta$.
\end{remark}

\begin{remark}[Targets other than $G$]\label{rem:other-targets}
    In some settings, it may be desirable to replace the constraints $\delta(Z)\overset{\mathcal{D}}{=}G$ in \eqref{eqn:pop-distrcstr} and $\pi\in\Gamma(F;G)$ in~\eqref{eqn:pop-distrcstr-Monge} with $\delta(Z)\overset{\mathcal{D}}{=}\tilde{G}$ in \eqref{eqn:pop-distrcstr} and $\pi\in\Gamma(F;\tilde{G})$, respectively, for some other distribution $\tilde{G}\neq G$, e.g., a low-order approximation of $G$~\citep{eddington1940correction,carroll2004low}. As explained in \citet[Remark~4.2]{GarciaTrillosSen}, the previous result still holds in this setting, provided that $\tilde{G}$ satisfies the same regularity assumption we have placed on $G$ above.
\end{remark}

An alternative to targeting some particular choice of $\tilde{G}\neq G$ is to specify some features of $G$ of interest, and then to target whichever $\tilde{G}$ matches these features and leads to the smallest risk.
This is taken up in the next subsection.

        \subsection{General-Constrained Denoising}\label{subsec:pop-gencstr}

        Lastly, we consider problem~\eqref{eqn:pop-gencstr} of general-constrained Bayes denoising.
        Recall that we fix a linearly-independent collection $\psi_1,\ldots,\psi_k:\Rbb^m\to\Rbb$ of continuous functions (which are also assumed to be $G$-integrable); for convenience, we write $\Psi:=\{\psi_1,\ldots, \psi_k\}$.
        Roughly speaking, many constraints on the distribution of $\delta(Z)$ can be cast in the form of \eqref{eqn:pop-gencstr}; indeed, \eqref{eqn:pop-gencstr} reduces to \eqref{eqn:Bayes} if $\Psi$ is empty, \eqref{eqn:pop-varcstr} if $\Psi$ consists of all quadratic forms, and \eqref{eqn:pop-distrcstr} if $\Psi$ consists of all bounded continuous functions.
        In addition to subsuming these earlier examples, we are motivated by the following.


        \begin{example}[Higher moment constraints]\label{ex:moment-cstr}
            In the univariate case $m=d=1$ and for some fixed positive integer $k$, one may aim to minimize the denoising risk while matching the first $k$ moments of $\delta(Z)$ to those of $\Theta$, thereby solving the problem
            \begin{equation}\label{eqn:gen-moments}
                \underset{\delta:\Rbb\to\Rbb}{\textnormal{minimize}}\quad\Ebb\left[|\delta(Z)-\Theta|^2\right]\qquad\textnormal{s.t.}\qquad \Ebb[(\delta(Z))^{\ell}] = \Ebb[\Theta^{\ell}] \quad \textnormal{ for all } 1\le \ell \le k.
            \end{equation}
        This, of course, coincides with the case of variance constraints when $k=2$, and higher moment constraints force the distribution of $\delta(Z)$ to be even closer to the distribution $G$ of $\Theta$.
        Along these lines it is shown in \cite{armstrong2022robust}, in the context of constructing nonparametric EB confidence intervals, that imposing constraints on the fourth moment leads to EB procedures which nearly adapt to the parametric (Gaussian) setting.
        \end{example}

\begin{example}[Errors-in-variables regression continued: moment adjusted imputation]
\label{example:errors_in_variables_mai}
\citet{bay1997adjusting} and \citet{thomas2011momentadjusteda} propose an intermediate approach called moment adjusted imputation that bridges the methods described in Examples~\ref{example:moment_reconstruction} and~\ref{example:errors_in_variables_distribution}. Their approach uses a denoiser that matches a controlled number of moments and cross-moments between $X$ and $Y$, extending beyond the two moments in moment reconstruction but without enforcing full distributional equivalence. For instance, when $X \in \mathbb{R}$, they aim to estimate $X$ while preserving multiple moments such as $\mathbb{E}[X]$, $\mathbb{E}[X^2]$, $\mathbb{E}[X^3]$, $\mathbb{E}[X^4]$, as well as cross-moments like $\mathbb{E}[XY]$ and $\mathbb{E}[X^2Y]$. Such a denoiser is given by problem~\eqref{eqn:pop-gencstr} with the choices $Z = (W,Y)$, $\Theta = (X,Y)$, $\psi_\ell(x,y) = x^{\ell}$ for $\ell=1,\dotsc,4$, and $\psi_{5}(x,y)=xy$, $\psi_{6}(x,y)=x^2y$. 
\end{example}

        \begin{example}[Variance constraint and support constraint]\label{ex:var-supp-cstr}
            Suppose that $G$ is known to have support contained in a closed set $S\subseteq\Rbb^m$. In that case, the posterior mean $\delta_{\B}(Z)$ 
            lies in the convex hull of $S$ almost surely. However, the variance constrained denoiser $\delta_{\varcstr}(\cdot)$ may take on values outside $S$ with positive probability, which can be undesirable.
            Instead,
            one may aim to minimize the denoising risk subject to simultaneous constraints on the variance and the support of $\delta(Z)$, thereby solving
            \begin{equation}    		\underset{\delta:\Rbb^{d}\to\Rbb^{m}}{\textnormal{minimize}} \,\,\Ebb\left[\|\delta(Z)-\Theta\|^2\right]\quad\textnormal{s.t.}\quad\Ebb[\delta(Z)] = \Ebb[\Theta],\,\Cov(\delta(Z)) = \Cov(\Theta),\,\, \Pbb(\delta(Z)\in S) = 1.
        \end{equation}
        To see that this can be cast as a special case of problem \eqref{eqn:pop-gencstr}, take $\psi(\eta) :=\min_{\theta\in S}\|\eta-\theta\|$ and note that $\Pbb(\delta(Z)\in S)=1$ is equivalent to $\Ebb[\psi(\delta(Z))] = \Ebb[\psi(\Theta)] = 0$
        In the case $m=d=1$ an important constraint on the support is nonnegativity, i.e., $S = [0,\infty)$.
        \end{example}

        We need some further assumptions and notation in order to state our main result on problem~\eqref{eqn:pop-gencstr}.
        First, let us define $\Gamma(F;G,\Psi)$ to be the set of all joint distributions $\pi$ on $\mathbb{R}^d \times \mathbb{R}^m$ such that if $(Z,D)\sim \pi$ then $Z\sim F$ and $\Ebb[\psi_{\ell}(D)]=\Ebb[\psi_{\ell}(\Theta)]$ for all $1\le \ell\le k$.

    Furthermore, we assume that
    \begin{equation}\label{eqn:cty}
            \delta_{\B}:\Rbb^d\to\Rbb^m\textnormal{ is continuous on a set of full } F\textnormal{-measure}.\tag{C}
    \end{equation}
    We note that \eqref{eqn:cty} is a mild condition, and that it is satisfied whenever $\{P_{\theta}\}_{\theta}$ is a suitable exponential family, regardless of $G$ \cite[Remark~2.2]{GarciaTrillosSen}; for example, it holds for Gaussian or Poisson families.
    (It holds trivially whenever $F$ has discrete support.)
    Next, we assume that 
    \begin{equation}\label{eqn:quadratic-growth}
        \Psi\textnormal{ contains a non-negative function with at least quadratic growth.}
        \tag{QG}
    \end{equation}
    Here, a function $\psi:\Rbb^m\to\Rbb$ is said to have \textit{at least quadratic growth} if there exist constants $C,R>0$ such that we have $\|\eta\|^2\le C\psi(\eta)$ for all $\eta\in\Rbb^m$ with $\|\eta\|>R$.
    Lastly, we need the technical condition
    \begin{equation}\label{eqn:UI}
        \Psi \textnormal{ is uniformly integrable under } \Gamma(F;G,\Psi)\tag{UI}
    \end{equation}
    where we apply $\Psi$ to the $D$ coordinate of all of the joint distributions $(Z,D)\sim \pi$ in $\Gamma(F;G,\Psi)$; in other words, condition~\eqref{eqn:UI} states $\lim_{R\to\infty}\sup_{1\le \ell \le k}\sup_{\pi\in \Gamma(F;G,\Psi)}\int_{\{\|\eta\|> R\}}|\psi_{\ell}(\eta)|\diff \pi(z,\eta) = 0$.
    Note that condition~\eqref{eqn:UI} holds trivially whenever the constraints force every feasible $\pi$ to be supported on a fixed compact set, on which each $\psi_{\ell}$ is hence bounded; this is the case in Example~\ref{ex:var-supp-cstr} when $S$ is bounded, and is implicit in the discretization described in Appendix~\ref{app:comp}.
        
    Finally, we get the following.
    We write $\delta_{\B,H}$ for the Bayes denoiser when $\Theta$ comes from distribution $H\in\mathcal{P}(\Rbb^m)$, and we recall that the cost function $c_G(\cdot,\cdot)$ is defined in~\eqref{eqn:costG}.

    \begin{theorem}\label{thm:gen-cstr-num}
        Under assumptions \eqref{eqn:Z-dens}, \eqref{eqn:Bayes-dens}, \eqref{eqn:cty}, \eqref{eqn:quadratic-growth}, and~\eqref{eqn:UI}, the problem \begin{equation}\label{eqn:gen-str-Kant}
    	\underset{\pi\in\Gamma(F;G,\Psi)}{\textnormal{minimize}} \int_{\Rbb^d\times\Rbb^m}c_G(z,\eta)\, \diff \pi(z,\eta),
    \end{equation}
    has a unique solution $\pi_{\gencstr}$, and this solution is concentrated on the graph of a function which we denote $\delta_{\gencstr}:\Rbb^d\to\Rbb^m$.
    If $H$ denotes the marginal distribution of $D$ when $(Z,D)\sim \pi_{\gencstr}$, then $\delta_{\gencstr}$ can be written as $\delta_{\gencstr} = \nabla \phi \circ \delta_{\B,H}$ for some convex function $\phi:\Rbb^m\to\Rbb \cup \{+\infty\}$, and we have $R_{\gencstr} = R_{\B} + W_2^2(H,(\delta_{\B})_{\#}F)$.
    Consequently, \eqref{eqn:gen-str-Kant} is a tight convex relaxation of~\eqref{eqn:pop-gencstr}, and $\delta_{\gencstr}$ is the $F$-a.s. unique solution to \eqref{eqn:pop-gencstr}.
    \end{theorem}

    As in the case of problem~\eqref{eqn:pop-distrcstr}, this result provides conditions under which problem~\eqref{eqn:pop-gencstr} has a unique solution which can be found by solving the linear program~\eqref{eqn:gen-str-Kant}.
    We refer to the probability measure $H$ as the \textit{projected latent variable distribution} since, among all probability measures whose integrals with $\Psi$ match those of $G$, it is the nearest to $(\delta_{\B})_{\#}F$ with respect to the distance $W_2$. 
    The benefit of the approach of problem~\eqref{eqn:pop-gencstr}, compared to the approach of using~\eqref{eqn:pop-distrcstr} with some particular $\tilde{G}\neq G$ as in Remark~\ref{rem:other-targets}, is that the projected latent variable distribution $H$ is exactly the $\tilde{G}$ which matches the desired constraints and which leads to the smallest risk.

    \begin{example}[Higher moment constraints continued]\label{ex:moment-cstr-continued}
        To provide further intuition on the projected latent variable distribution $H$, we characterize it in some particular cases of Example~\ref{ex:moment-cstr}.
        If $k = 0$ or $k=1$, then $\delta_{\B}$ is feasible hence optimal for \eqref{eqn:gen-moments}, so $H$ equals $(\delta_{\B})_{\#}F$.
        If $k=2$, then \eqref{eqn:gen-moments} is equivalent to \eqref{eqn:pop-varcstr}, so $H$ equals the pushforward of $F$ by an affine function of $\delta_{\B}$.
        If $G$ is compactly supported and $H_k$ denotes the projected latent variable distribution for constraints on $k$ moments, then $H_k$ converges weakly to $G$ as $k\to\infty$.
    \end{example}

	\section{Empirical Constrained Bayes Denoising}\label{sec:emp}

    In this section, we consider constrained denoising at the empirical level, i.e., we only observe $Z_1,\ldots, Z_n$ from the independent, identically-distributed (i.i.d.) sequence of pairs $(\Theta_1,Z_1),\ldots, (\Theta_n,Z_n)$ coming from model~\eqref{eq:Mdl}.
    Importantly, we assume throughout this section that the distribution $G$ of $\Theta_1,\ldots, \Theta_n$ is unknown; it must therefore be estimated (either implicitly or explicitly) in order to apply the results of Section~\ref{sec:pop} about the corresponding oracle constrained Bayes denoising problems.
    
    Mathematically, our goal in this section is as follows.
    Write $\mathbf{Z}=(Z_1,\ldots, Z_n)$ and $\mathbf{\Theta}=(\Theta_1,\ldots, \Theta_n)$, and write $\delta(\mathbf{Z})$ for the coordinatewise application of a denoiser $\delta$ to $\mathbf{Z}$.
    For a fixed oracle denoiser $ \delta_{\ast}\in\{\delta_{\B},\delta_{\varcstr},\delta_{\distrcstr},\delta_{\gencstr}\}$, we aim to develop an EB approximation $\hat{\delta}_{\ast}$ of $\delta_{\ast}$ and establish a rate of convergence for the quantity
    \begin{equation*}
        \|\hat{\delta}_{\ast}(\mathbf{Z})-\delta_{\ast}(\mathbf{Z})\|^2_n:=\frac{1}{n}\sum_{i=1}^{n}\left\|\hat{\delta}_{\ast}(Z_i)-\delta_{\ast}(Z_i)\right\|^2.
    \end{equation*}
    The quantity above measures the estimation accuracy of the empirical Bayes denoiser, and is closely related to the excess risk; if $\hat{\delta}_{\ast}$ were a fixed denoiser or independent of $Z_1,\ldots, Z_n$, then its expectation would equal the excess risk relative to $\delta_{\ast}$.
    While there are many ways to quantify the accuracy of an EB denoiser (see \citet[Remark~3.1]{SahaGuntuboyina}), this notion is natural in our setting because it allows us to bound both the excess denoising error and the excess deconvolution error as follows:
    \begin{align*}
        &\|\hat{\delta}_{\ast}(\mathbf{Z})-\mathbf{\Theta}\|_n - \|\delta_{\ast}(\mathbf{Z})-\mathbf{\Theta}\|_n \le \|\hat{\delta}_{\ast}(\mathbf{Z})-\delta_{\ast}(\mathbf{Z})\|_n, \\
        &W_2\left((\hat{\delta}_{\ast})_{\#}\bar{F}_n,G\right) - W_2\left((\delta_{\ast})_{\#}\bar{F}_n,G\right) \le W_2\left((\hat{\delta}_{\ast})_{\#}\bar{F}_n,(\delta_{\ast})_{\#}\bar{F}_n\right) \le \|\hat{\delta}_{\ast}(\mathbf{Z})-\delta_{\ast}(\mathbf{Z})\|_n
    \end{align*}
    where $\bar{F}_n$ is the empirical measure of $Z_1,\ldots, Z_n$.
    We emphasize that we follow the standard paradigm of EB whereby we do not aim to make claims about the oracle errors $\Ebb[\frac{1}{n}\sum_{i=1}^{n}\|\delta_{\ast}(Z_i)-\Theta_i\|^2]$ and $\Ebb[W_2^2((\delta_{\ast})_{\#}\bar{F}_n,G)]$, but rather about the excess errors of some EB approximation of the oracle.
    As we showed in Corollary~\ref{cor:risks} and Lemma~\ref{lem:risk-comparison}, and as we will explore numerically in Section~\ref{sec:sim}, the denoising error of the oracle constrained denoisers are not typically much larger than those of the oracle unconstrained denoisers.

    All of our estimators arise by transforming an existing (unconstrained) EB denoiser in order to make it target a suitable constrained oracle denoiser.
    More precisely, if we assume a priori that $\hat{\delta}_{\B}$ satisfies
    $\|\hat{\delta}_{\B}(\mathbf{Z})-\delta_{\B}(\mathbf{Z})\|^2_{n}\to 0$ in probability, then we can construct $\hat{\delta}_{\ast}$ from $\hat{\delta}_{\B}$ in such a way that we have
    $\|\hat{\delta}_{\ast}(\mathbf{Z})-\delta_{\ast}(\mathbf{Z})\|^2_{n}\to 0$ in probability; moreover the rate of convergence of $\hat{\delta}_{\ast}$ to $\delta_{\ast}$ depends explicitly on the rate of convergence of $\hat{\delta}_{\B}$ to $\delta_{\B}$.
    This flexibility in choosing $\hat{\delta}_{\B}$ is important for the broad applicability of the methodology; the reader can keep in mind the following examples for how EB denoisers may arise.

    \begin{example}[$G$-modeling]\label{ex:g-model}
    In many settings, one defines a deconvolution estimate $\hat G_n$ of $G$ by solving the following nonparametric maximum likelihood estimation (NPMLE) problem
    \begin{equation}\label{eqn:NPMLE}
        \hat G_n \; \in \;\underset{H\in \Pcal(\Rbb^m)}{\arg\max}\,\frac{1}{n}\sum_{i=1}^{n}\log  \left(\int_{\Rbb^m}p_{\theta}(Z_i)\diff H(\theta) \right),
    \end{equation}
    where $p_{\theta}$ is the density of $P_{\theta}$ with respect to some fixed reference measure (usually the Lebesgue measure or the counting measure), and then one considers the EB approximation $\hat{\delta}_{\B}:=\delta_{\B,\hat{G}_n}$.
    As such, $G$-modeling supplies both the desired EB denoiser and some desired information about the latent variable distribution (e.g., its moments or the entirety of $G$) at the same time.
    However, one of the fundamental insights of EB is that fast (i.e., nearly-parametric) rates of convergence of $\delta_{\B,\hat{G}_n}$ can occur despite the (typically) slow rate of convergence of $\hat{G}_n$;
    this phenomenon is well-understood in settings where $\{P_{\theta}\}_{\theta}$ is a Gaussian location family studied by \citet{JiangZhang,SahaGuntuboyina, Soloff} or a Poisson family studied by \citet{ShenWu, PolyanskiyWuPoisson}.
\end{example}

    \begin{example}[Smooth $G$-modeling]\label{ex:smooth-f-model}
        If it is known that the prior $G$ is itself a mixture of Gaussians whose components have variance satisfying some lower bound (in the positive semi-definite order), then one can restrict the optimization in  \eqref{eqn:NPMLE} accordingly, leading to the so-called \textit{smooth NPMLE} introduced by \citet{magder1996smooth} and recently studied by \citet{KimSen}.
        When such restrictions are well-specified (i.e., when the user-specified component variance is at least as small as the component variance in $G$), the rate of convergence is identical to that of the unrestricted NPMLE in Example~\ref{ex:g-model}.
        In some settings, the increased regularity of smooth NPMLE over generic NPMLE is desirable.
    \end{example}

    In addition to $G$-modeling and smooth $G$-modeling, we note that other EB approaches are possible, including $F$-modeling and conjugate parametric models.
    In Appendix~\ref{sec:additional-EB} we describe these in more detail, including some explicit formulas in the case of conjugate parametric models.
    Our results in this section are divided into variance constraints (Subsection~\ref{subsec:emp-varcstr}), distributional constraints (Subsection~\ref{subsec:emp-distrcstr}), and general constraints (Subsection~\ref{subsec:emp-gencstr}).
    The proofs of all results in this section can be found in Appendix~\ref{app:emp}.

    \subsection{Variance-Constrained Denoising}\label{subsec:emp-varcstr}

    We begin with problem~\eqref{eqn:pop-varcstr} in which our results are the most complete.
    That is, we consider finding, on the basis of the observations $Z_1,\ldots, Z_n$ and the known likelihood $\{P_{\theta}\}_{\theta}$, an estimator $\hat{\delta}_{\varcstr}$ of the solution $\delta_{\varcstr}$  to problem~\eqref{eqn:pop-varcstr}, defined in~\eqref{eq:varcstr}.
    Note that $\delta_{\varcstr}$ need not be the unique solution to \eqref{eqn:pop-varcstr}, but Theorem~\ref{thm:VCB-solution} guarantees that it always exists and hence it is an oracle denoiser that we may target.

    From the form of $\delta_{\varcstr}$ in~\eqref{eq:varcstr}, it is clear that, in addition to the existence of an EB denoiser $\hat{\delta}_{\B}$, we also need to estimate the matrices $A:=\Cov(\Theta)$ and $M:=\Cov(\delta_\B(Z))$.
    The latter can be estimated via the empirical covariance matrix of $\hat{\delta}_{\B}(Z_1),\ldots, \hat{\delta}_{\B}(Z_n)$, and the former can be estimated various ways, for example, by a suitable method of moments estimator.
    More precisely, we propose a constrained denoising procedure in Algorithm~\ref{alg:Gaussian-mod-NP-prior} and our main result below will provide an explicit rate of convergence for the resulting constrained denoiser.
    
    Let us describe the assumptions that will be used in our main result.    
    First, we assume that we are equipped with some EB approximation $\hat{\delta}_{\B}$ of $\delta_{\B}$ and some sequence $\alpha_n\to 0$ satisfying       \begin{equation}\label{eqn:EBQ}
        \frac{1}{n}\sum_{i=1}^{n}\left\|\hat{\delta}_{\B}(Z_i)-\delta_{\B}(Z_i)\right\|^2 = O_{\Pbb}(\alpha_n),
        \tag{EBQ}
    \end{equation}
    which will also be assumed in later subsections;
    as we discuss in detail in Appendix~\ref{app:rates}, we may determine $\alpha_n$ in many models of interest.
    Next, we need some conditions in order to estimate the requisite covariance matrices.
    More precisely, we require the following integrability of $G$, which is similar to but slightly stronger than what we assumed in \eqref{eqn:2M} for the oracle problem:
    \begin{equation}\label{eqn:FM}
        \int_{\Rbb^m}\|\theta\|^{4}\diff G(\theta)<\infty.
        \tag{4M}
    \end{equation}
    Additionally, we require the existence of estimators $\hat{\mu}_n$ of $\mu:=\Ebb[\Theta]$ and $\hat{A}_n$ of $A:=\Cov(\Theta)$ satisfying
    \begin{equation}\label{eqn:VE}
        \|\hat{\mu}_n-\mu\|^2 = O_{\Pbb}(n^{-1})\qquad\textnormal{ and }\qquad\|\hat A_n-A\|^2 = O_{\Pbb}\left(n^{-1}\right).
        \tag{ME}
    \end{equation}
    where $\|\cdot\|$ denotes any matrix norm (since they are all equivalent up to constants).

    \begin{example}[method of moments]\label{ex:MOM}
        Suppose the likelihood $\{P_{\theta}\}_{\theta}$ satisfies $\Ebb[Z\,|\,\Theta] = \Theta$ almost surely; this assumption forces $d=m$ and is satisfied for Gaussian and Poisson cases
    (but note that it forces us to use ``scale parameterizations'' rather than ``rate parameterizations'' in, e.g., geometric and exponential models).
    Then, we easily verify condition~\eqref{eqn:VE} by constructing the following method-of-moments estimators.
    The sample mean of $Z_1,\ldots,Z_n$ is a suitable estimator of  $\mu=\Ebb[\Theta]=\Ebb[Z]$ by~\eqref{eqn:FM} (or even~\eqref{eqn:2M}).
    Further, suppose that there exists an estimator $\hat\Sigma_n$ of $\Sigma:=\Ebb[\Cov(Z\,|\,\Theta)]$ satisfying $\|\hat \Sigma_n-\Sigma\|^2 = O_{\Pbb}\left(n^{-1}\right)$; this condition is satisfied in the Gaussian case since $\Sigma$ is exactly known and can be estimated via $\hat \Sigma_n := \Sigma$, and it is satisfied in the Poisson case since $\Sigma = \Ebb[Z]$ can be estimated via $\hat{\mu}_n$.
    Then the law of total covariance yields
    \begin{equation}\label{eqn:law-total-cov}
        \Cov(Z) = \Ebb[\Cov(Z\,|\,\Theta)] + \Cov(\Ebb[Z\,|\,\Theta]) = \Sigma + \Cov(\Theta),
    \end{equation}
    so $\hat{A}_n:=(\hat S_n-\hat\Sigma_n)_{+}$ is a suitable estimator of $A$, where $(\,\cdot\,)_{+}$ denotes positive semi-definite truncation and $\hat S_n$ denotes the empirical covariance matrix of $Z_1,\ldots, Z_n$.
    \end{example}

    In our simulations and applications we focus on the method of moments estimators above, but it may sometimes be desirable to verify condition~\eqref{eqn:VE} with other estimators (e.g., a plug-in approach by computing the mean and covariance of a deconvolution estimator $\hat{G}_n$, or robust mean and covariance estimators studied by \citet{Minsker,RobustCovariance}.)

    \begin{algorithm}[t]
	\caption{An EB approximation $\hat \delta_{\varcstr}$ of the oracle variance-constrained Bayes denoiser $\delta_{\varcstr}$. 
    }\label{alg:Gaussian-mod-NP-prior}
	
	\begin{algorithmic}[1]
            
		  \State \textbf{input:} samples $Z_1,\ldots, Z_n\in\Rbb^m$, likelihood $\{P_{\theta}\}_{\theta}$
		\State \textbf{output:} denoising function $\hat{\delta}_{\varcstr}:\Rbb^m\to\Rbb^m$
            \State $\hat{\delta}_{\B}(\,\cdot\,) \leftarrow$ EB approximation of $\delta_{\B}(\,\cdot\,)$
            \State $\hat M \leftarrow$ sample covariance matrix of $\hat{\delta}_{\B}(Z_1),\ldots \hat{\delta}_{\B}(Z_n)$
            \State $\hat \mu \leftarrow$ estimate of $\mu$
            \State $\hat A \leftarrow$ estimate of $A$
            \State $\hat{\transport}\leftarrow \transport_{\hat{M}}^{\hat{A}}=\hat M^{-\sfrac{1}{2}}(\hat M^{\sfrac{1}{2}} \hat A \hat M^{\sfrac{1}{2}})^{\sfrac{1}{2}}\hat M^{-\sfrac{1}{2}}$
            \State $\hat{\delta}_{\varcstr}(\,\cdot\,) \leftarrow \hat{\transport}(\hat{\delta}_{\B}(\,\cdot\,)-\hat \mu) + \hat \mu$

            \State \textbf{return} $\hat{\delta}_{\varcstr}$
	\end{algorithmic}
    \end{algorithm}

    This leads us to the main result, whose proof appears in Appendix~\ref{app:emp-varcstr}:

    \begin{theorem}\label{thm:varcstr-gm-npp}
        Under \eqref{eqn:FM}, \eqref{eqn:VE}, and \eqref{eqn:EBQ}, the denoiser $\hat \delta_{\varcstr}$ from Algorithm~\ref{alg:Gaussian-mod-NP-prior} satisfies
        \begin{equation*}
            \frac{1}{n}\sum_{i=1}^{n}\left\|\hat \delta_{\varcstr}(Z_i)-\delta_{\varcstr}(Z_i)\right\|_2^2 = O_{\Pbb}\left(\alpha_n\vee n^{-1}\right)\qquad \textnormal{as}\qquad n\to\infty.
        \end{equation*}
    \end{theorem}

    The conclusion of Theorem~\ref{thm:varcstr-gm-npp} for the denoiser given in Algorithm~\ref{alg:Gaussian-mod-NP-prior} is that the rate of convergence of variance-constrained Bayes denoising is determined by whichever of unconstrained Bayes denoising (line 3) and moment estimation (lines 4 to 6) is slower.
    Moreover, Appendix \ref{app:rates} contains a detailed discussion of the rates of convergence in the Gaussian and Poisson models arising from this result.
    In the case of a Gaussian likelihood and a subexponential latent variable distribution, we note that the resulting rate is nearly parametric, i.e. $n^{-1}$ up to logarithmic factors.
    In the case of an exponential family likelihood and a conjugate prior, the resulting rate is exactly parametric.

    \subsection{Distribution-Constrained Denoising}\label{subsec:emp-distrcstr}

    Next we consider the case of distributional constraints, as in problem~\eqref{eqn:pop-distrcstr}.
    Recall that in Subsection~\ref{subsec:pop-distrcstr} we provided sufficient conditions for the existence of a unique solution $\delta_{\distrcstr}$ to problem~\eqref{eqn:pop-distrcstr}; our present goal is to develop an EB methodology which targets this denoiser.

    From the characterization of $\delta_{\distrcstr}$ as the solution to the OT problem~\eqref{eqn:pop-distrcstr-Monge}, it is clear that, in addition to the existence of an EB denoiser $\hat{\delta}_{\B}$, we also need to be able to estimate the distribution $F$, the distribution $G$, and the cost function $c_G$. (Recall its definition in~\eqref{eqn:costG}.) 
    We may estimate $F$ via the empirical distribution of $Z_1,\ldots, Z_n$,  $G$ as $\hat{G}_n$ via deconvolution (as in, e.g., Examples~\ref{ex:g-model} and~\ref{ex:smooth-f-model}), and $c_G$ as a function of the given EB denoiser $\hat{\delta}_{\B}$; then we use these approximations in an OT problem paralleling~\eqref{eqn:pop-distrcstr-Monge}, and the barycentric projection of its optimal coupling is an approximation of $\delta_{\distrcstr}$.
    (Recall from Section~\ref{sec:prelim} that the barycentric projection recovers the OT map if the optimal coupling is supported on the graph of a function, and that it provides an approximation of the OT map otherwise.)
    In fact, \citet[Theorem~2.4]{GarciaTrillosSen} establish that the OT problem~\eqref{eqn:pop-distrcstr-Monge} from $F$ to $G$ with the non-standard cost can be equivalently written as an OT problem from $(\delta_{\B})_{\#}F$ to $G$ with the standard squared distance cost, and we adopt this in our explicit procedure given in Algorithm~\ref{alg:distr-cstr} since it allows one to apply existing numerical algorithms for OT; note that we use $\rho\in\mathcal{P}(\Rbb^m\times\Rbb^m)$ to denote the coupling of $(\delta_{\B})_{\#}F$ to $G$, and previously we used $\pi\in\mathcal{P}(\Rbb^d\times\Rbb^m)$ to denote the coupling of $F$ to $G$.
    Our main result below is a rate of convergence for this procedure, and details for implementing Algorithm~\ref{alg:distr-cstr} in practice can be found in Appendix~\ref{app:comp}.

    Let us describe the assumptions needed in order to state and prove our main result.
    First, as described above, we need a consistent estimator $\hat G_n$ of $G$, meaning that for some specified sequence $\beta_n\to 0$ we have
    \begin{equation}\label{eqn:CM}
        W_2^2(\hat G_n, G) = O_{\Pbb}(\beta_n).
        \tag{D}
    \end{equation}
    In most cases, we take $\hat G_n$ to be the NPMLE of $G$ (or a smoothed version thereof), in which case many recent results (e.g., \citet{Soloff, PolyanskiyWuPoisson,KimSen}) can be used to determine $\beta_n$. 
    Second, we need to quantify convergence of the empirical distribution of $\Theta_1,\ldots, \Theta_n$ to $G$ in the Wasserstein metric; that is, for some sequence $\gamma_n\to 0$, we assume
    \begin{equation}\label{eqn:EW}
        W_2^2\left(\bar G_n,\; G\right) = O_{\Pbb}(\gamma_n) \qquad \textnormal{ where }\qquad \bar G_n := \frac{1}{n}\sum_{i=1}^{n}\delta_{\Theta_i}.
        \tag{EW}
    \end{equation}
    It is known that $\gamma_n$ depends on the dimension $m$, and
    explicit forms of $\gamma_n$ can be found in \citet{FournierGuillin} and \citet[Section~2.9]{StatisticalOT}; see~\eqref{eqn:emp-wass-rate}.
    Third, recall that, under assumptions~\eqref{eqn:2M}, \eqref{eqn:Z-dens}, and \eqref{eqn:Bayes-dens}, Theorem~\ref{thm:pop-distrcstr} guarantees that problem~\eqref{eqn:pop-distrcstr} admits a unique solution $\delta_{\distrcstr} = \nabla \phi\circ \delta_{\B}$; presently, we further assume the regularity that
    \begin{equation}\label{eqn:L}
        \phi \textnormal{ is $\lambda$-strongly convex and $L$-smooth for some } \lambda,L>0.
        \tag{R}
    \end{equation}
    While this condition may be hard to verify in practice, it is indeed commonly assumed in literature which proves rates of convergence for estimation of OT maps (e.g., \citet{ManoleNilesWeed, SlawskiSen, DebGhosalSen}); see \citet[Section~2]{ManoleBalakrishnanNW} for some sufficient conditions.

    \begin{algorithm}[t]
	\caption{EB approximations $\hat \delta_{\distrcstr}$ and $\hat \delta_{\gencstr}$ of the oracle denoisers $\delta_{\distrcstr}$ and $\delta_{\gencstr}$.
    }\label{alg:distr-cstr}
	
	\begin{algorithmic}[1]
            
		  \State \textbf{input:} samples $Z_1,\ldots, Z_n$, likelihood $\{P_{\theta}\}_{\theta}$, domain $K\subseteq \Rbb^m$
		\State \textbf{output:} denoising function $\hat{\delta}_{\ast}:\{Z_1,\ldots, Z_n\}\to\Rbb^m$
            \State $\hat{\delta}_{\B}(\,\cdot\,) \leftarrow$ EB approximation of $\delta_{\B}(\,\cdot\,)$
            \State $\hat{G}_n\leftarrow$ estimate of $G$ \hspace{11.25em} \textbf{if} $\distrcstr$ 
            \State $\hat{I}_\ell\leftarrow$ estimate of $\int_{\Rbb^m}\psi_{\ell}\diff G$ for all $1\le \ell\le k$\hspace{1.25em} \textbf{if} $\gencstr$ 


            \State $\hat{\rho}_{\ast} \leftarrow$ \textbf{minimize} \hspace{2.2em} $\int_{\Rbb^m\times \Rbb^m}\|\vartheta-\eta\|^2\diff \rho(\vartheta,\eta)$
            \State \hspace{2.3em} \textbf{over} \hspace{4.7em} probability measures $\rho\in\Pcal(\Rbb^m\times K)$
            \State \hspace{2.3em} \textbf{with}\hspace{4.7em} $\rho(\{\hat{\delta}_{\B}(Z_i)\}\times K) = \frac{1}{n}$ for all $1\le i \le n$
            \State \hspace{2.3em} \textbf{and either} \hspace{1.8em}$\rho(\Rbb^m\times\diff\eta) = \hat{G}_n(\diff\eta)$ \hspace{10.1em} \textbf{if} $\distrcstr$ 
            \State  \hspace{4.6em} \textbf{or} \hspace{3.4em} $\int_{\Rbb^m\times K}\psi_{\ell}(\eta)\diff \rho(\vartheta,\eta) = \hat{I}_{\ell}$ for all $1\le \ell \le k$\hspace{1.4em} \textbf{if} $\gencstr$
            \State$\hat{\delta}_{\ast}(Z_i) \leftarrow \int_{K}\eta\,\diff \hat{\rho}_{\ast}(\eta\,|\,\delta_{\B}(Z_i))$ for all $1\le i \le n$
            \State \textbf{return} $\hat{\delta}_{\ast}$
	\end{algorithmic}
    \end{algorithm}

    These considerations lead us to our next main result.

    \begin{theorem}\label{thm:emp-distrcstr}
        Under assumptions \eqref{eqn:Z-dens}, \eqref{eqn:Bayes-dens},  \eqref{eqn:pos-def},  \eqref{eqn:L}, \eqref{eqn:CM}, \eqref{eqn:EW}, and \eqref{eqn:EBQ}, the denoiser $\hat{\delta}_{\distrcstr}$ from Algorithm~\ref{alg:distr-cstr} satisfies
        \begin{equation}\label{eqn:distr-cstr-regret}
            \frac{1}{n}\sum_{i=1}^{n}\left\|\hat{\delta}_{\distrcstr}(Z_i)-\delta_{\distrcstr}(Z_i)\right\|^2 =O_{\Pbb}\left(\alpha_n^{\sfrac{1}{2}}\vee \beta_n\vee \gamma_n\right) \qquad \textnormal{as}\qquad n\to\infty.
        \end{equation}
    \end{theorem}

    The conclusion of Theorem~\ref{thm:emp-distrcstr}  for the denoiser given in Algorithm~\ref{alg:distr-cstr} is that the rate of convergence of distribution-constrained Bayes denoising is determined by whichever of unconstrained Bayes denoising (line 3), deconvolution (line 4), and convergence of the empirical Wasserstein distance is slowest.
    While in Appendix~\ref{app:distrcstr-discussion} we give a detailed discussion of the resulting rates of convergence in some concrete settings of interest, we emphasize that $\beta_n$ typically dominates in nonparametric problems (see Appendix~\ref{app:rates} for a matching lower bound) and that $\gamma_n$ typically dominates in parametric problems.
    The term $\alpha_n^{\sfrac{1}{2}}$ is likely an artifact of the proof and we expect that it can be improved to $\alpha_n$ with further work, but we do not pursue this since it does not affect the rates in cases of interest.

    As discussed in Remark~\ref{rem:other-targets}, the difficulty of deconvolution leads to slow rates of convergence for $\hat{\delta}_{\distrcstr}$. It is possible to achieve faster rates by instead constraining the distribution of $\delta(Z)$ to equal some distribution $\tilde{G}\neq G$ such as the low-order approximation of $G$ proposed by~\citet{eddington1940correction} and~\citet{carroll2004low}.
    The empirical Bayes results of this subsection generalize as well, provided that all of the hypotheses replace $G$ with $\tilde{G}$; for example, assumption~\eqref{eqn:CM} is replaced by the assumption that there exists a consistent estimator $\hat{G}_n$ of $\tilde{G}$ and that its rate of convergence is $\beta_n$.
    An alternative approach for mitigating slow rates of deconvolution is via problem~\eqref{eqn:pop-gencstr}, which provides a way to specify certain features of $G$ that should be matched by a denoiser (rather than specifying a particular distribution $\tilde{G}$) while leading to the minimal risk; the next subsection concerns EB approximation of the corresponding oracle.
    
    \subsection{General-Constrained Denoising}\label{subsec:emp-gencstr}

    In this last subsection, we consider EB approximations of the oracle denoiser from~\eqref{eqn:pop-gencstr}, which enforces a constraint weaker than that of~\eqref{eqn:pop-distrcstr}.
    Our strategy is described explicitly in Algorithm~\ref{alg:distr-cstr} and which we discuss in detail in Appendix~\ref{sec:gen-cstr}.
    Methodologically, we regard this denoiser as a tradeoff between problems \eqref{eqn:pop-varcstr} and \eqref{eqn:pop-distrcstr}; the former only encodes limited information about the latent variable distribution (the first two moments), and the latter suffers from slow rates of convergence (due to the difficulty of the deconvolution problem) and may have larger denoising error. Often some flexible intermediate behavior may be desired.

    We require three types of assumptions for our result on the rate of convergence of $\hat{\delta}_{\gencstr}$ to $\delta_{\gencstr}$.
    First are the population-level assumptions~\eqref{eqn:Z-dens}, \eqref{eqn:Bayes-dens}, \eqref{eqn:cty}, \eqref{eqn:quadratic-growth}, and~\eqref{eqn:UI} already required in Theorem~\ref{thm:gen-cstr-num}.
    Second are conditions analogous to \eqref{eqn:L}, \eqref{eqn:CM}, \eqref{eqn:EW}, and \eqref{eqn:EBQ} required in Theorem~\ref{thm:emp-distrcstr}, but with $G$ replaced by the projected latent variable distribution $H$ defined in Theorem~\ref{thm:gen-cstr-num} and the discussion thereafter.
    Third are genuinely new conditions, e.g., we require some consistent estimate $\hat{I}_{\ell}$ of $\int_{\Rbb^m}\psi_{\ell}\diff G$ for each $1\le \ell\le k$.
    In the theory of Appendix~\ref{sec:gen-cstr} we restrict to the case that $\hat{I}_{\ell}$ is taken to be a plug-in estimator $\int_{\Rbb^m}\psi_{\ell}\diff \hat{G}_n$ for some deconvolution estimate $\hat{G}_n$ of $G$; other methods of estimating the constraints may be desirable in practice (e.g., method of moments estimators of Subsection~\ref{subsec:emp-varcstr} but in the higher-moment constraints setting of Example~\ref{ex:moment-cstr}), but they may require an initial projection step (as in, e.g., \cite{WuYang}) to ensure feasibility of the linear program in Algorithm~\ref{alg:distr-cstr}.
    Under such assumptions we obtain the following:
    
    \begin{theorem}[informal version of Theorem~\ref{thm:gen-cstr-emp}]\label{thm:gencstr-informal}
        Under suitable assumptions, the denoiser $\hat{\delta}_{\gencstr}$ from Algorithm~\ref{alg:distr-cstr} satisfies
        \begin{equation}\label{eqn:gen-cstr-regret-body}
            \frac{1}{n}\sum_{i=1}^{n}\left\|\hat{\delta}_{\gencstr}(Z_i)-\delta_{\gencstr}(Z_i)\right\|^2 =o_{\Pbb}(1) \qquad \textnormal{as}\qquad n\to\infty.
        \end{equation}
    \end{theorem}

    The precise statement of this result includes an explicit rate of convergence which is analogous to Theorem~\ref{thm:emp-distrcstr}, but with $\beta_n$ and $\gamma_n$ replaced with counterparts of  \eqref{eqn:CM} and \eqref{eqn:EW} which are based on $H$ instead of $G$.

    \section{Heterogeneous (Empirical) Constrained Bayes Denoising}\label{sec:heterosked-body}

    Our results and methods so far have been 
    developed under the assumption that all $n$ observations follow model~\eqref{eq:Mdl}. While this setting covers many important applications, it excludes settings with heterogeneity in the likelihood, e.g., settings with varying precision in the observations, as well as settings with heterogeneity in the prior. To accommodate such heterogeneity, we extend model~\eqref{eq:Mdl} by including a heterogeneity parameter $\Xi\in\Rbb^p$ in addition to the latent variable $\Theta$ and the observation $Z$:
    \begin{equation}\label{eq:Mdl-hetrogen}
        (\Theta\,|\,\Xi=\xi)\sim G_{\xi}\qquad \mbox{and} \qquad (Z \mid \Theta = \theta,\Xi=\xi) \sim P_{\theta,\xi}.
    \end{equation}
    Here, $\{G_{\xi}\}_{\xi}$  is an unknown family of probability measures on $\Rbb^m$ and $\{P_{\theta,\xi}\}_{\theta,\xi}$ is a known family of probability distributions on $\Rbb^d$.
    We assume that $\Theta$ is unobserved while $(Z,\Xi)$ is observed.
    We also consider the EB setting where we have i.i.d. samples $(\Theta_1,\Xi_1,Z_1),\ldots, (\Theta_n,\Xi_n,Z_n)$ from~\eqref{eq:Mdl-hetrogen} and we observe $(\Xi_1,Z_1),\ldots, (\Xi_n,Z_n)$. 

    In our application to astronomy in Subsection~\ref{subsec:astro}, we model measurements as $Z_i \sim \mathcal{N}(\Theta_i, \Sigma_i)$ where the covariance matrices $\Sigma_i$ vary across observations. This fits framework~\eqref{eq:Mdl-hetrogen} by setting $\Xi_i = \Sigma_i$ and $P_{\theta,\xi} = \mathcal{N}(\theta, \xi)$, allowing us to handle the heterogeneous measurement uncertainty in astronomical data. In our baseball application in Subsection~\ref{subsec:baseball}, part of our modeling involves Poisson observations with heterogeneous exposure, where we set $\Xi_i = \lambda_i$ and $P_{\theta,\xi} = \textnormal{Poi}(\xi\theta)$. We call $\lambda_i$ the exposure~\citep{norberg1979credibility}, representing the number of games played by player $i$, so that $\Theta_i = \mathbb E[Z_i/\lambda_i]$ is the rate normalized by exposure time.
    In our marketing application in Subsection~\ref{subsec:marketing} where $Z_i$ are $2\times 2$ contingency tables, we let $\Xi_i$ be the marginal (i.e., row and column) counts, $\Theta_i$ the log-odds ratio of a successful conversion to website visit, and $Z_i$ has one entry distributed as a hypergeometric random variable.
    In all three applications, we assume that $\{G_{\xi}\}_{\xi}$ does not depend on $\xi$, i.e., there is some unknown $G$ satisfying $G_{\xi} = G$ for all $\xi$.\footnote{Allowing for dependence of $G_{\xi}$ on $\xi$ is important in some applications, see Appendix~\ref{subsec:heterogeneous_applications} and~\citet{ChenEB}.}

    In Appendix~\ref{sec:heterosked}, we discuss applications of~\eqref{eq:Mdl-hetrogen} in more detail (Appendix~\ref{subsec:heterogeneous_applications}) and
    describe modifications of the EB procedures in Section~\ref{sec:emp} to construct constrained denoisers under model~\eqref{eq:Mdl-hetrogen}.
    Constraints on the denoising procedure have two possible forms, namely a constraint on the marginal distribution of $\delta(Z,\Xi)$ or a constraint on the conditional distribution of $\delta(Z,\Xi)$ given $\{\Xi=\xi\}$ for almost every $\xi$.
    In particular, we discuss marginal and conditional variance constraints (Appendix~\ref{subsec:het-varcstr}) leading to denoisers
    $\delta_{\mathcal{M}\varcstr}$ and $\delta_{\mathcal{C}\varcstr}$ respectively, marginal distribution constraints (Appendix~\ref{subsec:marg-distr-cstr}) leading to a denoiser $\delta_{\mathcal{M}\distrcstr}$, and marginal general constraints (Appendix~\ref{subsec:marg-gen}), leading to a denoiser $\delta_{\mathcal{M}\gencstr}$; as usual, we write $\hat{\delta}_{\ast}$ to denote an EB approximation of an oracle denoiser $\delta_{\ast}$.
    See also Section~\ref{subsec:marg-v-cond-illustration} for a numerical illustration of the difference between the marginal and conditional variance constraints.

    \begin{figure}[t]
        \centering
        \includegraphics[width=1.0\linewidth]{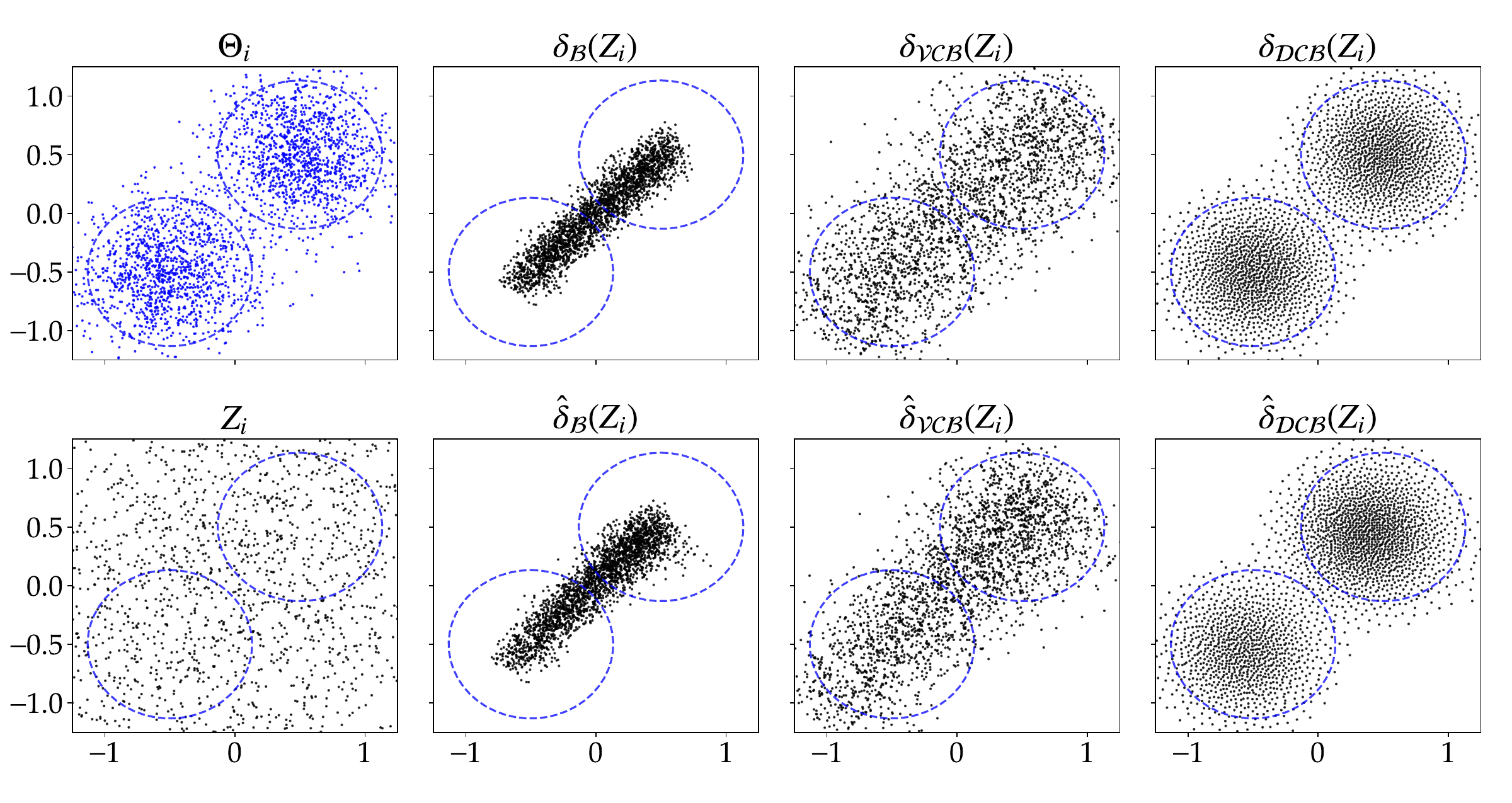}
        \caption{Denoising the simulated bivariate dataset of Section~\ref{sec:sim}.        
        The latent variables come from a two-component Gaussian mixture model, and the observations come from a Gaussian likelihood.
        We show the latent variables and the observations (first column), the Bayes and EB denoisers (second column), the variance-constrained Bayes and EB denoisers (third column), and the distribution-constrained Bayes and EB denoisers (fourth column).
        In each plot, we also show the contour for two standard deviations around each component (blue) for reference.}
        \label{fig:simulation}
    \end{figure}

    \section{Numerical Experiments}\label{sec:sim}

    This section contains a numerical experiment which provides quantitative evidence that the constrained denoisers we consider in this work can provide a balance between the denoising problem and the deconvolution problem.

    The two notions of error we consider are as follows, where $\delta$ is a denoiser of interest.
    By the \textit{denoising error} we mean the expected value of $\frac{1}{n}\sum_{i=1}^{n}\|\delta(Z_i)-\Theta_i\|^2$, and by the \textit{deconvolution error} we mean the expected value of the squared $W_2$ distance between the empirical measure of $\delta(Z_1),\ldots, \delta(Z_n)$ and $G$.
    In the case of the data-driven denoisers, we note that $\delta$ is based on the same $Z_1,\ldots, Z_n$ appearing in the error.

    Our experiment is a simulation in $\Rbb^2$ where the latent variables come from a two-component Gaussian mixture model (with component means $\pm(\sfrac{1}{2},\sfrac{1}{2})$ and component variance $\tau^2I_2$ with $\tau^2= 0.1$) and the observations come from a Gaussian likelihood (with likelihood variance $\sigma^2I_2$ with $\sigma^2=1.0$).
    We take $n=2{,}500$ samples for each of $B=10$ independent trials.

    Our denoisers of interest are the oracle and empirical versions of the~\eqref{eqn:Bayes},~\eqref{eqn:pop-varcstr}, and~\eqref{eqn:pop-distrcstr} denoisers, defined as follows.
    For the EB approximation of~\eqref{eqn:Bayes}, we use the plug-in of $\hat{G}_n$ as the smooth NPMLE (Example~\ref{ex:smooth-f-model}) with $\tau^2=0.1$.
    For the EB approximation of \eqref{eqn:pop-varcstr}, we combine this unconstrained EB denoiser with the method of moments estimator (Example~\ref{ex:MOM}).
    For the EB approximation~\eqref{eqn:pop-distrcstr}, we combine this unconstrained EB denoiser with the same deconvolution estimate $\hat{G}_n$.
    As baselines, we also include the MLE given by $\delta_{\textnormal{MLE}}(z)=z$, and the smooth NPMLE $\hat{G}_n$ itself, which can be used for deconvolution but not for denoising.
    
    In Table~\ref{tab:sim-2d} we show the Monte Carlo estimates of the two errors for the denoisers of interest and the baselines, and in Figure~\ref{fig:simulation} we show a visualization, for one trial, of the latent variables, the observations, and the denoisers.
    As expected, the (unconstrained) Bayes denoisers achieve the minimum denoising errors, and the distribution-constrained Bayes denoisers achieve the minimum deconvolution errors.
    Observe that adding constraints to the Bayes denoiser increases the denoising error by a modest amount, while decreasing the deconvolution error substantially.
    This phenomenon transfers also to the empirical level, as the EB approximations closely match their oracle counterparts.
    See Appendix~\ref{app:further-sim} for further simulations exhibiting the same phenomenon.

    \begin{table}
        \centering
        \resizebox{0.95\width}{!}{
        \begin{tabular}{|c|c|c|c|c|c|}
            \hline
            \makecell{oracle\\ denoiser} & \makecell{denoising\\error} & \makecell{deconvolution\\error} & \makecell{empirical\\denoiser} & \makecell{denoising\\ error} & \makecell{deconvolution\\error} \\
            \hline
        $\delta_{\textnormal{MLE}}$ & $2.0000\pm0.0331$ & $0.8794\pm0.0332$ & --- & --- & --- \\
        $\delta_{\mathcal{B}}$ & $\mathbf{0.4617\pm0.0092}$ & $0.1301\pm0.0038$ & $\hat{\delta}_{\mathcal{B}}$ & $\mathbf{0.4691\pm0.0071}$ & $0.1256\pm0.0226$ \\
        $\delta_{\mathcal{VCB}}$ & $0.5978\pm0.0137$ & $0.0144\pm0.0011$ & $\hat{\delta}_{\mathcal{VCB}}$ & $0.6042\pm0.0234$ & $0.0435\pm0.0124$ \\
        $\delta_{\mathcal{DCB}}$ & $0.5903\pm0.0128$ & $\mathbf{0.0003\pm0.0000}$ & $\hat{\delta}_{\mathcal{DCB}}$ & $0.6009\pm0.0214$ & $\mathbf{0.0144\pm0.0062}$ \\
        --- & --- & --- & $\hat{G}_{\textnormal{NPMLE}}$ & --- & $\mathbf{0.0142\pm0.0062}$ \\
        \hline

        \end{tabular}
        }
        \caption{Monte Carlo estimates of the denoising error and the deconvolution error, in the simulated bivariate dataset of Section~\ref{sec:sim}.
        We display the mean error, plus or minus one standard error, for each of the oracle and EB denoisers of interest in the paper.
        As a baseline, we also display the denoising and deconvolution errors for the MLE, and the deconvolution error for the (smooth) NPMLE.}
        \label{tab:sim-2d}        
    \end{table}

\section{Applications}\label{sec:app}

    In this section, we consider applications of our methodology to three domains.
    \subsection{Astronomy}\label{subsec:astro}

    First we consider an application to the problem of denoising the relative chemical abundances present in a catalog of stars.
    In short, the relative abundances of various chemicals influence stellar formation, evolution, and dynamics, hence their estimation is a fundamental task in astronomy.

    Our data comes from the Apache Point Observatory Galactic Evolution Experiment survey (APOGEE); see \citet{Majewski_2017, Abolfathi_2018,Ratcliffe_2020,Soloff} for further detail.
    More specifically, we focus on Oxygen-Iron and Nitrogen-Iron relative abundances (denoted [O/Fe] and [N/Fe], respectively) for stars in the red clump catalog, following the analysis of \cite{Ratcliffe_2020}.
    For each star ($n = 2{,}000$ stars sampled at random from the full catalog of $2.7\times 10^4$ stars), we have an estimate of the values [O/Fe] and [N/Fe], and the measurement error of these estimates is known.
    
    To state our model explicitly, we make the following modeling assumptions:
    \begin{enumerate*}
        \item Every star $i$ has latent chemical abundances $\Theta_i = (\Theta_i^{\textnormal{[O/Fe]}},\Theta_i^{\textnormal{[N/Fe]}})$, and $\Theta_1,\ldots, \Theta_n$ are i.i.d. from an unknown distribution $G$ on $\Rbb^2$.
        \item The measurement of every star $i$ admits error given by a diagonal covariance matrix $\Sigma_i$, the matrices $\Sigma_1,\ldots, \Sigma_n$ are i.i.d. from some unknown distribution on $\covspace(2)$, and $\Theta_1,\ldots, \Theta_n$ are independent of $\Sigma_1,\ldots, \Sigma_n$.
        \item For each star $i$, the observed pair of measurements $Z_i = (Z_i^{\textnormal{[O/Fe]}},Z_i^{\textnormal{[N/Fe]}})$ has conditional distribution $\Ncal(\Theta_i,\Sigma_i)$ given $\Theta_1,\ldots,\Theta_n,\Sigma_1,\ldots, \Sigma_n$, and $Z_1,\ldots, Z_n$ are conditionally independent given $\Theta_1,\ldots,\Theta_n,\Sigma_1,\ldots, \Sigma_n$.
    \end{enumerate*}
    The assumptions above are common in existing literature, e.g., each $\Sigma_i$ being diagonal \citep[page 16]{TingWeinberg}, and $\Theta_i$ and $\Sigma_i$ being independent \citep[Section~1.1]{Soloff}.

    Next we describe the denoising methods of interest.
    All of our approaches will be based on smooth $G$-modeling, as in Example~\ref{ex:smooth-f-model}, where we assume that $G$ is a Gaussian location mixture model, where the covariance matrix of each Gaussian component is bounded below, in the PSD order, by $\tau^2I_2$, where $\tau^2 = 0.0025$.
    Then, we compute the heteroskedastic multivariate NPMLE over this restricted class of $G$, yielding an unconstrained EB denoiser $\hat{\delta}_{\B}$.
    Last, we use $\hat{\delta}_{\B}$ to compute the marginal variance-constrained denoiser $\hat{\delta}_{\mathcal{M}\varcstr}$ and the marginal distribution-constrained denoiser $\hat{\delta}_{\mathcal{M}\distrcstr}$, as described in Section~\ref{sec:heterosked}.

\begin{figure}[t]
        \centering
        \includegraphics[width=1.0\linewidth]{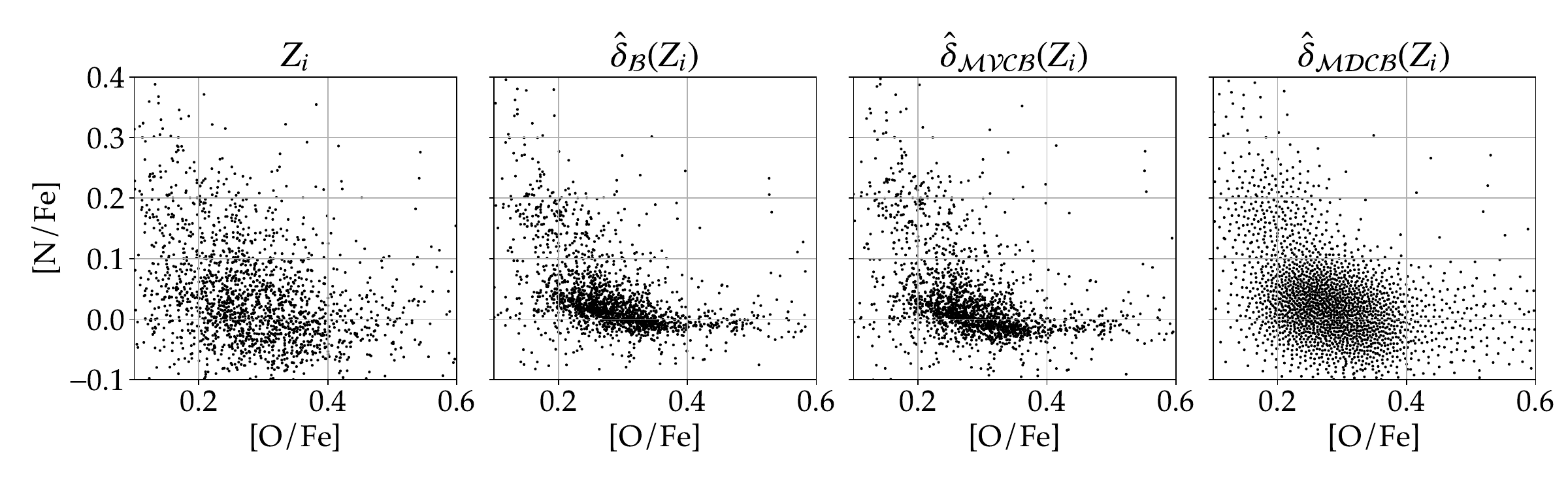}
        \caption{Denoising the stellar relative chemical abundances dataset of Subsection~\ref{subsec:astro}. We plot the raw data (first), the unconstrained EB denoised data (second), the marginal variance-constrained EB denoised data (third), and the marginal distribution-constrained EB denoised data (fourth).}
        \label{fig:astro_2}
    \end{figure}

    The results are shown in Figure~\ref{fig:astro_2}.
    First, note that the EB denoiser $\hat{\delta}_{\B}$ recovers some interesting latent low-dimensional structure from the data.
    Next, we note that the transformation from the EB denoiser $\hat{\delta}_{\B}$ to the variance-constrained EB denoiser $\hat{\delta}_{\mathcal{M}\varcstr}$ is affine but it is not merely scaling, since the upper branch slightly changes its angle (which suggests that the correlation between the two components may be underestimated by the EB denoiser); although the visual difference is quite subtle, a quantitative way to see this is to note that the affine transformation is given by the matrix
    \begin{equation*}
        \hat{\transport} = \begin{pmatrix}
            1.063 & 0.041 \\ 0.041 & 1.193
        \end{pmatrix},
    \end{equation*}
    meaning the [O/Fe] coordinate changes by roughly 6.3\% and the [N/Fe] coordinate changes by roughly 19.3\%.
    
    Lastly, we note that the distribution-constrained EB denoiser $\hat{\delta}_{\mathcal{M}\distrcstr}$ reveals a similar latent low-dimensional structure, but that the resulting distribution is more equispaced and this allows one to identify more subtle features in the distribution (e.g., the separated cluster appearing at the ends of the top branch, and the positions of points outside the bulk of the data set).
    \subsection{Baseball}\label{subsec:baseball}

    Second, we consider an application to denoising the joint distribution of minor-league batting skill and major-league batting skill for rookie baseball players.\footnote{Strictly speaking, a ``rookie'' player is one who is in his first season meeting at least one criterion involving the number of games played or the number of at-bats.
    In this paper, we simply mean a player that is in his first major league appearance, since number of games played will be a part of our model.}
    For the sake of simplicity, we focus on the statistic RBI (runs batted in) which is a common measure of batter performance.
    
    Our data consists of all batters in their first season of major league play during the 2022 to 2024 seasons of the MLB which is publicly available through the website FanGraphs.\footnote{Available at: \url{https://www.fangraphs.com/}}
    For each player ($n=324$), we observe his RBIs and his number of games played, in both his final minor league season and his initial major league season.
    To state our model explicitly, we make the following modeling assumptions:
    \begin{enumerate*}
        \item Every player $i$ has a latent bivariate skill $\Theta_i = (\Theta_i^{\textnormal{min}},\Theta_i^{\textnormal{maj}})$, and $\Theta_1,\ldots,\Theta_n$ are i.i.d. from an unknown distribution $G$ on $[0,\infty)\times[0,\infty)$.
        \item Every player $i$ participates in a random proportion $\lambda_i=(\lambda_i^{\textnormal{min}},\lambda_i^{\textnormal{maj}})$ of all games per minor- and major-league season, and $\lambda_1,\ldots, \lambda_n$ are i.i.d. from an unknown distribution on $[0,\infty)\times[0,\infty)$, and $\lambda_1,\ldots, \lambda_n$ are independent of $\Theta_1,\ldots, \Theta_n$.\footnote{   The assumption of independence between $\Theta_1,\ldots, \Theta_n$ and $\lambda_1,\ldots, \lambda_n$ is debatable and warrants further investigation in future work.
    On the one hand, good batters are more likely to be placed in a batting lineup than bad batters.
    On the other hand, players who are desired for non-batting skills (catching, fielding, etc.) must be placed in the batting lineup in order to exercise these other skills. }
        \item For each player $i$, his RBI count $Z_i = (Z_i^{\textnormal{min}},Z_i^{\textnormal{maj}})$ has conditional distribution $\textnormal{Poi}(\lambda_i^{\textnormal{min}}\Theta_i^{\textnormal{min}})\otimes \textnormal{Poi}(\lambda_i^{\textnormal{maj}}\Theta_i^{\textnormal{maj}})$ given $\Theta_1,\ldots, \Theta_n,\lambda_1,\ldots,\lambda_n$, and $Z_1,\ldots, Z_n$ are conditionally independent given $\Theta_1,\ldots, \Theta_n,\lambda_1,\ldots,\lambda_n$.
    \end{enumerate*}

    Next we describe our denoising methods.
    All of our methods are based on $G$-modeling (without smoothness assumptions) where we posit that $G$ is a bivariate heterogeneous Poisson mixture model with conditionally independent components; this has been  studied in \citep[Section 6]{JanaPoliyanskiyWu}.
    
    First, we compute the NPMLE for $G$, yielding an unconstrained EB denoiser $\hat{\delta}_{\B}$.
    Second, as outlined in Section~\ref{sec:heterosked}, we compute the marginal variance-constrained denoiser $\hat{\delta}_{\mathcal{M}\varcstr}$ and the marginal general-constrained denoiser $\hat{\delta}_{\mathcal{M}\gencstr}$ corresponding to constraints on the first and second moments, on component-wise nonnegativity, and on all players having major league skill less than minor league skill.

    The results are shown in Figure~\ref{fig:baseball_bivariate}, which we now explain.
    First, note that the (unconstrained) EB denoiser $\hat{\delta}_{\B}$ reveals a latent structure with a positive trend between the two coordinates and with most players lying below the diagonal dashed line; this reflects the intuitive facts that major league skill is correlated with minor league skill and that most players have lower major league skill than their minor league skill.
    Second, we compare the unconstrained EB denoiser $\hat{\delta}_{\B}$ with the marginal variance-constrained denoiser $\hat{\delta}_{\mathcal{M}\varcstr}$; the shrinkage in the former leads one to underestimate variability between players.
    Lastly, compare the marginal variance-constrained denoiser $\hat{\delta}_{\mathcal{M}\varcstr}$ with the marginal general-constrained denoiser $\hat{\delta}_{\mathcal{M}\gencstr}$, and notice that most of the denoised data set is similar, except that $\hat{\delta}_{\mathcal{M}\varcstr}$ produces some negative estimates and some estimates above the diagonal, while $\hat{\delta}_{\mathcal{M}\gencstr}$ does not.

    \begin{figure}[t]
        \centering
        \includegraphics[width=1.0\linewidth]{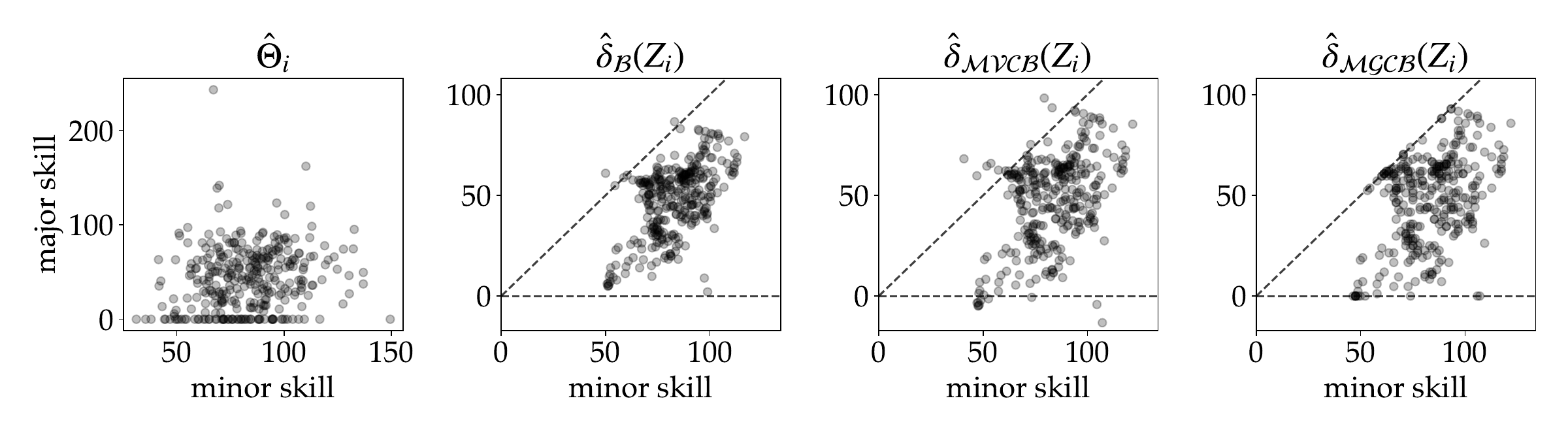}
        \caption{Denoising the rookie batters' dataset of Subsection~\ref{subsec:baseball}.
        We plot the exposure-standardized observations (first), the unconstrained EB denoised data (second), the marginal variance-constrained EB denoised data (third), and the marginal general-constrained EB denoised data (fourth) corresponding to variance constraints, component-wise nonnegativity constraints, and the constraint that major league skill is less than minor league skill.}
        \label{fig:baseball_bivariate}
    \end{figure}

    \subsection{Marketing}\label{subsec:marketing}

    Lastly, we consider the marketing application discussed in the introduction.
    Our data comes from the Hillstrom MineThatData experiment\footnote{Available at: \url{https://blog.minethatdata.com/2008/03/minethatdata-e-mail-analytics-and-data.html}} which aims to quantify the effect of marketing emails on consumer behavior.
    Specifically, the experiment tracked the online shopping website visits of $64{,}000$ recent purchasers after receiving a women's merchandise e-mail or no e-mail, and from this data we form market segments by stratifying customers on pre-treatment covariates: months since last purchase, dollars spent in the past year (binned), zip-code type (urban, suburban, or rural), whether the customer previously bought men's or women's merchandise, and whether they are a new customer.
    We also filtered out segments with too few customers or too few website visits, leading to $n=543$ segments.
    
    We impose the following model on the data $Z_1,\ldots,Z_n$, which are $2\times 2$ contingency tables, and we direct the reader to Appendix~\ref{subsec:heterogeneous_applications} for further detail:
    \begin{enumerate*}
        \item Each segment $i$ has parameter $\Theta_i\in\Rbb$ representing the log-odds ratio of successes (i.e., visits) in the treated arm, and, following \cite{vanhouwelingen1993bivariate}, $\Theta_1,\ldots,\Theta_n$ are i.i.d.\ from some unknown distribution $G$ on $\Rbb$. 
        \item Each segment $i$ has table margins $\Xi_i$ 
        so that $\Xi_1,\ldots,\Xi_n$ are i.i.d.\ and independent of $\Theta_1,\ldots,\Theta_n$.
        \item For each segment $i$, conditional on $\Theta_i$ and $\Xi_i$, the number of successes in the treated arm is sampled from Fisher's noncentral hypergeometric likelihood (see \eqref{eqn:nchg}), and the remaining elements of a $2\times 2$ contingency table $Z_i$ are fully determined by combining this with $\Xi_i$.
    \end{enumerate*}

Our denoising methods are as follows.
First, the natural frequentist baseline is the sample log-odds ratio, which we stabilize by adding pseudocounts of $\sfrac{1}{2}$ to each entry, i.e., $\hat{\Theta}_i = \log((X_{i,1}+\sfrac{1}{2})/(Y_{i,1}+\sfrac{1}{2})) -  \log((X_{i,0}+\sfrac{1}{2})/(Y_{i,0}+\sfrac{1}{2}))$.
Then, our empirical Bayes methods involve smooth $G$-modeling (Example~\ref{ex:smooth-f-model}) with component variance $\tau^2=0.01$ and atoms supported on $[-3,3]$.
Our unconstrained Bayes denoiser is the posterior mean given by plugging in $\hat{G}_n$ for $G$; since the setting is heterogeneous in the sense of Supplement~\ref{sec:heterosked}, we use marginal implementations for the constrained denoisers, i.e., the marginal variance-constrained denoiser (third panel) of Supplement~\ref{subsec:marg-var-cstr} and the marginal distribution-constrained denoiser (fourth panel) of Supplement~\ref{subsec:marg-distr-cstr}.
    For the variance-constrained denoiser, moreover, we estimate the mean and variance of $\Theta_i$ by plug-in with respect to $\hat{G}_n$.
    The results are shown in Figure~\ref{fig:hillstrom}, which we previously discussed in the Introduction.

    \section{Discussion}

    Constrained denoising is by now a classical idea that has found use across applications. The optimal transport framework developed in this paper substantially extends the scope of constrained denoising by showing that it is fully compatible with nonparametric empirical Bayes estimators, extending it to higher-dimensional settings, and introducing constraints beyond variance---including distributional and general constraints, as well as marginal and conditional constraints under heterogeneity. 
    From a broader perspective, our work illustrates that optimal transport can provide new insights and new methods for empirical Bayes and related latent variable models; see \cite{LiangI,LiangII,RigolletWeed} for other recent works in this direction.
    
    For practitioners, several natural questions arise: Which empirical Bayes denoiser should be used among a suite of parametric and nonparametric options? Which type of constraint is most appropriate---variance, distributional, or general? And if general constraints are chosen, which functionals of the latent distribution should be matched? We do not provide definitive answers to these questions, as the optimal choices depend on both the specific application domain and the goals of downstream analyses. \citet{lee2024improving} conducted a detailed study providing guidelines for choosing among denoising methods (including variance-constrained approaches) for the analysis of multisite trials. With the expanded framework developed here, we envisage similar guidelines emerging for other concrete application areas, guided by the requirements of domain experts.

    \subsection*{Reproducibility}

    All figures in the paper can be reproduced in Python with the Jupyter notebooks available at \url{https://github.com/aqjaffe/constrained-denoising-EB-OT}.

    \subsection*{Acknowledgments}
    We would like to thank Stephen Raudenbush and Tom Loredo for helpful discussions and for pointing us to the literature on constrained EB estimation.
    We also thank Brad Ross for providing detailed comments on an early version of this manuscript, and Jake Soloff and Ed George for many helpful suggestions.
    Finally, we thank Ellen Woods for many useful conversations about baseball.

    \onehalfspacing
    \setlength{\bibsep}{0.2pt plus 0.3ex}
    \bibliography{refs}
    \bibliographystyle{abbrvnat}

    \appendix
    
    \section{Additional Empirical Bayes Approaches}\label{sec:additional-EB}

    While the main body of the paper is primarily focused on nonparametric EB methods, we note that our results apply to (unconstrained) EB denoisers $\hat{\delta}_{\B}$ arising from many different approaches.
    In this section we briefly review some alternative approaches which we believe will be important in other applications.

    \subsection{$F$-Modeling}\label{subsec:f-mod}

    In some settings, the Bayes denoiser $\delta_{\B}$ depends only on the marginal distribution $F$ of $Z_1,\ldots, Z_n$.
    Then an EB approximation $\hat{\delta}_{\B}$ of $\delta_{\B}$ arises by modeling $F$ directly, rather than $G$.
    This is possible when $P$ is a continuous exponential family with natural parameterization, due to \textit{Eddington's/Tweedie's formula} \citep{dyson1926method, efron2011tweedie}, but it is also possible in many special cases where similar formulas are available~\citep{Robbins1956, robbins1982estimating, cressie1982useful}.
    Here we follow~\citet{efron2014two} and distinguish between $G$-modeling and $F$-modeling approaches, although we note that $F$-modeling can also be viewed as a special case of $G$-modeling.

    \subsection{Conjugate Parametric Models}\label{subsec:conj-param}

    Suppose in $d=m=1$ that $\{P_{\theta}\}_{\theta}$ is a mean-parameterized exponential family and that $G$ is a conjugate prior for $\{P_{\theta}\}_{\theta}$.
        We also assume that $\{P_{\theta}\}_{\theta}$ has a \textit{quadratic variance function}, meaning there exists a quadratic polynomial $V:\Rbb\to\Rbb$ such that $P_{\theta}$ has variance $V(\theta)$ for all $\theta$; see \cite{NEF_QVF} for more on such models.
        It is well-known (see \cite{DiaconisYlvisaker}) that the Bayes denoiser $\delta_{\B}$ is a linear function in this setting; we have
        \begin{equation*}
            \delta_{\B}(z) = \frac{1}{1+\frac{1}{2}V''(0)+a^{-2}V(\mu)}z+\left(1-\frac{1}{1+\frac{1}{2}V''(0)+a^{-2}V(\mu)}\right)\mu
        \end{equation*}
        where $\mu=\Ebb[\Theta]$ and $a^2 =\Var(\Theta)$ are the mean and variance of $G$.
        Although $\mu$ and $a^2$ are not known, one can estimate them via the marginal distribution of $Z$ alone, leading to an EB denoiser of the form
        \begin{equation*}
		\hat{\delta}_{\B}(z) = \frac{1-\hat s^{-2}V(\hat\mu)}{1+\frac{1}{2}V''(0)}z + \left(1-\frac{1-\hat s^{-2}V(\hat\mu)}{1+\frac{1}{2}V''(0)}\right)\hat\mu,
	\end{equation*}
        where $\hat \mu:=\frac{1}{n}\sum_{i=1}^{n}Z_i$ and $\hat s^2 := \frac{1}{n-1}\sum_{i=1}^{n}(Z_i-\hat \mu)^2$ are the empirical mean and variance of $Z_1,\ldots, Z_n$.
        One can easily show in this setting that the convergence of $\hat{\delta}_{\B}$ to $\delta_{\B}$ has the parametric rate $n^{-1}$.

    In this setting, we may use Theorem~\ref{thm:VCB-solution} to see that the variance-constrained Bayes denoiser is given explicitly by
    \begin{align*}
        \delta_{\varcstr}(z) &= \frac{1}{\sqrt{1+\frac{1}{2}V''(0)+a^{-2}V(\mu)}}z+\left(1-\frac{1}{\sqrt{1+\frac{1}{2}V''(0)+a^{-2}V(\mu)}}\right)\mu \\
        &= \sqrt{\frac{1- s^{-2}V(\mu)}{1+\frac{1}{2}V''(0)}}z + \left(1-\sqrt{\frac{1-s^{-2}V(\mu)}{1+\frac{1}{2}V''(0)}}\right)\mu,
    \end{align*}
    and also that $\hat{\delta}_{\varcstr}$ is given by
    \begin{equation*}
        \hat{\delta}_{\varcstr}(z) = \sqrt{\frac{1- \hat s^{-2}V(\hat \mu)}{1+\frac{1}{2}V''(0)}}z + \left(1-\sqrt{\frac{1-\hat s^{-2}V(\hat \mu)}{1+\frac{1}{2}V''(0)}}\right)\hat\mu.
    \end{equation*}
    In the remainder of this section, we give some explicit formulas for the denoisers above, for the usual parametric models of interest.
    
    \subsubsection*{Gaussian}
 
    As we already discussed in the introduction, a fundamental example is $G=\mathcal{N}(\mu,a^2)$ and $P_{\theta} = \mathcal{N}(\theta,\sigma^2)$, where $\mu\in\Rbb$ and $a^2>0$ are unknown and $\sigma^2>0$ is known.
    Of course, this fits into the framework above with constant variance function $V\equiv\sigma^2.$
    As such, the Bayes denoiser is
	\begin{equation*}
		\delta_{\varcstr}(z) = \frac{a}{\sqrt{a^2+\sigma^2}}z + \left(1-\frac{a}{\sqrt{a^2+\sigma^2}}\right)\mu,
	\end{equation*}
	and its EB approximation is
	\begin{equation*}
		\hat{\delta}_{\varcstr}(z) = \frac{\sqrt{(\hat s^2-\sigma^2)_{+}}}{\hat s}z + \left(1- \frac{\sqrt{(\hat s^2-\sigma^2)_{+}}}{\hat s}\right)\hat\mu.
	\end{equation*}
        This coincides with the formulas given in \citet{LouisI} and \citet{Ghosh}.
        
	\subsubsection*{Poisson}

        Suppose that $G$ is a Gamma distribution with unknown shape $k>0$ and scale $\alpha>0$, and that $P_{\theta}=\textnormal{Poi}(\theta)$; this fits into the setting above with linear variance function $V(\theta) = \theta$.
        Consequently, we can compute the oracle variance-constrained Bayes denoiser and the empirical variance-constrained Bayes denoiser to be
	\begin{equation*}
		\delta_{\varcstr}(z) = \sqrt{\frac{\alpha}{\alpha+1}}z  + \left(1-\sqrt{\frac{\alpha}{\alpha+1}}\right)k\alpha,
	\end{equation*}
	and
	\begin{equation*}
		\hat{\delta}_{\varcstr}(z) = \frac{\sqrt{(\hat s^2-\hat \mu)_{+}}}{\hat s}z + \left(1- \frac{\sqrt{(\hat s^2-\hat \mu)_{+}}}{\hat s}\right)\hat\mu.
	\end{equation*}

	\subsubsection*{Exponential}

        Suppose $G$ is an inverse Gamma distribution with unknown shape $k>0$ and scale $\alpha>1$, and that $P_{\theta}$ is an exponential distribution with mean $\theta$.
        (Note that we are using the scale-parameterization of the exponential distribution, rather than the usual rate parameterization.)
        This fits into the setting above with $V(\theta) = \theta^2$.
        Thus, the oracle and empirical variance-constrained Bayes denoisers are given by
	\begin{equation*}
		\delta_{\varcstr}(z) = \frac{1}{\sqrt{\alpha}}z + \left(1- \frac{1}{\sqrt{\alpha}}\right)\frac{k}{\alpha-1}
	\end{equation*}
	and
	\begin{equation*}
		\hat{\delta}_{\varcstr}(z) = \frac{\sqrt{\hat s^2 - \hat \mu^2}}{\hat s \sqrt{2}}z + \left(1-\frac{\sqrt{\hat s^2 - \hat \mu^2}}{\hat s \sqrt{2}}\right)\hat\mu.
	\end{equation*}
	respectively.
	
	\subsubsection*{Geometric}

        Suppose that $P_{\theta}$ is a Geometric distribution on $\{1,2,\ldots\}$ with success probability equal to $1/(1+\theta)$ so that its variance function is $V(\theta) =\theta + \theta^2$.
        (Note that this is the mean parameterization rather than the natural parameterization.)
	It is easy to check in this case that the conjugate prior is $\{G_{k,\alpha}\}_{k,\alpha}$ given by
	\begin{equation*}
		\frac{\diff G_{k,\alpha}}{\diff \lambda}(\theta) = \frac{\Gamma(k+\alpha)}{\Gamma(k+1)\Gamma(\alpha-1)}\cdot\frac{\theta^k}{(\theta+1)^{k+\alpha}},
	\end{equation*}
	since one can indeed verify that the conditional distribution of $\Theta$ given $\{Z=z\}$ is $G_{k+z-1,\alpha+1}$.
	We can also compute that, for $\alpha>3$, the first two non-central moments of $G_{k,\alpha}$ are
	\begin{align*}
		\Ebb[\Theta] &= \int_{0}^{\infty}\theta\diff G_{k,\alpha}(\theta) = \frac{k+1}{\alpha-2} \\
		\Ebb[\Theta^2] &= \int_{0}^{\infty}\theta^2\diff G_{k,\alpha}(\theta) = \frac{(k+2)(k+1)}{(\alpha-3)(\alpha-2)},
	\end{align*}
	hence we have
	\begin{equation*}
		\Var(\Theta) = \frac{(k+1)(k+\alpha-1)}{(\alpha-3)(\alpha-2)^2}.
	\end{equation*}
	In particular, we compute the oracle variance-constrained Bayes denoiser to be
	\begin{equation*}
		\delta_{\varcstr}(z) = \frac{1}{\sqrt{\alpha-1}}z + \left(1-\frac{1}{\sqrt{\alpha-1}}\right)\frac{k+1}{\alpha-2},
	\end{equation*}
	and the empirical variance-constrained Bayes denoiser to be
	\begin{equation*}
		\hat{\delta}_{\varcstr}(z) = \frac{\sqrt{(\hat s^2-\hat \mu- \hat \mu^2)_{+}}}{\hat s\sqrt{2}}z + \left(1- \frac{\sqrt{(\hat s^2-\hat \mu - \hat \mu^2)_{+}}}{\hat s\sqrt{2}}\right)\hat\mu.
	\end{equation*}

        \section{Discussion of Rates of Convergence}\label{app:rates}

        \subsection{Variance-Constrained Denoising}

        This section contains a detailed discussion of the rates of convergence that Theorem~\ref{thm:varcstr-gm-npp} implies for the convergence of $\hat{\delta}_{\varcstr}$ to $\delta_{\varcstr}$, where $\delta_{\varcstr}$ is the solution to problem \eqref{eqn:pop-varcstr} defined in~\eqref{eq:varcstr} and $\hat{\delta}_{\varcstr}$ is its EB approximation defined in Algorithm~\ref{alg:Gaussian-mod-NP-prior}.

        First we consider the setting of Gaussian model and nonparametric prior, as in Example~\ref{ex:g-model}, where we use $G$-modeling to estimate $\hat{\delta}_{\B}$ via the NPMLE.
        If $G$ is sub-exponential, then we apply the rate of convergence given in \cite[Theorem~9]{Soloff} to take $\alpha_n = n^{-1}(\log n)^{c_m}$, for some constant $c_m>0$ depending only on the dimension $m$.
        Hence, $\hat{\delta}_{\varcstr}$ achieves the rate $n^{-1}(\log n)^{c_m}$.
        The results from \citet{Soloff} may be used to derive further rates of convergence for $\hat{\delta}_{\varcstr}$ under the Gaussian model, but where $G$ has some other structural properties of interest (e.g., sparse support).

    Second, we consider the setting of a Poisson model and nonparametric prior, where we can use either the $G$-modeling approach of the NPMLE (Example~\ref{ex:g-model}) or an $F$-modeling approach (Appendix~\ref{subsec:f-mod}) in particular the Robbins estimator \cite{Robbins1956}.
    By the results of \citet{JanaPoliyanskiyWu, PolyanskiyWuPoisson}, and \citet{ShenWu}, we may derive the following rates:
    \begin{itemize}
        \item If $G$ is compactly-supported, then for both the $G$- and $F$-modeling approaches we may take $\alpha_n=n^{-1}(\log n/\log\log n)^2$, which shows that $\hat{\delta}_{\varcstr}$ achieves the rate of convergence $n^{-1}(\log n/\log\log n)^2$.
        \item If $G$ is sub-exponential, then for both the $G$- and $F$-modeling approaches we may take $\alpha_n=n^{-1}(\log n)^3$, so $\hat{\delta}_{\varcstr}$ achieves the rate of convergence $n^{-1}(\log n)^3$.
        \item If $G$ is sufficiently heavy-tailed that we have \eqref{eqn:FM} but no higher moments, then it is known that the $G$-modeling approach leads to the rate $\alpha_n = n^{-2/3}(\log n)^c$ for some $c>0$ and that the $F$-modeling approach leads to the rate $\alpha_n = n^{-1/2}(\log n)^{c'}$ for some $c'>0$. (See \citet{ShenWu}.)
    \end{itemize}
    Despite the fact that the rates above for $G$-modeling and $F$-modeling agree in many cases, we remark that $G$-modeling is generally preferable in practice.
 
    Last, let us consider the setting of an exponential family with its conjugate prior, as in Appendix~\ref{subsec:conj-param}.
    As we already discussed, the denoiser $\hat{\delta}_{\varcstr}$ from Algorithm~\ref{alg:Gaussian-mod-NP-prior} has the following formula:
    \begin{equation*}
	\hat{\delta}_{\varcstr}(z) = \sqrt{\frac{(1-\hat s^{-2}V(\hat\mu))_{+}}{1+\frac{1}{2}V''(0)}}z + \left(1-\sqrt{\frac{(1-\hat s^{-2}V(\hat\mu))_{+}}{1+\frac{1}{2}V''(0)}}\right)\hat\mu,
    \end{equation*}
    where $\hat \mu:=\frac{1}{n}\sum_{i=1}^{n}Z_i$ and $\hat \sigma^2 := \frac{1}{n-1}\sum_{i=1}^{n}(Z_i-\hat \mu)^2$ are the empirical mean and variance of $Z_1,\ldots, Z_n$.
    Since we can take $\alpha_n = n^{-1}$, we obtain the expected behavior, that $\hat{\delta}_{\varcstr}$ achieves the parametric convergence rate $n^{-1}$.

    \subsection{Distribution-Constrained Denoising}\label{app:distrcstr-discussion}

    This section contains a brief discussion of the rates of convergence yielded by Theorem~\ref{thm:emp-distrcstr} for the convergence of $\hat{\delta}_{\distrcstr}$ to $\delta_{\distrcstr}$, where $\delta_{\distrcstr}$ is the unique solution to problem \eqref{eqn:pop-distrcstr} guaranteed by Theorem~\ref{thm:pop-distrcstr} and $\hat{\delta}_{\distrcstr}$ is its EB estimator defined in Algorithm~\ref{alg:distr-cstr}.

    First, consider the setting of a Gaussian model and nonparametric compactly supported latent variable distribution $G$.
    We already saw in Appendix~\ref{app:rates} that assumptions~\eqref{eqn:2M},~\eqref{eqn:Z-dens}, and~\eqref{eqn:Bayes-dens} are satisfied, and that \citet[Theorem~9]{Soloff} allows us to take $\alpha_n=n^{-1}(\log n)^{c_m}$ for some constant $c_m>0$ depending only on the dimension $m$.
    We may also use~\citet[Theorem~11]{Soloff} to take $\beta_n=(\log n)^{-1}$.
    Lastly, we use \citet[equation~(2.21)]{StatisticalOT} to take
    \begin{equation}\label{eqn:emp-wass-rate}
        \gamma_n := \begin{cases}
            n^{-\sfrac{1}{2}} &\textnormal{ if } 1\le m\le 3 \\
            n^{-\sfrac{1}{2}}\log n &\textnormal{ if } m = 4 \\
            n^{-\sfrac{2}{m}} &\textnormal{ if } m\ge 5,
        \end{cases}
    \end{equation}
    which is a worst-case rate and is achieved by discrete distributions, and which can be improved under further smoothness assumptions; for instance, if $G$ possesses a sufficiently regular density, then we may take $\gamma_n = n^{-1}$ in $m=1$ (e.g., \citet{BobkovLedoux}) and $\gamma_n=n^{-1}\log n$ in $m=2$ (e.g., \citet{AmbrosioStraTrevisan}).
    For any dimension $m$, the deconvolution term dominates and we conclude that $\hat{\delta}_{\distrcstr}$ achieves the slow rate $(\log n)^{-1}$.
    In fact, there is a matching lower bound in expectation which can be easily shown as follows, by using the triangle inequality for the $W_2$ distance and Jensen's inequality:
    \begin{align*}
        \Ebb\left[\frac{1}{n}\sum_{i=1}^{n}\left\|\hat{\delta}_{\distrcstr}(Z_i)-\delta_{\distrcstr}(Z_i)\right\|^2\right] &\ge \Ebb\left[W_2^2\left((\hat{\delta}_{\distrcstr})_{\#}\bar{F}_n,(\delta_{\distrcstr})_{\#}\bar{F}_n\right)\right] \\
        &\ge \Ebb\left[\left|W_2\left((\hat{\delta}_{\distrcstr})_{\#}\bar{F}_n,G\right)-W_2\left(G,(\delta_{\distrcstr})_{\#}\bar{F}_n\right)\right|^2\right] \\
        &\ge \left|\Ebb\left[W_2\left((\hat{\delta}_{\distrcstr})_{\#}\bar{F}_n,G\right)-W_2\left(G,(\delta_{\distrcstr})_{\#}\bar{F}_n\right)\right]\right|^2.
    \end{align*}
    Now note that the first term on the right side is at least as large as $(\log n)^{-1}$ (see \citet[Theorem~1]{DedeckerMichel})  and that the second term  is polynomial in $n$ by~\eqref{eqn:emp-wass-rate}, hence the whole expression is at least as large as $(\log n)^{-1}$.

    When $G$ is still nonparametric but further structure is assumed, the rates of convergence above can improve but deconvolution is typically still the dominant term.
    For instance, recent work \cite{KimSen} shows that, if $G$ is a Gaussian mixture, then deconvolution can achieve a polynomial rate, with exponent depending on the relative size of the component variance and likelihood variance.

    Second, consider the setting of an exponential family with its conjugate prior, as in Subsection~\ref{subsec:conj-param}.
    As discussed in Appendix~\ref{app:rates}, assumptions~\eqref{eqn:2M},~\eqref{eqn:Z-dens}, and~\eqref{eqn:Bayes-dens} are satisfied, and we may take $\alpha_n=n^{-1}$.
    Also, it is expected that one can typically take $\beta_n = n^{-1}$.
    Lastly, note that we may also apply \cite[equation~(2.21)]{StatisticalOT} to determine that $\gamma_n$ is given by~\eqref{eqn:emp-wass-rate}.
    For any dimension $m$, the empirical Wasserstein distance dominates (since it is always at least as slow as $n^{-\sfrac{1}{2}}$), and we conclude that $\hat{\delta}_{\distrcstr}$ achieves the rate $\gamma_n$.

    \section{Results on Empirical General-Constrained Denoising}\label{sec:gen-cstr}

    In this section we precisely state our main result on the rate of convergence of the EB approximation $\hat{\delta}_{\gencstr}$ of $\delta_{\gencstr}$ which we discussed informally in Subsection~\ref{subsec:emp-gencstr}
    Since in Theorem~\ref{thm:gen-cstr-num} we provided sufficient conditions for the existence of a unique solution to problem~\eqref{eqn:pop-gencstr}, we presently develop an EB methodology for targeting this oracle denoiser.
    Methodologically, we regard \eqref{eqn:pop-gencstr} as a way to tradeoff between problems \eqref{eqn:pop-varcstr} and \eqref{eqn:pop-distrcstr}; the former only encodes limited information about the latent variable distribution (the first two moments), the latter suffers from slow rates of convergence (due to the difficulty of the deconvolution problem) and may have larger denoising risk. Often some intermediate behavior may be desired.
    
    We describe our proposed procedure in Algorithm~\ref{alg:distr-cstr} in the main body, which is similar to the approach we previously used for the distribution-constrained problem~\eqref{eqn:pop-distrcstr}.
    Since we can estimate the marginal distribution $F$ via $\bar F_n:=\frac{1}{n}\sum_{i=1}^{n}\delta_{Z_i}$ and the cost function $c_G$ via $\hat c_n$ through $\hat{\delta}_{\B}$, it only remains to approximate the constraints of the linear program~\eqref{eqn:gen-str-Kant}.
    Assuming that this is possible (that is, that the integrals of $\Psi$ with $\hat{G}_n$ converge to the integrals of $\Psi$ with $G$), we may let $\hat{\delta}_{\gencstr}$ denote the barycentric projection of an optimal solution to the corresponding linear program.
    Note also that we restrict the $\eta$ variable of the linear program to lie in a compact set $K\subseteq\Rbb^m$ (line 6).

    In order to give a rate of convergence of $\hat{\delta}_{\gencstr}$ to $\delta_{\gencstr}$, we need to present a set of assumptions.
    To set this up, we first assume \eqref{eqn:Z-dens}, \eqref{eqn:Bayes-dens}, \eqref{eqn:cty}, \eqref{eqn:quadratic-growth}, and~\eqref{eqn:UI}, so that Theorem~\ref{thm:gen-cstr-num} is in  effect; then, we recall that there is a well-defined projected latent variable distribution which we denote $H$.
    Second, we assume    \begin{equation}\label{eqn:CS}
        \textnormal{the support of } H \textnormal{ is contained in the compact set } K\subseteq\Rbb^m
        \tag{CS}
    \end{equation}
    which implies that $H$ is compactly supported and that the support constraint of Algorithm~\ref{alg:distr-cstr} is well-specified.
    Third, we need
    \begin{equation}\label{eqn:int}
        \begin{pmatrix}
            \int_{\Rbb^m}\psi_1\diff G \\ \vdots \\\int_{\Rbb^m}\psi_k\diff G
        \end{pmatrix} \in \left\{\begin{pmatrix}
            \int_{\Rbb^m}\psi_1\diff G' \\ \vdots \\\int_{\Rbb^m}\psi_k\diff G'
        \end{pmatrix}: G'\in\Pcal(\Rbb^m) \right\}^{\circ}=:\mathcal{M}^{\circ}(\Psi).\tag{INT}
    \end{equation}
    where $A^{\circ}$ denotes the interior of a set $A\subseteq\Rbb^k$;
    this condition states that the integrals of $\psi_1,\ldots, \psi_k$ under the latent variable distribution $G$ do not lie on the boundary of all possible values of the integrals.
    Fourth, we impose assumption \eqref{eqn:CM} on $\hat{G}_n$ therein, typically taken to be the NPMLE or a smoothed version thereof, and we assume that it additionally satisfies
    \begin{equation}\label{eqn:CC}
        \int_{\Rbb^m}\psi_{\ell}(\eta) \diff \hat G_n(\eta) \to \int_{\Rbb^m}\psi_{\ell}(\eta) \diff G(\eta) \textnormal{ for all } 1\le \ell \le k
        \tag{CC}
    \end{equation}
    in probability as $n\to\infty$.
    For $\psi_{\ell}$ with at most quadratic growth, such convergence follows from~\eqref{eqn:CM} because of \citet[Definition~6.8]{Villani}.
    For $\psi_{\ell}$ whose growth is polynomial of degree $p\ge 2$, it is sufficient to additionally require $W_p(\hat G_n,G) = o_{\Pbb}(1)$ on top of~\eqref{eqn:CM}.

    Next we introduce some assumptions which are analogous to assumptions used in Theorem~\ref{thm:emp-distrcstr}.
    First, write $\hat{H}_n$ for the distribution of the $\Rbb^m$-coordinate under the optimal coupling $\hat{\pi}_{\gencstr}$ from Algorithm~\ref{alg:distr-cstr}.
    We require that, for some specified sequence $\delta_n\to 0$, we have
    \begin{equation}\label{eqn:CM-H}
        W_2^2(\hat H_n, H) = O_{\Pbb}(\delta_n).
        \tag{D$'$}
    \end{equation}
    It will be shown during the course of our main theorem that \eqref{eqn:CM-H} follows from \eqref{eqn:CS}, \eqref{eqn:int}, and \eqref{eqn:CC}, so this assumption only amounts to quantifying the rate of convergence.
    As we discuss further below, $\delta_n$ is not necessarily the rate of deconvolution, but rather the rate of convergence of a suitable projection of $\hat{G}_n$ onto the subspace of $\mathcal{P}_2(\Rbb^m)$ whose integrals against $\psi_1,\ldots, \psi_k$ match those of $G$.
    Second, we assume
    \begin{equation}\label{eqn:L-H}
        \delta_{\gencstr} = \nabla\phi\circ\delta_{\B,H} \textnormal{ where }\phi \textnormal{ is $\lambda$-strongly convex and $L$-smooth for some } \lambda,L>0,
        \tag{R$'$}
    \end{equation}
     which is analogous to assumption~\eqref{eqn:L} but with $\delta_{\B,H}$ in place of $\delta_{\B,G} \equiv \delta_{\B}$.
    Next, we make the following assumption on the convergence of the empirical Wasserstein distance for i.i.d.~samples $\eta_1,\ldots, \eta_n$ from $H$ defined via $\eta_i := \delta_{\gencstr}(Z_i)$: for some specified sequence $\varepsilon_n\to 0$ we have
    \begin{equation}\label{eqn:EW-H}
        W_2^2\left(\bar H_n,\; H\right) = O_{\Pbb}(\varepsilon_n) \qquad \textnormal{ where }\qquad \bar H_n := \frac{1}{n}\sum_{i=1}^{n}\delta_{\eta_i} 
        \tag{EW$'$}
    \end{equation}
    which is analogous to \eqref{eqn:EW}.
    Lastly, we state a technical condition analogous to~\eqref{eqn:UI}, namely
    \begin{equation}\label{eqn:UI-emp}
        \Psi \textnormal{ is uniformly integrable under } \Gamma(\bar{F}_n;\hat{G}_n,\Psi),\tag{UI$'$}
    \end{equation}
    i.e., $\lim_{R\to\infty}\sup_{1\le \ell \le k}\sup_{\pi\in \Gamma(\bar{F}_n;\hat{G}_n,\Psi)}\int_{\{\|\eta\|> R\}}|\psi_{\ell}(\eta)|\diff \pi(z,\eta) = 0$.
    As before, condition~\eqref{eqn:UI-emp} holds trivially whenever the constraints force every feasible $\hat{\pi}_n$ to be supported on a fixed compact set, for example if we take $K$ to be compact in Algorithm~\ref{alg:distr-cstr}.
       
    Finally, these considerations lead us to the following result.

    \begin{theorem}\label{thm:gen-cstr-emp}
        Under assumptions \eqref{eqn:Z-dens}, \eqref{eqn:Bayes-dens}, \eqref{eqn:cty}, \eqref{eqn:CS}, \eqref{eqn:quadratic-growth},  \eqref{eqn:UI},~\eqref{eqn:UI-emp}, \eqref{eqn:int}, \eqref{eqn:CM}, \eqref{eqn:CM-H}, \eqref{eqn:L-H}, \eqref{eqn:EW-H},  and \eqref{eqn:EBQ}, the denoiser $\hat \delta_{\gencstr}$ from Algorithm~\ref{alg:distr-cstr} satisfies 
        \begin{equation}\label{eqn:gen-cstr-regret}
            \frac{1}{n}\sum_{i=1}^{n}\left\|\hat{\delta}_{\gencstr}(Z_i)-\delta_{\gencstr}(Z_i)\right\|^2 =O_{\Pbb}(\alpha_n^{\sfrac{1}{2}}\vee \delta_{n}\vee \varepsilon_n)
        \end{equation}
        as $n\to\infty$.
    \end{theorem}

    The conclusion of Theorem~\ref{thm:gen-cstr-emp} for the denoiser $\hat{\delta}_{\gencstr}$ given in Algorithm~\ref{alg:distr-cstr} is that the rate of convergence of general-constrained Bayes denoising is determined by whichever of unconstrained Bayes denoising, convergence of the projected latent variable distributions, and convergence of the empirical Wasserstein distance is slowest.
    The rate of convergence of the projected latent variable distributions, described by $\delta_n$, exactly quantifies the statistical tradeoff of interpolating between problems \eqref{eqn:pop-varcstr} and \eqref{eqn:pop-distrcstr}. In the cases we consider in Example~\ref{ex:moment-cstr-continued}, we may anticipate $\delta_n$ to decay at nearly parametric rates when $k=0$ (no constraints), or when $k=2$ (variance constraints) and the empirical Bayes denoiser may be estimated at fast rates (as discussed in Example~\ref{ex:g-model}), or inversely polylogarithmic in $n$ when $k\to \infty$ (and we effectively need to deconvolve the latent variable distribution).
    In general, it is possible for certain functionals of the latent variable distribution to converge at a faster rate than deconvolution itself; see \cite{HanJiaoWeissman} for such results in a different but related discrete setting.

    \section{Details for Heterogeneous (Empirical) Constrained Bayes}\label{sec:heterosked}

    \noindent The unconstrained denoising problem in model~\eqref{eq:Mdl-hetrogen} is as follows:    \begin{equation}\label{eqn:heterosked}
	\begin{cases}
		\textnormal{minimize} &\Ebb\left[\|\delta(Z;\Xi)-\Theta\|^2\right] \\
		\textnormal{over}& \delta:\Rbb^{d}\times \Rbb^p\to\Rbb^{m}.
        \end{cases}
    \end{equation}
    Its solution is the conditional mean of $\Theta$ given $(Z,\Xi)$, which we denote by
    \begin{equation*}
        \delta_{\B}(z;\xi) = \Ebb[\Theta\,|Z=z,\Xi=\xi].
    \end{equation*}
    In the next subsections, we consider adding constraints to problem~\eqref{eqn:heterosked}. While we will not prove rigorous results regarding asymptotic theory, we extend the strategies of the previous sections to accommodate the heterogeneity.
    An important feature of heterogeneous constrained denoising is that we now must distinguish between two kinds of constraints:
    the so-called \textit{marginal constraints} on the distribution of $\delta_{\B}(Z;\Xi)$, and the so-called \textit{conditional constraints} on the distribution of $\delta_{\B}(Z;\Xi)$ given $\{\Xi=\xi\}$.

    Before explaining how to extend our constrained denoising methods to model~\eqref{eq:Mdl-hetrogen}, we first discuss possible applications of the model in more detail.

    \subsection{Applications of the heterogeneous setting}
    \label{subsec:heterogeneous_applications}
    Model~\eqref{eq:Mdl-hetrogen} allows for heterogeneity in both the prior $(G_{\xi})$ and the likelihood $P_{\theta,\xi}$. 
    Two important examples for heterogeneity in the likelihood are the following:
    \begin{itemize}
        \item (Heteroskedastic Gaussian.) For $m=d$, assume that the parameter $\theta$ is a vector in $\Rbb^m$ and the parameter $\xi$ is a positive semi-definite $m\times m$ matrix, which we choose to denote by $\Sigma$. Let $P_{\theta,\Sigma} = \Ncal(\theta,\Sigma)$. As one concrete application, suppose that conditional on $\Theta_i$, $Z_i$ is a sample average of $K_i$ i.i.d. observations $Z_{ij}$ with $Z_{ij} \sim \mathcal{N}(\Theta_i, \sigma^2)$. Then $Z_i \sim \mathcal{N}(\Theta_i, \sigma_i^2)$  with $\sigma_i^2=\sigma^2/K_i$.
        The heteroskedastic Gaussian setting encompasses this situation by taking $m=d=1$ and $\Xi_i = \sigma^2/K_i$.
        \item (Poisson with heterogeneous exposure.) For $m=d=1$, the parameter $\theta$ is a non-negative real number and the parameter $\xi$ is a non-negative real number, which we choose to denote by $\lambda$. Then let $P_{\theta,\lambda} = \textnormal{Poi}(\lambda\theta)$. We call $\lambda$ the exposure~\citep{norberg1979credibility}. The intuition here is that we are interested in the mean of the $i$-th observation normalized by the exposure time, that is, $\Theta_i = \mathbb E[Z_i/\lambda_i]$; see Subsection~\ref{subsec:baseball} for an application to baseball in which $\lambda_i$ is the number of games played by the $i$-th player. 
                \item (Noncentral hypergeometric / $2\times2$ contingency tables.) Suppose each $Z_i$ is  a $2\times2$ contingency table,
    \begin{equation*}
        Z_i = \begin{pmatrix} X_{1,i} & X_{0,i} \\ Y_{1,i} & Y_{0,i}
        \end{pmatrix},
    \end{equation*}
    with $X_{1,i}, X_{0,i}$ the successes and $Y_{1,i}, Y_{0,i}$ the failures
    in the treated and control arms. The effect of interest is the log-odds
    ratio $\Theta_i \in \mathbb R$. Let $\Xi_i = (n_{1,i}, n_{0,i}, t_i)$ be the margins,
    where $n_{1,i} = X_{1,i}+Y_{1,i}$, $n_{0,i} = X_{0,i}+Y_{0,i}$, and
    $t_i = X_{1,i}+X_{0,i}$. Conditional on $\Xi_i$ the table is determined
    by $X_{1,i}$, which follows Fisher's noncentral hypergeometric
    distribution: for $\max(0, t_i - n_{0,i}) \le x \le \min(t_i, n_{1,i})$,
    \begin{equation}\label{eqn:nchg}
        P_{\theta,\xi}(X_{1,i} = x)
        = \frac{\binom{n_{1,i}}{x}\binom{n_{0,i}}{\,t_i - x\,}\,e^{\theta x}}
               {\sum_{u}\binom{n_{1,i}}{u}\binom{n_{0,i}}{\,t_i - u\,}\,e^{\theta u}}.
    \end{equation}
    This places the problem in model~\eqref{eq:Mdl-hetrogen}, with the margins
    $\Xi_i$ as heterogeneity parameter. In an empirical Bayes context, this model has been considered e.g., by~\citet{vanhouwelingen1993bivariate} and \citet{efron1996empirical}.
    \end{itemize}

    Moreover, in several applications, the distribution of $\Theta$ may depend on $\Xi$. Such dependence implies that $\Theta_1,\ldots,\Theta_n$ are no longer exchangeable. Examples where this situation occurs are as follows.
    \begin{itemize}
        \item (Heteroskedastic Gaussian, continued.) In our application to astronomy (Subsection~\ref{subsec:astro}), we assume that $\Theta$ is independent of $\Xi \equiv \Sigma$, and as mentioned therein, this is a common assumption in the literature. However, several authors have noted that this assumption can be violated~\citep{stephens2017false, gu2017unobserved, weinstein2018grouplinear,ChenEB, ignatiadis2025empirical}.  If we model $G_{\xi}$ as a function of $\xi$ (using, e.g., the G-modeling approach of~\citet{ChenEB}), then we can address this issue (and the general model in~\eqref{eq:Mdl-hetrogen} accommodates this dependence).
        \item (Side-information.) In several common applications, we observe  covariates (features) $X_i$ in addition to the noisy observations $Z_i$ for each $i$. The covariates $X_i$ are informative about $\Theta_i$, though commonly they do not directly impact the likelihood, that is, $Z_i$ is independent of $X_i$ conditionally on $\Theta_i$. The setting with side-information falls under our framework by letting $\Xi_i=X_i$ and $P_{\theta,\xi}=P_{\theta}$. Examples of empirical Bayes modeling with side-information include~\citet{fayiii1979estimates, cohen2013empirical, ignatiadis2019covariatepowered, banerjee2023nonparametric}.
        \item (Heteroskedasticity and side-information.) It is possible to have both side-information and heterogeneity in the likelihood at the same time, such as described in~\citet{ChenEB, kwon2023optimal, ghosh2025steins}. We can model this situation, e.g., in the heteroskedastic Gaussian setting, by letting $\Xi_i = (\Sigma_i, X_i)$. Herein, $P_{\theta,\xi}$ only depends on the likelihood covariance (not on the side-information), while the prior can depend on both $\Sigma_i$ and $X_i$. 
    \end{itemize}

    \subsection{Variance-Constrained Denoising}\label{subsec:het-varcstr}

    We first seek to extend the variance-constrained denoising problem~\eqref{eqn:pop-varcstr} to the heterogeneous setting and consider two approaches; imposing a marginal variance constraint and imposing a conditional variance constraint.

    \begin{algorithm}[h!]
	\caption{EB approximation $\hat \delta_{\mathcal{M}\varcstr}$ of $\delta_{\mathcal{M}\varcstr}$.
    }\label{alg:heterosked-var-cstr-marg}

	\begin{algorithmic}[1]
            
		  \State \textbf{input:} samples $Z_1,\ldots, Z_n\in\Rbb^m$,  heterogeneity parameters $\Xi_1,\ldots, \Xi_n$, likelihood $\{P_{\theta,\xi}\}_{\theta,\xi}$
		\State \textbf{output:} denoising function $\hat{\delta}_{\mathcal{M}\varcstr}:\Rbb^m\to\Rbb^m$
            \State $\hat{\delta}_{\B}(\,\cdot\,;\,\cdot\,) \leftarrow$ EB approximation of $\delta_{\B}(\,\cdot\,;\,\cdot\,)$
            \State $\hat{G}_n\leftarrow$ estimate of $G$

            \State $\hat M \leftarrow$ sample covariance matrix of $\hat{\delta}_{\B}(Z_1;\Xi_1),\ldots \hat{\delta}_{\B}(Z_n;\Xi_n)$
            \State $\hat \mu \leftarrow \int_{\Rbb^m}\theta\diff \hat{G}_n(\theta)$
            \State $\hat A \leftarrow \int_{\Rbb^m}(\theta-\hat{\mu})(\theta-\hat{\mu})^{\top}\diff \hat{G}_n(\theta)$
            \State $\hat{\transport}\leftarrow \hat M^{-\sfrac{1}{2}}(\hat M^{\sfrac{1}{2}} \hat A \hat M^{\sfrac{1}{2}})^{\sfrac{1}{2}}\hat M^{-\sfrac{1}{2}}$
            \State $\hat{\delta}_{\mathcal{M}\varcstr}(\,\cdot\,;\,\cdot\,) \leftarrow \hat{\transport}(\hat{\delta}_{\B}(\,\cdot\,;\,\cdot\,)-\hat \mu) + \hat \mu$

            \State \textbf{return} $\hat{\delta}_{\mathcal{M}\varcstr}$
	\end{algorithmic}
    \end{algorithm}

    \begin{algorithm}[h!]
	\caption{EB approximation $\hat \delta_{\mathcal{C}\varcstr}$ of $\delta_{\mathcal{C}\varcstr}$.
    }\label{alg:heterosked-var-cstr-cond}

	\begin{algorithmic}[1]
            
		  \State \textbf{input:} samples $Z_1,\ldots, Z_n\in\Rbb^m$,  heterogeneity parameters $\Xi_1,\ldots, \Xi_n$, likelihood $\{P_{\theta,\xi}\}_{\theta,\xi}$
		\State \textbf{output:} denoising function $\hat{\delta}_{\mathcal{C}\varcstr}(\,\cdot\,;\,\cdot\,):\Rbb^m\times\Rbb^p\to\Rbb^m$
            \State $\hat{\delta}_{\B}(\,\cdot\,;\,\cdot\,) \leftarrow$ EB approximation of $\delta_{\B}(\,\cdot\,;\,\cdot\,)$
            \State $\hat{G}_n=\sum_{j=1}^{r}w_j\delta_{\theta_j}\leftarrow$ estimate of $G$

            \State $\hat M(\xi) \leftarrow \sum_{j=1}^{r}w_j\int_{\Rbb^d}(\hat\delta_{\B}(z,\xi)-\mu)(\hat\delta_{\B}(z,\xi)-\mu)^{\top}P_{\theta_j,\xi}(\diff z)$ for all $\xi\in\Rbb^p$
            \State $\hat \mu \leftarrow \int_{\Rbb^m}\theta\diff \hat{G}_n(\theta)$
            \State $\hat A \leftarrow \int_{\Rbb^m}(\theta-\hat{\mu})(\theta-\hat{\mu})^{\top}\diff \hat{G}_n(\theta)$
            \State $\hat{\transport}(\xi) \leftarrow \transport_{\hat{M}(\xi)}^{\hat{A}}= \hat M(\xi)^{-\sfrac{1}{2}}(\hat M(\xi)^{\sfrac{1}{2}} \hat A \hat M(\xi)^{\sfrac{1}{2}})^{\sfrac{1}{2}}\hat M(\xi)^{-\sfrac{1}{2}}$ for all $\xi\in\Rbb^p$
            \State $\hat{\delta}_{\mathcal{C}\varcstr}(z;\xi) \leftarrow \hat{\transport}(\xi)(\hat{\delta}_{\B}(z;\xi)-\hat \mu) + \hat \mu$ for all $(z,\xi)\in\Rbb^m\times\Rbb^p$

            \State \textbf{return} $\hat{\delta}_{\mathcal{C}\varcstr}$
	\end{algorithmic}
    \end{algorithm}

    \subsubsection{Marginal Variance-Constrained Denoising}\label{subsec:marg-var-cstr}
        
    For the first approach, we require 
    that the marginal mean and variance of the denoised data match the marginal mean and variance of the latent variable.
    That is, we solve problem~\eqref{eqn:heterosked} subject to the constraints
    \begin{equation}\label{eqn:marg-varcstr}
        \Ebb[\delta(Z;\Xi)] = \Ebb[\Theta] \qquad \textnormal{ and }\qquad\Cov(\delta(Z;\Xi)) = \Cov(\Theta).
    \end{equation}
    This approach was also followed by~\citet{Ghosh} in a small area application with heterogeneous sample sizes across areas.
    By noticing that this is closely related to the setting of Section~\ref{subsec:pop-varcstr} but with observation $\tilde{Z} := (Z,\Xi)$, it can be shown that the oracle denoiser is
    \begin{equation*}
        \delta_{\mathcal{M}\varcstr}(z;\xi):= \transport_{\Cov(\delta_{\B}(Z;\Xi))}^{\Cov(\Theta)}(\delta_{\B}(z;\xi)-\mu) + \mu
    \end{equation*}
    where $\mu = \Ebb[\Theta]$. 
    In order to implement this in practice, we can apply Algorithm~\ref{alg:Gaussian-mod-NP-prior} but with the requisite terms $\Ebb[\Theta]$, $\Cov(\Theta)$, and $\Cov(\delta_{\B}(Z;\Xi))$ estimated through modified approaches.
    It is clear that the third term can be estimated via the sample covariance of $\hat{\delta}_{\B}(Z_1;\Xi_1),\ldots, \hat{\delta}_{\B}(Z_n;\Xi_n)$ where $\hat{\delta}_{\B}(\,\cdot\,;\,\cdot\,)$ is any EB denoiser which accommodates the heterogeneity.
    When $\hat{\delta}_{\B}(\,\cdot\,;\,\cdot\,)$ arises from $G$-modeling, then one has an estimate $\hat G_n$ of $G$ and one can estimate $\Ebb[\Theta]$ and $\Cov(\Theta)$ via
    \begin{equation}\label{eqn:moment-NPMLE}
        \hat \mu:=\int_{\Rbb^m}\theta\diff \hat G_n(\theta)\qquad \text{ and}\qquad \,      \hat A:=\int_{\Rbb^m}(\theta-\hat\mu)(\theta-\hat\mu)^{\top}\diff \hat G_n(\theta).
    \end{equation}
    We summarize this in Algorithm~\ref{alg:heterosked-var-cstr-marg}.

    
    \subsubsection{Conditional Variance-Constrained Denoising}
    \label{subsec:conditional_variance}

    Another possible constraint is that the conditional mean and covariance of $\delta(Z;\Xi)$ match the conditional mean and covariance of $\Theta$ given $\Xi$.
    In other words, we may solve problem~\eqref{eqn:heterosked} subject to the constraints
    \begin{equation}\label{eqn:cond-varcstr}
        \Ebb[\delta(Z;\Xi)\,|\,\Xi]= \Ebb[\Theta\,|\,\Xi]\qquad\textnormal{ and } \qquad\Cov(\delta(Z;\Xi)\,|\,\Xi) = \Cov(\Theta\,|\,\Xi)\qquad \textnormal{ almost surely.}
    \end{equation}
    In this setting, it can be shown that the oracle denoiser is
    \begin{equation*}
        \delta_{\mathcal{C}\varcstr}(z,\xi) = \transport_{\Cov(\delta_{\B}(Z;\Xi)\,|\,\Xi=\xi)}^{\Cov(\Theta\,|\,\Xi=\xi)}(\delta_{\B}(z;\xi)-\mu) + \mu,
    \end{equation*}
    where $\mu=\Ebb[\Theta]$.
    Our EB approach will target the denoiser $\delta_{\mathcal{C}\varcstr}$ by plugging in suitable estimates for $\Ebb[\Theta], \Cov(\Theta\,|\,\Xi=\xi)$, and $\Cov(\delta_{\B}(Z,\Xi)\,|\,\Xi=\xi)$.
    If we additionally assume $\Theta\Perp \Xi$ (as we will assume in our later applications), then we may estimate $\Ebb[\Theta]$ and $\Cov(\Theta) \equiv \Cov(\Theta\,|\,\Xi)$ in the same way as above, via the corresponding moments of the NPMLE given in~\eqref{eqn:moment-NPMLE}.   
    However, $\Cov(\delta_{\B}(Z,\Xi)\,|\,\Xi=\xi)$ must be handled more carefully.
    If we are in the setting of $G$-modeling and the measure $\hat G_n$ is finitely supported (which is true for the NPMLE), then we may approximate
    \begin{equation*}
        \Cov(\delta_{\B}(Z;\Xi)\,|\,\Xi=\xi) =\int_{\Rbb^m}\int_{\Rbb^d}(\delta_{\B}(z;\xi)-\mu)(\delta_{\B}(z;\xi)-\mu)^{\top}P_{\theta,\xi}(\diff z)\diff G(\theta).
    \end{equation*}
   via
    \begin{equation}\label{eqn:CVC-numerical-integral}
        \sum_{j=1}^{r}w_j\int_{\Rbb^d}(\hat\delta_{\B}(z,\xi)-\mu)(\hat\delta_{\B}(z,\xi)-\mu)^{\top}P_{\theta_j,\xi}(\diff z).
    \end{equation}
    whenever
     $\hat G_n = \sum_{j=1}^{r}w_j\delta_{\theta_j}$.
    Lastly, the finitely-many integrals appearing on the right side may be estimated in many different ways, since the likelihood $\{P_{\theta,\xi}\}_{\theta,\xi}$ is known, e.g., by numerical integration, or by Monte Carlo.
    See the Algorithm~\ref{alg:heterosked-var-cstr-cond} for a summary of this method.

    \subsubsection{Numerical Illustration}\label{subsec:marg-v-cond-illustration}

    In order to highlight the difference between the marginal variance-constrained and conditional variance-constrained denoisers $\hat{\delta}_{\mathcal{M}\varcstr}$ and $\hat{\delta}_{\mathcal{C}\varcstr}$, we consider a simulation in the univariate heteroskedastic Gaussian setting.

    More precisely, we suppose that $\sigma^2_1,\ldots, \sigma^2_n$ are i.i.d. which are equal to $0.5$ or $8$ with equal probability, and we let $\Theta_1,\ldots,\Theta_n$ be i.i.d from a standard Gaussian distribution $\mathcal{N}(0,1)$.
    Conditional on these values, we let $Z_1,\ldots, Z_n$ be independent with $Z_i$ possessing distribution $\mathcal{N}(\theta_i,\sigma_i^2)$ for all $1\le i \le n$ with $n=1,500$.
    An observation $Z_i$ is called a \textit{low-variance} observation if $\sigma^2_i = 0.5$ and it is called a \textit{high-variance} observation if $\sigma^2_i = 8$.
    
    In Figure~\ref{fig:heterosked-EB} we show the result of applying three denoisers to this data set.
    The marginal variance-constrained denoiser (third column) leads to the incorrect variance when applied only to the low-variance observations (second row) or the high-variance observations (third row).
    Contrarily, the conditional variance-constrained denoiser (fourth column) obtains the correct variance for both subsets of the data.

    \subsection{Marginal Distribution-Constrained Denoising}\label{subsec:marg-distr-cstr}

    Next we seek to extend the distribution-constrained problem~\eqref{eqn:pop-distrcstr} to the heterogeneous setting.
    To state the problem precisely, we aim to solve problem~\eqref{eqn:heterosked} subject to the constraint that the marginal distribution of the denoised data matches the marginal distribution of the latent variable, meaning
    \begin{equation}\label{eqn:marg-varcstr-2}
        \delta(Z;\Xi) \overset{\mathcal{D}}{=} \Theta.
    \end{equation}
    As before, we observe that this is similar to Section~\ref{subsec:pop-distrcstr} with observation $\tilde{Z} := (Z,\Xi)$, hence the oracle denoiser is the composition of the Bayes denoiser $\delta_{\B}$ with the OT map from the distribution of $\delta_{\B}(Z;\Xi)$ to the distribution $G$ of $\Theta$.
    As such, modifying Algorithm~\ref{alg:distr-cstr} to implement this requires only replacing $\delta_{\B}(Z_i)$ with $\delta_{\B}(Z_i;\Xi_i)$ everywhere (i.e., in line 7 and 10).
    This leads to Algorithm~\ref{alg:marg-distr-cstr}, which is still just a linear program; as explained in Subsection~\ref{subsec:emp-distrcstr}, we use \citet[Theorem~2.4]{GarciaTrillosSen} to write this in terms of the standard cost rather than the non-standard cost $c_G$.

    While we focus only on the case of marginal distribution constraints, we remark that an extension to conditional distributional constraints is likely possible, but with further statistical and computational complications.

    \begin{figure}
        \centering
        \includegraphics[width=1.0\linewidth]{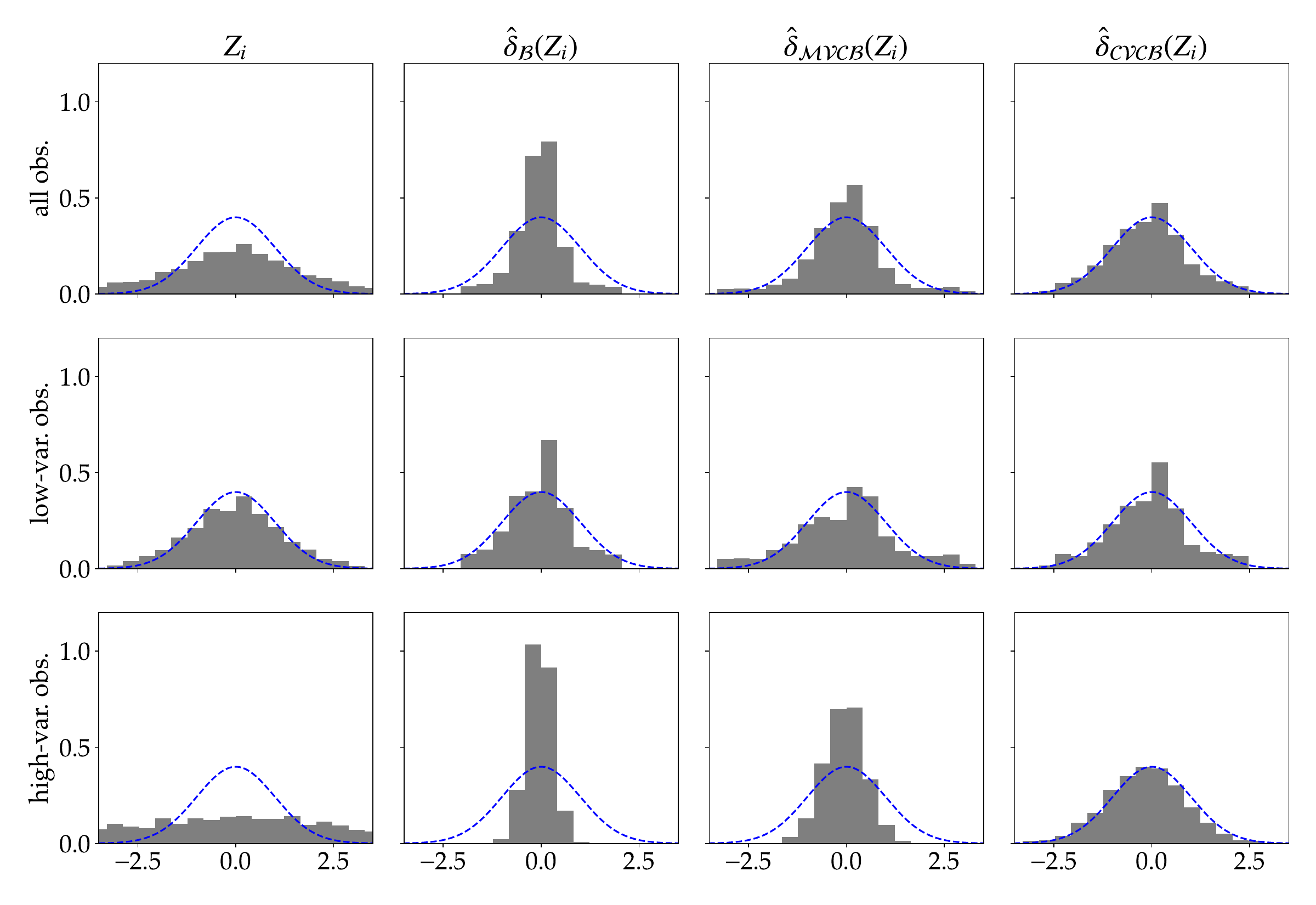}
        \caption{Denoising in the heterogeneous setting of Section~\ref{sec:heterosked} to a toy data set.
        By rows, we show histograms of all of the observations (first row), the subset of low-variance observations (second row), and the subset of high-variance observations (third row).
        By columns, we show the raw data (first column), the EB denoised data set (second column), the marginal variance-constrained EB denoised data set (third column), and the conditional variance-constrained EB denoised data set (fourth column).
        In each plot we show the density of the latent variable distribution (blue) for convenience.}
        \label{fig:heterosked-EB}
    \end{figure}

    \subsection{Marginal General-Constrained Denoising}\label{subsec:marg-gen}

    Lastly, we consider general constraints on the marginal distribution of $\delta(Z,\Xi)$, meaning we require
    \begin{equation}
        \Ebb[\psi_{\ell}(\delta(Z,\Xi))] = \Ebb[\psi_{\ell}(\Theta)] \quad \textnormal{ for all }1 \le \ell \le k,
    \end{equation}
    where $\psi_1,\ldots, \psi_k$ are some fixed functions.
    As in Subsection~\ref{subsec:pop-gencstr}, the oracle denoiser in this problem is the composition of the Bayes denoiser $\delta_{\B}$ with an OT map from the distribution of $\delta_{\B}(Z,\Xi)$ to some distribution $H$ (whose integrals with respect to $\Psi$ agree with those of $G$) but which is not a priori known.
    Nonetheless, this denoiser can be approximated via applying the same modifications of Subsection~\ref{subsec:marg-distr-cstr} to Algorithm~\ref{alg:distr-cstr}, leading to Algorithm~\ref{alg:marg-distr-cstr}.

    \begin{algorithm}[t]
	\caption{EB estimators $\hat \delta_{\mathcal{M}\distrcstr}$ and $\hat \delta_{\mathcal{M}\gencstr}$ of the oracle denoisers $\delta_{\mathcal{M}\distrcstr}$ and $\delta_{\mathcal{M}\gencstr}$.
    }\label{alg:marg-distr-cstr}
	
	\begin{algorithmic}[1]
            
		  \State \textbf{input:} samples $Z_1,\ldots, Z_n$, heterogeneity parameters $\Xi_1,\ldots, \Xi_n$, likelihood $\{P_{\theta}\}_{\theta}$, domain $K\subseteq \Rbb^m$
		\State \textbf{output:} denoising function $\hat{\delta}_{\ast}:\{(Z_1,\Xi_1),\ldots, (Z_n,\Xi_n)\}\to\Rbb^m$
            \State $\hat{\delta}_{\B}(\,\cdot\,;\,\cdot\,) \leftarrow$ EB approximation of $\delta_{\B}(\,\cdot\,;\,\cdot\,)$
            \State $\hat{G}_n\leftarrow$ estimate of $G$
            \State $\hat{\rho}_{\ast} \leftarrow$ \textbf{minimize} \hspace{2.2em} $\int_{\Rbb^m\times \Rbb^m}\|\vartheta-\eta\|^2\diff \rho(\vartheta,\eta)$
            \State \hspace{2.3em} \textbf{over} \hspace{4.7em} probability measures $\rho\in\Pcal(\Rbb^m\times K)$
            \State \hspace{2.3em} \textbf{with}\hspace{4.7em} $\rho(\{\hat{\delta}_{\B}(Z_i;\Xi_i)\}\times K) = \frac{1}{n}$ for all $1\le i \le n$
            \State \hspace{2.3em} \textbf{and either} \hspace{1.8em}$\rho(\Rbb^m\times\diff\eta) = \hat{G}_n(\diff\eta)$ \hspace{13.5em} \textbf{if} $\distrcstr$ 
            \State  \hspace{4.6em} \textbf{or} \hspace{3.4em} $\int_{\Rbb^m\times K}\psi_{\ell}(\eta)\diff \rho(\vartheta,\eta) = \int_{K}\psi_{\ell}\diff \hat{G}_n$ for all $1\le \ell \le k$\hspace{1.4em} \textbf{if} $\gencstr$
            \State$\hat{\delta}_{\ast}(Z_i;\Xi_i) \leftarrow \int_{K}\eta\,\diff \hat{\rho}_{\ast}(\eta\,|\,\delta_{\B}(Z_i,\Xi_i))$ for all $1\le i \le n$
            \State \textbf{return} $\hat{\delta}_{\ast}$
	\end{algorithmic}
    \end{algorithm}

\section{Additional Simulation}\label{app:further-sim}

This Appendix contains further experiments to continue the story of Section~\ref{sec:sim}.
That is, we consider a set of experiments which measures the same errors for the same set of estimators, but while some parameters of the problem vary.
    More precisely, we consider a simulation in $\Rbb$ where the latent variables come from a two-component Gaussian mixture model with component means $\pm1/\sqrt{2}$, component variance $\tau^2$, and likelihood variance $\sigma^2$, and $n=500$ samples;
    note that the latent variables and the data are precisely the projection of the previous experiment onto the line $\{(x_1,x_2): x_1=x_2\}$.
    The results are shown in Figure~\ref{fig:sim-1d}, where we fix $\tau^2=0.1$ and vary $\sigma^2$ over $[0,2]$, and where we fix $\sigma^2=1$ and vary $\tau^2$ over $[0,0.2]$.
    In all cases the same qualitative behavior occurs: the MLE has the largest error for both denoising and deconvolution; the denoising errors increase corresponding to the strength of the constraint; the deconvolution errors decrease corresponding to the strength of the constraint.

    \begin{figure}[t]
        \centering
        \includegraphics[width=\linewidth]{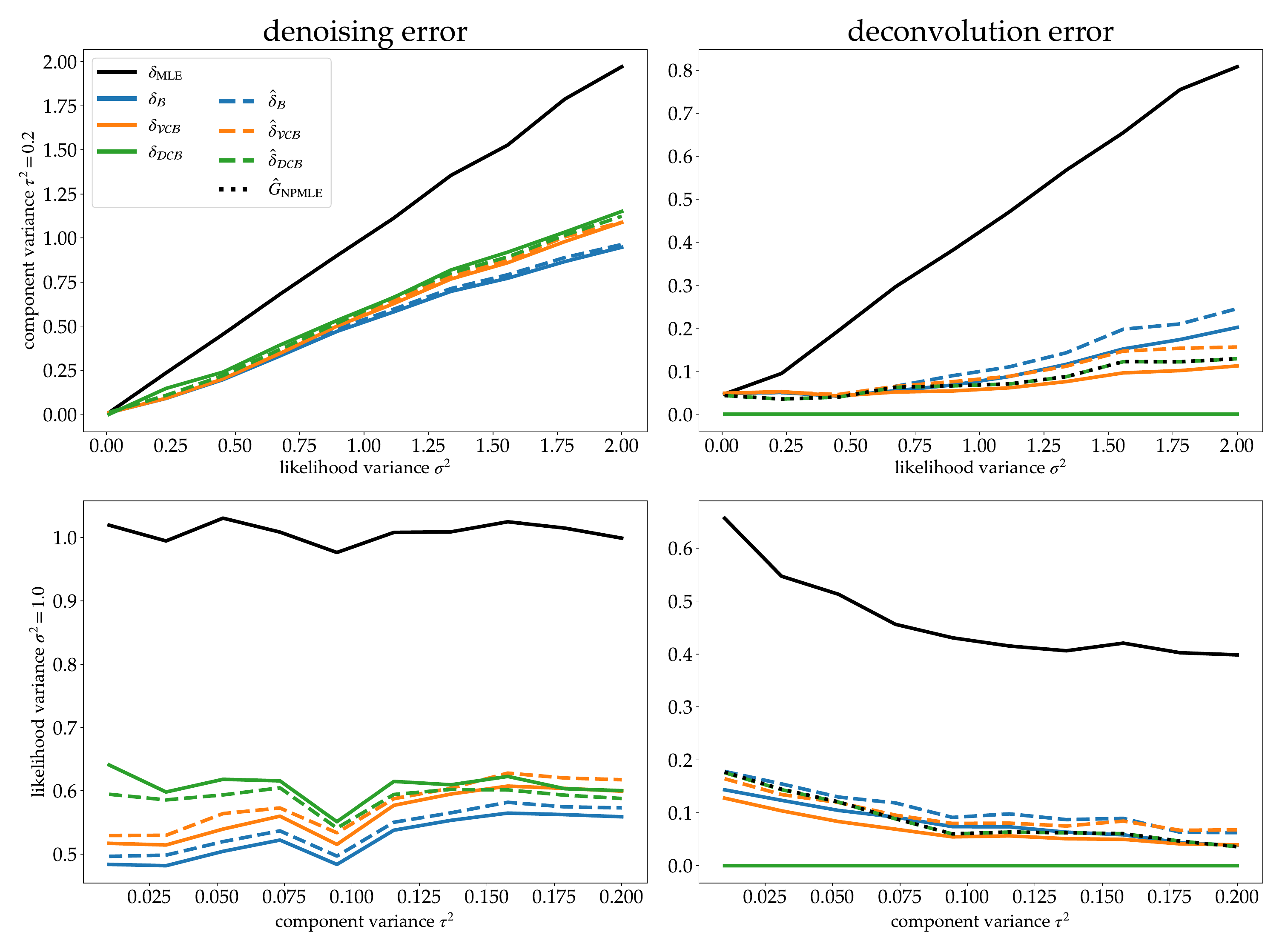}
        \caption{Monte Carlo estimates of the denoising and deconvolution errors in a simulated one-dimensional data set.
        For fixed $\tau^2=0.1$ and varying $\sigma^2$, we display the denoising error (first) and the deconvolution error (second).
        For fixed $\sigma^2=1.0$ and varying $\tau^2$, we display the denoising error (third) and the deconvolution error (fourth).}
        \label{fig:sim-1d}
    \end{figure}

\section{Computational Considerations}\label{app:comp}

    In this section, we briefly describe some computational considerations that arise in the implementation of various algorithms described in the main body of the paper. Our procedures also inherit some computational complexity from the unconstrained EB denoisers which they take as input. 
    Standard OT problems are solved with the Python Optimal Transport (POT) package \citep{POT}, and other convex optimization problems are solved via CVXPY~\citep{diamond2016cvxpy} using the Mosek~
    \citep{mosek} interior point convex programming solver.

    \subsection{$G$-modeling via nonparametric maximum likelihood}
    \label{subsec:npmle_computation}

    Our $G$-modeling methods with a discrete NPMLE are based on the software implementation in the Python package \texttt{npeb}\footnote{Available at: \url{https://github.com/jake-soloff/npeb}}. This implementation approximates the NPMLE through a two-step procedure:
    \begin{enumerate}
    \item Computing the NPMLE over a fixed discrete set of atoms using CVXPY/Mosek following the framework of~\citet{koenker2014convex} (cf.~\citet{Soloff} for a theoretical study of discretization considerations).
    \item Running expectation-maximization to improve atoms and weights jointly.
    \end{enumerate}
    In the Gaussian models (Figure~\ref{fig:simulation}, Figure~\ref{fig:heterosked-EB}, and Figure~\ref{fig:astro_2}) we initialize the NPMLE atoms in Step 1 to be the data points themselves, following the so-called \textit{exemplar$+$} method of \cite{Soloff}.
    In the Poisson model of Figure~\ref{fig:baseball_bivariate} we initialize the NPMLE atoms of step 1 to be a $50\times 50$ equispaced grid for the minimal bounding box of the data.
    
    For our $G$-modeling (e.g., in Section~\ref{subsec:astro}) with a smooth prior $G$ that is a Gaussian mixture model with a lower-bounded covariance matrix, we employ the strategy of~\citet[Section 2]{magder1996smooth}. This approach reduces the computation of the NPMLE over the mixture class to that of the discrete NPMLE, allowing us to again apply the \texttt{npeb} package.  Here we also choose the number of grid points and their location for Step 1.\ of \texttt{npeb} to be  a $150\times 150$ equispaced grid for the minimal bounding box of the data.

    \subsection{Variance-Constrained Denoising}
    \label{subsec:variance_constrained_computation}
    First we consider problem~\eqref{eqn:pop-varcstr} on denoising subject to variance constraints.
    Because of the explicit formulas afforded by the geometry of the Bures-Wasserstein space, we can avoid numerically solving an optimization problem in this case.
    Indeed, each step of Algorithm~\ref{alg:Gaussian-mod-NP-prior} has an explicit formulation (except for line 3, which can be implemented by existing EB approaches).

    Second, we consider the two variance-constrained denoising problems in the heterogeneous setting of Section~\ref{sec:heterosked}.
    The marginal variance-constrained denoiser $\hat{\delta}_{\mathcal{M}\varcstr}$ does not require any further computational burden compared to $\hat{\delta}_{\varcstr}$, at least in the case where $\hat{\delta}_{\B}$ arises from $G$-modeling.
    This is the case in all of our applications except for Figure~\ref{fig:heterosked-EB}.
    The conditional variance-constrained denoiser $\hat{\delta}_{\mathcal{C}\varcstr}$ is more complicated since it requires the numerical integration of the terms in \eqref{eqn:CVC-numerical-integral} for each evaluation of $\hat{\delta}_{\mathcal{C}\varcstr}(Z_i,\Xi_i)$; we compute these integrals via Monte Carlo with 100 trials each.

    \subsection{Distribution- and General-Constrained Denoising}
    \label{subsec:distribution_computation}
    Our procedures Algorithm~\ref{alg:distr-cstr} for approximating the solution to the distribution-constrained problem~\eqref{eqn:pop-distrcstr} and the general-constrained problem~\eqref{eqn:pop-gencstr} both require numerically solving a linear program (lines 6 to 10), so we briefly discuss some details of this.
    For problem~\eqref{eqn:pop-distrcstr} we use the solver from POT which is specifically designed to handle OT problems, and for problem~\eqref{eqn:pop-gencstr} we use Mosek's general-purpose solvers for linear programming in CVXPY.
    In either case, some discretization of the $\eta\in\Rbb^m$ coordinate is necessary.
    We always choose the discretization to be an equispaced grid on some axis-aligned bounding box, and the main difficulty is choosing an appropriate discretization.

    In the cases where $m=1$ (i.e., Figure~\ref{fig:heterosked-EB}), the complexity of the linear program problem scales linearly in the number of discretization points $k$.
    So, it is possible to take $k$ quite large without computational burden; for the sake of simplicity, in these cases, we fixed the discretization of $\eta$ to consist of $k=200$ equispaced points in the minimal bounding interval of the observations.
    
    In the cases where $m\ge 2$ (i.e., Figure~\ref{fig:simulation}, Figure~\ref{fig:baseball_bivariate}, and Figure~\ref{fig:astro_2}) the computation is more burdensome.
    We take $k=300$ in problem \eqref{eqn:pop-distrcstr} since the POT package is highly optimized for this task, and we take $k=200$ for problem~\eqref{eqn:pop-gencstr} although computation is a bit slow.
        
    Additionally, we remark that naively discretizing the linear program in Algorithm~\ref{alg:distr-cstr} may lead to infeasibility; this is because $\hat G_n$ witnesses feasibility for arbitrary probability measures but not necessarily for probability measures supported on the prescribed grid points.
    To avoid this, we augment the discretization by adding a subset of the NPMLE atoms (e.g., those with probability at least as large as some fixed $\varepsilon>0$) in addition to the grid points.
    We also note that the resulting procedures are relatively stable across varying hyperparameters like the bounding box and grid size.

        \section{Proofs from Section~\ref{sec:pop}}\label{app:oracle}

        This section contains the proofs of the main results in Section~\ref{sec:pop} on the population-level denoising problems.

        \begin{proof}[Proof of Theorem~\ref{thm:VCB-solution}]
		We begin by performing some manipulation on the objective of \eqref{eqn:pop-varcstr}.
        Note by assumption~\eqref{eqn:2M} that all terms below are well-defined.
	For feasible $\delta$, let us expand
        		\begin{align*}
			& \Ebb\left[\|\delta(Z)-\Theta\|^2\right] = \Ebb\left[\left\|\delta_\B(Z)-\Theta+\delta(Z)-\delta_\B(Z)\right\|^2\right] \\
			&= \Ebb\Big[\left\|\delta_\B(Z)-\Theta\right\|^2\Big]+\Ebb\Big[ \left\|\delta(Z)-\delta_\B(Z)\right\|^2\Big] + 2\Ebb\Big[\left(\delta_\B(Z)-\Theta\right)^{\top}(\delta(Z)-\delta_\B(Z))\Big].
		\end{align*}
		Notice that the first term on the right side is exactly the Bayes risk and that it is irreducible since it does not depend on $\delta$.
		The second term on the right side is exactly equal to
		\begin{equation*}
			\Ebb\left[\left\|\delta(Z)-\delta_\B(Z)\right\|^2\right] =\Var\left(\delta(Z)-\delta_\B(Z)\right),
		\end{equation*}
		where $$\Var(X) :=\trace(\Cov(X)) = \sum_{i=1}^{m}\Var(X_i)$$ denotes the variance of a random vector $X = (X_1,\ldots, X_m)$, 
		since $\delta$ being feasible implies $\Ebb[\delta(Z)-\delta_\B(Z)] = 0$.
		Finally, the third term on the right vanishes, since we can use the law of total expectation to compute:
		\begin{align*}
			\Ebb\left[\left(\delta_\B(Z)-\Theta\right)^{\top}(\delta(Z)-\delta_\B(Z))\right] & = \Ebb\left[\left(\delta_\B(Z)-\Ebb[\Theta\,|\,Z]\right)^{\top}(\delta(Z)-\delta_\B(Z))\right] \\
			&= \Ebb\left[0^{\top}(\delta(Z)-\delta_\B(Z))\right] = 0.
		\end{align*}
		In particular, we have shown
		\begin{equation*}
			\Ebb\left[\|\delta(Z)-\Theta\|^2\right] = R_{\B} +\Var\left(\delta(Z)-\delta_\B(Z)\right),
		\end{equation*}
		which we regard as a decomposition of the objective of \eqref{eqn:pop-varcstr} into the Bayes risk and a certain excess risk that we must minimize.
		
		The preceding paragraph shows that the solution set to problem \eqref{eqn:pop-varcstr} is equal to the solution set of the optimization problem
		\begin{equation}\label{eqn:solution-1}
			\begin{cases}
				\textnormal{minimize} &\Var(\delta(Z)-\delta_\B(Z)) \\
				\textnormal{over}& \delta:\Rbb^{d}\to\Rbb^{m} \\
				\textnormal{with}& \Ebb[\delta(Z)] = \Ebb[\Theta] \\
				\textnormal{and}& \Cov(\delta(Z)) = \Cov(\Theta).
			\end{cases}
		\end{equation}
		Now we apply a ``Gaussianization trick'' whereby we show that problem \eqref{eqn:solution-1} has the same optimal value as the following problem:
		\begin{equation}\label{eqn:solution-2}
			\begin{cases}
				\textnormal{minimize} &\Var(D-E) \\
				\textnormal{over}& \textnormal{jointly Gaussian } (D,E) \\
				\textnormal{with}& \textnormal{marginal distribution } D\sim \Ncal(\Ebb[\Theta],\Cov(\Theta)) \\
				\textnormal{and}&  \textnormal{marginal distribution }E\sim \Ncal(\Ebb[\Theta],\Cov(\delta_\B(Z))) \\
			\end{cases}
		\end{equation}
		To see that \eqref{eqn:solution-2} $\le$ \eqref{eqn:solution-1}, it suffices to find, for each $\delta$ that is feasible for \eqref{eqn:solution-1}, some $(D,E)$ that is feasible for \eqref{eqn:solution-2} such that $\Var(D-E) = \Var(\delta(Z) - \delta_\B(Z))$.
		Of course, we can let
		\begin{equation}\label{eq:Gauss-Coupling}
			\begin{pmatrix}
				D \\ E
			\end{pmatrix} \sim \Ncal\left(\begin{pmatrix}
			\Ebb[\Theta] \\ \Ebb[\Theta]
			\end{pmatrix}, \begin{pmatrix}
            \Cov(\Theta) & \Cov(\delta(Z),\delta_\B(Z)) \\
            \Cov(\delta_\B(Z),\delta(Z)) & \Cov(\delta_\B(Z))
            \end{pmatrix}\right)
		\end{equation}
		where $\Cov(X,Y):= \Ebb[XY^{\top}]-\Ebb[X]\Ebb[Y]^{\top}$ denotes the cross-covariance between two random vectors $X$ and $Y$ of the same dimension.
		In this case $(D,E)$ and $(\delta(Z),\delta_\B(Z))$ have the same second-order structure, so we have $\Var(D-E) = \Var(\delta(Z) -\delta_\B(Z))$.

To see that \eqref{eqn:solution-1} $\le$ \eqref{eqn:solution-2}, it suffices to show that the optimal value of \eqref{eqn:solution-2} is achieved by a coupling which corresponds to some feasible $\delta$ for \eqref{eqn:solution-1}; we do this by exhibiting an explicit denoiser with this property.
Indeed, note that~\eqref{eqn:solution-2} is indeed the (squared) 2-Wasserstein distance between the two multivariate normal distributions $\Ncal(\Ebb[\Theta], \Cov(\delta_\B(Z)))$ and $\Ncal(\Ebb[\Theta], \Cov(\Theta))$.
Then recall \citep{DowsonLandau, OlkinPukelsheim} that the optimal value of \eqref{eqn:solution-2} is exactly given by the Bures-Wasserstein distance
		\begin{equation*}
			\BW^2(\Cov(\Theta),\Cov(\delta_\B(Z))),
		\end{equation*}
        where $\BW(\Sigma,\Sigma') := W_2(\Ncal(0,\Sigma),\Ncal(0,\Sigma'))$,
		so, in particular we can use \eqref{eqn:pos-def} to define
		\begin{equation*}
			\delta_{\varcstr}(z) := \transport_{\Cov(\delta_\B(Z))}^{\Cov(\Theta)}(\delta_\B(z)-\Ebb[\Theta]) + \Ebb[\Theta].
		\end{equation*}
		Now use the fact that for two symmetric positive definite matrices $A,B$, we have $\trace(ABA) = \|AB^{1/2}\|_2^2$ (here $\|\cdot\|_2$ denotes the Frobenius norm), so that we have
		\begin{align*}
			\Var(\delta_{\varcstr}(Z)-\delta_\B(Z)) &= \Var\left(\left(\transport_{\Cov(\Ebb[\Theta\,|\,Z])}^{\Cov(\Theta)}-I\right)\delta_\B(Z)\right) \\
                &= \trace\left(\left(\transport_{\Cov(\delta_\B(Z))}^{\Cov(\Theta)}-I\right)\Cov(\delta_\B(Z))\left(\transport_{\Cov(\delta_\B(Z))}^{\Cov(\Theta)}-I\right)\right) \\
			&=  \left\|\left(\transport_{\Cov(\delta_\B(Z))}^{\Cov(\Theta)}-I\right)\Big(\Cov(\delta_\B(Z))\Big)^{\sfrac{1}{2}}\right\|_2^2 \\
			&= \BW^2 \Big(\Cov(\Theta),\Cov(\delta_\B(Z)) \Big).
		\end{align*}
		This completes the proof.
        Since we have exhibited a particular optimal denoiser $\delta_{\varcstr}$, the uniqueness claim follows from the general result \citet[Theorem~2.4]{GarciaTrillosSen} which requires the conditions \eqref{eqn:Bayes-dens} and \eqref{eqn:Z-dens}.
	\end{proof}

        \begin{proof}[Proof of Corollary~\ref{cor:risks}]
		We of course have $R_{\B}\le R_{\varcstr}$ by nature of imposing constraints, so it suffices to prove $R_{\varcstr}\le 2R_{\B}$ which is equivalent to $R_{\varcstr}-R_{\B}\le R_{\B}$.
        Since $\Ebb[\delta_{\B}(Z)] = \Ebb[\Theta] =:\mu$, we can use the tower property to compute:
        \begin{align*}
            R_{\B} &= \Ebb\left[\left\|\delta_{\B}(Z)-\Theta\right\|^2\right] \\
            &= \Ebb\left[\left\|\Theta - \mu\right\|^2\right] - \Ebb\left[\left\|\delta_{\B}(Z)-\mu\right\|^2\right] \\
            &= \trace(\Cov(\Theta)) - \trace(\Cov(\delta_{\B}(Z))) \\
            &= \trace(\Cov(\Theta)-\Cov(\delta_{\B}(Z))).
        \end{align*}
        Now recall $\Cov(\delta_{\B}(Z)) \preceq\Cov(\Theta)$ from \eqref{eqn:M-prec-A}, so we have shown
		\begin{equation*}
			R_{\B}= \left\|\Cov(\Theta) - \Cov(\delta_\B(Z))\right\|_1,
		\end{equation*}
        where $\|\cdot\|_1$ denotes the trace norm (also called the nuclear norm or the Schatten 1-norm), defined as $\|A\|_1:=\trace(A)$ for a real PSD matrix $A$.
        Now note that Theorem~\ref{thm:VCB-solution} implies
		\begin{equation*}
			R_{\varcstr}-R_{\B} =\BW^2\left(\Cov(\Theta),\Cov(\delta_\B(Z))\right).
		\end{equation*}
		Thus, it suffices to show that we have
		\begin{equation*}
			\BW^2\left(\Cov(\Theta),\Cov(\delta_\B(Z))\right) \le \left\|\Cov(\Theta) - \Cov(\delta_\B(Z))\right\|_1.
		\end{equation*}
		To show this, we let $\|\cdot\|_2$ denotes the Frobenius norm (also called the Schatten 2-norm), defined as $\|A\|_2:=\sqrt{\trace(A^2)}$ for a real symmetric matrix $A$.
        Then we note that we in fact have $\BW^2(A,B) \le  \|A^{\sfrac{1}{2}}-B^{\sfrac{1}{2}}\|_2^2 \le \|A - B\|_1$ for all $A,B\in\covspace(m)$ as a consequence of the Procrustes representation of the Bures-Wasserstein metric \cite[Lemma~1]{PanaretosSantoroBW} and the Powers-St{\o}rmer inequality \cite[Lemma~4.1]{PowersStormer}, respectively.
		(See also \cite[Lemma~1]{PanaretosSantoroBW}.)
	\end{proof}

\begin{proof}[Proof of Lemma~\ref{lem:risk-comparison}]
    Define the denoiser
    \begin{equation*}
        \tilde{\delta}(z) =  \sqrt{\frac{\Var(\Theta)}{\Var(\delta_0(Z))}}(\delta_0(z)-\Ebb[\Theta]) + \Ebb[\Theta].
    \end{equation*}
    Since this is feasible for problem~\eqref{eqn:pop-varcstr}, we have $R_{\varcstr}\le \Ebb[|\tilde{\delta}(Z)-\Theta|^2]=R(\tilde{\delta})$ and it thus suffices to show $R(\tilde{\delta})=\Ebb[|\tilde{\delta}(Z)-\Theta|^2] \le \Ebb[|\delta_0(Z)-\Theta|^2]=R(\delta_0)$.
    To do this, write $\alpha:=(\Var(\Theta)/\Var(\delta_0(Z)))^{1/2}$ and note that $\Ebb[\delta_0(Z)\,|\,\Theta]=\Theta$ implies $\Cov(\delta_0(Z),\Theta) = \Var(\Theta)$, hence we can expand
    \begin{align*}
        R(\tilde{\delta}) &= \Ebb\Big[|\tilde{\delta}(Z)-\Theta|^2\Big] \\
        &= \Var\Big(\tilde{\delta}(Z)\Big)+\Var(\Theta)-2 \Cov\Big(\tilde{\delta}(Z),\Theta\Big) \\
        &= \alpha^2\Var(\delta_0(Z))+\Var(\Theta)-2 \alpha\Cov(\delta_0(Z),\Theta) \\
        &= \alpha^2\Var(\delta_0(Z))+(1-2\alpha)\Var(\Theta)
    \end{align*}
    and
    \begin{equation*}
        R(\delta_0) = \Ebb\Big[|\delta_0(Z)-\Theta|^2\Big] = \Var(\delta_0(Z))-\Var(\Theta).
    \end{equation*}
    Thus, the claim reduces to
    \begin{equation*}
        \alpha^2\Var(\delta_0(Z))+(1-2\alpha)\Var(\Theta)\le  \Var(\delta_0(Z))-\Var(\Theta)
    \end{equation*}
    which is equivalent to $1+\alpha-2\alpha^2\ge 0$.
    But this holds for all $0<\alpha<1$, and we indeed have $0<\Var(\Theta)/\Var(\delta_0(Z))<1$ by the law of total variance.
\end{proof}

        \begin{proof}[Proof of Theorem~\ref{thm:gen-cstr-num}]
        First let us show that the integrand of \eqref{eqn:gen-str-Kant}  is well-defined for every $\pi\in\Gamma(F;G,\Psi)$.
        To see this, write $H$ for the second marginal of $\pi$,
        note that
        \begin{equation*}
            c_{G}(z,\eta) = \|\eta - \delta_\B(z)\|^2 \le 2\|\eta\|^2 + 2\|\delta_\B(z)\|^2.
        \end{equation*}
        Then, use the contraction property of conditional expectation to bound
        \begin{align*}
            \int_{\Rbb^d\times\Rbb^m}c_{G}(z,\eta)\diff\pi(z,\eta) &\le 2\int_{\Rbb^m}\|\eta\|^2\diff H(\eta) + 2\int_{\Rbb^d}\|\delta_\B(z)\|^2\diff F(z) \\
            &\le 2\int_{\Rbb^m}\|\eta\|^2\diff H(\eta)+ 2\int_{\Rbb^m}\|\eta\|^2\diff G(\eta).
        \end{align*}
        By \eqref{eqn:quadratic-growth}, we may get a function $\psi\in\Psi$ and some $C,R>0$ such that we have $\|\eta\|^2 \le  C\psi(\eta)$ for all $\eta\in\Rbb^m$ with $\|\eta\|>R$.
        Then we can bound the right side via
        \begin{align*}
            \int_{\Rbb^m}\|\eta\|^2\diff H(\eta) &=\int_{B_R(0)}\|\eta\|^2\diff H(\eta) + \int_{\Rbb^m\setminus B_R(0)}\|\eta\|^2\diff H(\eta) \\
            &\le\int_{B_R(0)}\|\eta\|^2\diff H(\eta) + C\int_{\Rbb^m\setminus B_R(0)}\psi(\eta)\diff H(\eta) \\
            &\le R^2+ C\int_{\Rbb^m\setminus B_R(0)}\psi(\eta)\diff H(\eta) \\
            &\le R^2 + C\int_{\Rbb^m}\psi\diff H,
        \end{align*}
        and similar for $G$.
        Since the right side above is finite by the assumption $\psi\in\Psi\subseteq L^1(G)$, it follows that the integral is well-defined for each feasible $\pi$.

        Next let us show that \eqref{eqn:gen-str-Kant} admits solutions.
        By \eqref{eqn:cty}, there exists a Borel set $A\subseteq\Rbb^d$ with $F(A) = 1$ such that $c_G:A\times \Rbb^m\to [0,\infty)$ is continuous.
        Since
        \begin{equation*}
            \int_{\Rbb^d\times\Rbb^m}c_G(z,\eta)\diff \pi(z,\eta) = \int_{A\times\Rbb^m}c_G(z,\eta)\diff \pi(z,\eta),
        \end{equation*}
        for all $\pi\in\Gamma(F;G)$,
        it follows from the Portmanteau lemma that the objective of \eqref{eqn:gen-str-Kant} is weakly lower semi-continuous.
        Since a lower semi-continuous function on a compact set must achieve a minimizer, it only remains to show that the feasible set $\Gamma(F;G,\Psi)$ is weakly compact.
        To do this, note that we can write $\Gamma(F;G,\Psi) = \bigcup_{H\in\Hcal}\Gamma(F;H)$ for
        \begin{equation*}
            \Hcal :=\left\{H\in\Pcal(\Rbb^m): \int_{\Rbb^m}\psi_{\ell}(\eta)\diff H(\eta) = \int_{\Rbb^m}\psi_{\ell}(\eta)\diff G(\eta) \textnormal{ for all } 1 \le \ell \le k\right\},
        \end{equation*}
        and it follows from \cite[Lemma~4.3]{Villani} that $\Gamma(F;G,\Psi)$ is relatively weakly compact in $\Pcal(\Rbb^d\times\Rbb^m)$ as soon as $\Hcal$ is relatively weakly compact in $\Pcal(\Rbb^m)$.
        Thus, it suffices (by Prokhorov's theorem) to show that $\Hcal$ is tight and weakly closed.
        For tightness, use \eqref{eqn:quadratic-growth} to get (as above) some $\psi\in\Psi$ and $C,R>0$ with $\|\eta\|^2\le C\psi(\eta)$ for all $\eta\in\Rbb^m$ with $\|\eta\|>R$.
        In particular, the compact sets $K_{r}:=\{\eta\in\Rbb^m:\psi(\eta)\le r\}$ for $r\ge 0$ satisfy, by Markov's inequality:
        \begin{equation*}
            H(\Rbb^m\setminus K_r) \le \frac{1}{r}\int_{\Rbb^m}\psi(\eta)\diff H(\eta) = \frac{1}{r}\int_{\Rbb^m}\psi(\eta)\diff G(\eta)
        \end{equation*}
        for all $H\in \mathcal{H}$ and $r\ge 0$, hence
        \begin{equation*}
            \sup_{H\in\mathcal{H}}H(\Rbb^m\setminus K_r) \le \frac{1}{r}\int_{\Rbb^m}\psi(\eta)\diff G(\eta)\to 0
        \end{equation*}
        as $r\to\infty$, as needed.
        For closedness, simply use~\eqref{eqn:UI}.        
        This proves that $\mathcal{H}$ is weakly compact, and hence guarantees the existence of a solution to \eqref{eqn:gen-str-Kant}.

        Now we finish the proof, by applying Theorem~\ref{thm:pop-distrcstr} and the remarks thereafter which show that, under assumptions~\eqref{eqn:Bayes-dens} and \eqref{eqn:quadratic-growth},  any solution to \eqref{eqn:gen-str-Kant} must be induced by a map.
        Indeed, it suffices to show that \eqref{eqn:gen-str-Kant} admits a unique solution, so, towards a contradiction, assume that there were distinct solutions $\pi_1,\pi_2$ to \eqref{eqn:gen-str-Kant}.
        Since the objective of \eqref{eqn:gen-str-Kant} is affine, it follows that $\frac{1}{2}(\pi_1+\pi_2)$ is also a solution.
        However, if $\pi_1,\pi_2$ are induced by maps, then $\frac{1}{2}(\pi_1+\pi_2)$ cannot be induced by a map.
        This is a contradiction, hence the proof is complete.
        \end{proof}

        \section{Proofs from Section~\ref{sec:emp}}\label{app:emp}

        This section contains the proofs of the main results in Section~\ref{sec:emp} on EB denoising.

        \subsection{Proofs from Subsection~\ref{subsec:emp-varcstr}}\label{app:emp-varcstr}

        We begin with the proofs of the results from Subsection~\ref{subsec:emp-varcstr} on variance-constrained EB denoising.
        In particular, we will prove Theorem~\ref{thm:varcstr-gm-npp} which provides a rate of convergence for the denoiser $\hat{\delta}_{\varcstr}$ given in Algorithm~\ref{alg:Gaussian-mod-NP-prior}.
        Adopting the notation from Subsection~\ref{subsec:emp-varcstr}, we begin with the following estimate. 

    \begin{lemma}\label{lem:emp-cov-bd}
        Under \eqref{eqn:FM} and \eqref{eqn:EBQ}, we have
        \begin{equation*}
            \|\hat M - M\|_{\textnormal{F}}^2 = O_{\Pbb}\left(\alpha_n\vee n^{-1}\right)
        \end{equation*}
        as $n\to\infty$.
    \end{lemma}

    \begin{proof}
        We begin by proving the following deterministic result:
        There exists a universal constant $C>0$ such that, for any $x_1,\ldots, x_n$ and $y_1,\ldots, y_n$ in $\Rbb^m$, we have
        \begin{equation}\label{eqn:emp-var-bd}
            \left\|\frac{1}{n}\sum_{i=1}^{n}(x_i-\bar x)(x_i-\bar x)^{\top} - \frac{1}{n}\sum_{i=1}^{n}(y_i-\bar y)(y_i-\bar y)^{\top}\right\|_{\textnormal{F}} \le CD(D+\tilde D),
        \end{equation}
        where we have defined $\bar x := \frac{1}{n}\sum_{i=1}^{n}x_i$ and $\bar y := \frac{1}{n}\sum_{i=1}^{n}y_i$, as well as
        \begin{equation*}
            D:=\sqrt{\frac{1}{n}\sum_{i=1}^{n}\|x_i-y_i\|^2}
        \end{equation*}
        and
        \begin{equation*}
            \tilde D:=\sqrt{\frac{1}{n}\sum_{i=1}^{n}\|y_i-\bar y\|^2}.
        \end{equation*}
        To see this, note that for all $u,v\in\Rbb^m$ we have the bound $\|uu^{\top}- vv^{\top}\|_{\textnormal{F}} \le \|u-v\|(\|u\| + \|v\|)$, hence we get
        \begin{align*}
            &\left\|\frac{1}{n}\sum_{i=1}^{n}(x_i-\bar x)(x_i-\bar x)^{\top} - \frac{1}{n}\sum_{i=1}^{n}(y_i-\bar y)(y_i-\bar y)^{\top}\right\|_{\textnormal{F}} \\
            &\le \frac{1}{n}\sum_{i=1}^{n}\left\|(x_i-\bar x)(x_i-\bar x)^{\top} - (y_i-\bar y)(y_i-\bar y)^{\top}\right\|_{\textnormal{F}} \\
            &\le \frac{1}{n}\sum_{i=1}^{n}\|x_i-\bar{x}-y_i+\bar{y}\|(\|x_i-\bar x\|+\|y_i-\bar y\|).
        \end{align*}
        By applying Cauchy-Schwarz to the right side and then squaring both sides of the resulting inequality, the display above yields
        \begin{align*}
            &\left\|\frac{1}{n}\sum_{i=1}^{n}(x_i-\bar x)(x_i-\bar x)^{\top} - \frac{1}{n}\sum_{i=1}^{n}(y_i-\bar y)(y_i-\bar y)^{\top}\right\|_{\textnormal{F}}^2 \\
            &\le \frac{1}{n}\sum_{i=1}^{n}\|x_i-\bar{x}-y_i+\bar{y}\|^2\cdot\frac{2}{n}\sum_{i=1}^{n}\left(\|x_i-\bar x\|^2+\|y_i-\bar y\|^2\right).
        \end{align*}
        In particular, it suffices to show that there is some universal constant $C>0$ such that we have
        \begin{equation*}
            \frac{1}{n}\sum_{i=1}^{n}\|x_i-\bar x\|^2 \le C\left(\frac{1}{n}\sum_{i=1}^{n}\|x_i-y_i\|^2 + \frac{1}{n}\sum_{i=1}^{n}\|y_i-\bar y\|^2\right)
        \end{equation*}
        To do this, we use the elementary inequality
        \begin{equation*}
            \|x_i-\bar x\|^2 \le 3\left(\|x_i-y_i\|^2 + \|y_i - \bar y \|^2 + \|\bar y - \bar x\|^2\right)
        \end{equation*}
        and the observation
        \begin{equation*}
            \|\bar x - \bar y\|^2 \le \frac{1}{n}\sum_{i=1}^{n}\|x_i-y_i\|^2
        \end{equation*}
        to conclude.

        Now we return to the main result.
        Set
        \begin{align*}
            x_i := \hat{\delta}_{\B}(Z_i) \qquad \mbox{and}\qquad y_i := \delta_{\B}(Z_i),
        \end{align*}
        in \eqref{eqn:emp-var-bd} and note that $\hat M$ is the sample covariance matrix of $\{x_i\}_{i=1}^n$ and $M$ is the population covariance matrix of $\{y_i\}_{i=1}^n$. Letting $\bar M$ denote the sample covariance matrix
        \begin{equation*}
            \bar M := \frac{1}{n}\sum_{i=1}^{n}y_iy_i^{\top} - \bar y\bar y^{\top},
        \end{equation*}
        we note that by the triangle inequality we have
        \begin{equation}\label{eqn:moment-contract}
            \|\hat M - M\|_{\textnormal{F}}^2\le 2\|\hat M - \bar M\|_{\textnormal{F}}^2+ 2\|\bar M - M\|_{\textnormal{F}}^2.
        \end{equation}
        Since $y_1,\ldots, y_n$ are i.i.d., we know from classical asymptotic statistics that the second term is $O_{\Pbb}(n^{-1})$ provided that we have $\Ebb[\|\delta_{\B}(Z)\|^{4}]<\infty$; to see this, simply use the contraction property of conditional expectation to compute:
        \begin{equation*}
            \Ebb\left[\|\delta_{\B}(Z)\|^{4}\right] = \Ebb\left[\|\Ebb[\Theta\,|\,Z]\|^{4}\right] \le \Ebb\left[\|\Theta\|^{4}\right] < \infty,
        \end{equation*}
        where the finiteness holds by \eqref{eqn:FM}.
        For the first term, note by \eqref{eqn:emp-var-bd} that it suffices to show $D^2 = O_{\Pbb}(\alpha_n)$ and $\tilde D^2 = O_{\Pbb}(1)$; the latter follows from classical asymptotic statistics under \eqref{eqn:FM}, and the former holds by assumption \eqref{eqn:EBQ}.
        This finishes the proof.
    \end{proof}

    Now we can prove the main result.

    \begin{proof}[Proof of Theorem~\ref{thm:varcstr-gm-npp}]
        Since all elements of the problem are equivariant under translation, we may assume $\mu:=\Ebb[Z] = \Ebb[\Theta]$ is $\mu=0$.
        Then expand
        \begin{align*}
            \hat \delta_{\varcstr}(Z_i)-\delta_{\varcstr}(Z_i) &= \left(\transport_{\hat M}^{\hat A}- \transport_{M}^{A}\right)\hat \delta_{\B}(Z_i) + \transport_{M}^{A}\left(\hat \delta_{\B}(Z_i) - \delta_{\B}(Z_i)\right) + (I-\transport_{\hat M}^{\hat A}) \hat \mu,
        \end{align*}
        and let us use this to bound
        \begin{equation}\label{eqn:G-NP-consistency-1}
            \begin{split}
                \frac{1}{n}\sum_{i=1}^{n}\left\|\hat \delta_{\varcstr}(Z_i)-\delta_{\varcstr}(Z_i)\right\|_2^2
                &\le\left\|\transport_{\hat M}^{\hat A}- \transport_{M}^{A}\right\|_{2\to 2}^2\frac{3}{n}\sum_{i=1}^{n}\left\|\hat \delta_{\B}(Z_i)\right\|_2^2 \\
                &+ \left\|\transport_{M}^{A}\right\|_{2\to 2}^2\frac{3}{n}\sum_{i=1}^{n}\left\|\hat \delta_{\B}(Z_i) - \delta_{\B}(Z_i)\right\|_2^2 \\
                &+ 3\left\|I-\transport_{\hat M}^{\hat A}\right\|_{2\to 2}^2 \|\hat \mu\|_2^2,
            \end{split}
        \end{equation}
        where $\|\cdot\|_{2\to 2}$ denotes the operator norm of a linear map from $(\Rbb^m,\|\cdot\|)$ to itself.
        In order to complete the proof, it suffices to show that the three terms on the right side are all $O_{\Pbb}(\alpha_n\vee n^{-1})$ as $n\to\infty$.
        We do this in reverse order.

        For the third term of \eqref{eqn:G-NP-consistency-1}, it suffices to show that $\|I-\transport_{\hat M}^{\hat A}\|_{2\to 2}^2 = O_{\Pbb}(1)$, since we have $\|\hat\mu\|_2^2 = O_{\Pbb}(n^{-1})$ by assumption~\eqref{eqn:VE}.
        For this, write $\transport:=\transport_{\hat M}^{\hat A}$ and recall that $\transport$ is symmetric positive semi-definite and satisfies $\transport\hat M\transport = \hat A$; thus, if $v$ is a unit eigenvector of $\transport$ associated with $\lambda_{\max}(\transport) = \|\transport\|_{2\to 2}$, then
        \begin{equation*}
            \|\transport\|_{2\to 2}^2\,\lambda_{\min}(\hat M) \le \|\transport\|_{2\to 2}^2\, v^{\top}\hat M v = v^{\top}\transport\hat M\transport v = v^{\top} \hat A v \le \|\hat A\|_{2\to 2},
        \end{equation*}
        whence $\|\transport_{\hat M}^{\hat A}\|_{2\to 2} \le (\|\hat A\|_{2\to 2}/\lambda_{\min}(\hat M))^{\sfrac{1}{2}}$.
        Since $\hat M \to M$ by Lemma~\ref{lem:emp-cov-bd} and $\hat A \to A$ by \eqref{eqn:VE}, and since $\lambda_{\min}(\cdot)$ and $\|\cdot\|_{2\to 2}$ are continuous, the right side above converges in probability to $(\|A\|_{2\to 2}/\lambda_{\min}(M))^{\sfrac{1}{2}}$, which is finite by \eqref{eqn:pos-def}.
        Therefore $\|I-\transport_{\hat M}^{\hat A}\|_{2\to 2}\le 1 + \|\transport_{\hat M}^{\hat A}\|_{2\to 2} = O_{\Pbb}(1)$, and the proof of this term is complete.

        For the second term, we simply note that the first factor is non-random and that the second factor is $O_{\Pbb}(\alpha_n)$ by assumption~\eqref{eqn:EBQ}.

        Lastly, we consider the first term of \eqref{eqn:G-NP-consistency-1}.
        Indeed, we note that it suffices to show
        \begin{equation}\label{eqn:G-NP-consistency-2}
            \left\|\transport_{\hat M}^{\hat A}- \transport_{M}^{A}\right\|_{2\to 2}^2 = O_{\Pbb}\left(\alpha_n\vee n^{-1}\right)
        \end{equation}
        and
        \begin{equation}\label{eqn:G-NP-consistency-3}
            \frac{1}{n}\sum_{i=1}^{n}\left\|\hat \delta_{\B}(Z_i)\right\|_2^2 = O_{\Pbb}(1)
        \end{equation}
        as $n\to\infty$, and we do this in two separate steps.

        To see \eqref{eqn:G-NP-consistency-2}, we use the triangle inequality to get
        \begin{equation*}
            \|\transport_{\hat M}^{\hat A} - \transport_{M}^{A}\|_{2\to 2} \le \|\transport_{\hat M}^{\hat A} - \transport_{M}^{\hat A}\|_{2\to 2} + \|\transport_{M}^{\hat A} - \transport_{M}^{A}\|_{2\to 2}.
        \end{equation*}
        Of course, all norms are equivalent up to constants which may depend only on the dimension, so it suffices to show that both
        \begin{equation*}
            \|\transport_{\hat M}^{\hat A} - \transport_{M}^{\hat A}\|_{\textnormal{F}} \qquad \textnormal{and}\qquad \|\transport_{M}^{\hat A} - \transport_{M}^{A}\|_{\textnormal{F}}
        \end{equation*}
        are $O_{\Pbb}(\alpha_n\vee n^{-1})$ as $n\to\infty$.
        To do this, we define $\phi^{B}:\covspace(m)\to \symspace(m)$ via $\phi^B(A):=\transport_{A}^{B}$, and note from \citet[Lemma~A.2]{Kroshnin} that $\phi^B$ is continuously differentiable and its Jacobian $\nabla_A\phi^{B}$ is jointly continuous in $(A,B)$.
        At the same time, the mean value theorem yields
        \begin{equation*}
            \|\transport_{\hat M}^{\hat A} - \transport_{M}^{\hat A}\|_{\textnormal{F}} \le \sup_{0\le t\le 1}\|\nabla_{(1-t)\hat M + tM}\phi^{\hat A}\|_{\textnormal{F}\to\textnormal{F}}\cdot\|\hat M - M\|_{\textnormal{F}},
        \end{equation*}
        where $\|\cdot\|_{\textnormal{F}\to\textnormal{F}}$ denotes the operator norm of a linear map from $(\Rbb^{m\times m},\|\cdot\|_{\textnormal{F}})$ to itself.
        Now fix $r>0$, and note that the ball $B_r^{\textnormal{F}}(B):= \{A\in \covspace(m): \|A-B\|_{\textnormal{F}} \le r\}$ is compact and convex for all $B\in\Rbb^{m\times m}$.
        In particular, we have
        \begin{equation*}
            C_1:=\sup_{\tilde{M}\in B^{\textnormal{F}}_{r}(M),\tilde{A}\in B^{\textnormal{F}}_{r}(A)}\|\nabla_{\tilde{M}}\phi^{\tilde{A}}\|_{\textnormal{F}\to\textnormal{F}} < \infty,
        \end{equation*}
        and
        \begin{equation*}
            \|\transport_{\hat M}^{\hat A} - \transport_{M}^{\hat A}\|_{\textnormal{F}} \le C_1\|\hat M - M\|_{\textnormal{F}},
        \end{equation*}
        on the event $\{\|\hat M - M\|_{\textnormal{F}}\le r,\|\hat A - A\|_{\textnormal{F}}\le r\}$.
        Additionally, note that we may define the map $\phi_A:\covspace(m)\to\symspace(m)$ via $\phi_A(B) := \transport_A^{B}$, and of course this is related to $\phi^B$ by the identity $\phi_A = (\phi^A)^{-1}$.
        Since inversion is differentiable on $\covspace(m)$, this implies that $\phi_A$ is also continuously differentiable.
        In particular, we have
        \begin{equation*}
            C_2:=\sup_{\tilde{A}\in B^{\textnormal{F}}_{r}(A)}\|\nabla_{\tilde{A}}\phi_{M}\|_{\textnormal{F}\to\textnormal{F}} < \infty,
        \end{equation*}
        and
        \begin{equation*}
            \|\transport_{M}^{\hat A} - \transport_{M}^{A}\|_{\textnormal{F}} \le C_2\|\hat A - A\|_{\textnormal{F}},
        \end{equation*}
        on the event $\{\|\hat A - A\|_{\textnormal{F}}\le r\}$.
        Since we have $\Pbb(\|\hat M - M\|_{\textnormal{F}}\le r, \|\hat A - A\|_{\textnormal{F}}\le r)\to 1$, we see that~\eqref{eqn:G-NP-consistency-2} follows from Lemma~\ref{lem:emp-cov-bd} and assumption \eqref{eqn:VE}.
        
        To see \eqref{eqn:G-NP-consistency-3}, we simply bound:
        \begin{equation}
            \frac{1}{n}\sum_{i=1}^{n}\left\|\hat \delta_{\B}(Z_i)\right\|_2^2 \le \frac{2}{n}\sum_{i=1}^{n}\left\|\delta_{\B}(Z_i)\right\|_2^2 + \frac{2}{n}\sum_{i=1}^{n}\left\|\hat \delta_{\B}(Z_i)-\delta_{\B}(Z_i)\right\|_2^2.
        \end{equation}
        Lastly, note that the first term is $O_{\Pbb}(1)$ by classical asymptotic statistics and \eqref{eqn:FM}, and the second term is $O_{\Pbb}(\alpha_n)$ by another application of~\eqref{eqn:EBQ}.
        This shows that the sum is $O_{\Pbb}(1)$, which establishes \eqref{eqn:G-NP-consistency-3}. This finishes the proof.
    \end{proof}

    \subsection{Proofs from Subsection~\ref{subsec:emp-distrcstr}}\label{app:emp-distrcstr}

    Next we turn to the results of Subsection~\ref{subsec:emp-distrcstr} on EB denoising with distribution constraints.
    In particular, we will prove Theorem~\ref{thm:emp-distrcstr} which provides a rate of convergence of the denoiser $\hat{\delta}_{\distrcstr}$ given in Algorithm~\ref{alg:distr-cstr}.

    \begin{proof}[Proof of Theorem~\ref{thm:emp-distrcstr}]
        We begin with a useful estimate, which is based on the proof of \citet[Theorem~2.1]{DebGhosalSen} and \citet[Theorem~7]{SlawskiSen}.
        To derive it, recall from Theorem~\ref{thm:pop-distrcstr} that $\delta_{\distrcstr}$ must be of the form $\nabla \phi \circ \delta_{\B}$ for some convex function $\phi:\Rbb^m\to\Rbb$.
        Consequently, we have $(\nabla \phi^{\ast})(\delta_{\distrcstr}(z)) = \delta_{\B}(z)$ holding for $F$-almost all $z\in\Rbb^d$, where $\phi^{\ast}$ is the Fenchel-Legendre dual of $\phi$.
        Also, let us define
        \begin{equation*}
            \tilde G_n := (\delta_{\distrcstr})_{\#}\bar F_n \qquad \textnormal{and}\qquad\bar \pi_n := (\id,\delta_{\distrcstr})_{\#}\bar F_n
        \end{equation*}
        for each $n\in\Nbb$ (recall that $\bar G_n:=\frac{1}{n}\sum_{i=1}^{n}\delta_{\Theta_i}$ and $\bar F_n:=\frac{1}{n}\sum_{i=1}^{n}\delta_{Z_i}$).
        Then, we can use the smoothness from \eqref{eqn:L} to compute:
        \begin{align*}
            &\int_{\Rbb^m}\phi^{\ast}(\eta)\,\diff \hat G_n(\eta) - \int_{\Rbb^m}\phi^{\ast}(\eta)\,\diff \tilde G_n(\eta) \\
            &= \int_{\Rbb^d}\int_{\Rbb^m}\phi^{\ast}(\eta)\,\diff \hat{\pi}_n(\eta\,|\,z)\,\diff \bar F_n(z) - \int_{\Rbb^m}\phi^{\ast}(\eta)\,\diff \tilde G_n(\eta) \\
            &\ge \int_{\Rbb^d}\phi^{\ast}(\hat \delta_{\distrcstr}(z))\,\diff \bar F_n(z) - \int_{\Rbb^d}\phi^{\ast}(\delta_{\distrcstr}(z))\,\diff \bar F_n(z) \\
            &= \int_{\Rbb^d}\left(\phi^{\ast}(\hat \delta_{\distrcstr}(z))-\phi^{\ast}(\delta_{\distrcstr}(z))\right)\diff \bar F_n(z) \\
            &\ge \int_{\Rbb^d}\left\langle(\nabla \phi^{\ast})(\delta_{\distrcstr}(z)),\hat \delta_{\distrcstr}(z)-\delta_{\distrcstr}(z)\right\rangle\diff \bar F_n(z) \\
            &\qquad\qquad\qquad + \frac{1}{2L}\int_{\Rbb^d}\left\|\hat \delta_{\distrcstr}(z)-\delta_{\distrcstr}(z)\right\|^2\diff \bar F_n(z) \\
            &= \int_{\Rbb^d}\left\langle\delta_{\B}(z),\hat \delta_{\distrcstr}(z)-\delta_{\distrcstr}(z)\right\rangle\diff \bar F_n(z) \\
            &\qquad\qquad\qquad + \frac{1}{2L}\int_{\Rbb^d}\left\|\hat \delta_{\distrcstr}(z)-\delta_{\distrcstr}(z)\right\|^2\diff \bar F_n(z)
        \end{align*}
        Here, the first inequality follows from Jensen and the second inequality follows from assumption~\eqref{eqn:L}.
        Rearranging the above, we have shown
        \begin{eqnarray}\label{eqn:gen-cstr-1}
                \frac{1}{2L}\int_{\Rbb^d}\left\|\hat \delta_{\distrcstr}(z)-\delta_{\distrcstr}(z)\right\|^2\,\diff \bar F_n(z) & \le & \int_{\Rbb^m}\phi^{\ast}(\eta)\diff \hat G_n(\eta) - \int_{\Rbb^m}\phi^{\ast}(\eta)\diff \tilde G_n(\eta) \nonumber \\
                && \;\; -\int_{\Rbb^d}\left\langle\delta_{\B}(z),\hat \delta_{\distrcstr}(z)-\delta_{\distrcstr}(z)\right\rangle\diff \bar F_n(z). \qquad
        \end{eqnarray}
        Now we make two computations.
        First, we get
        \begin{eqnarray}\label{eqn:distr-cstr-2}
                \int_{\Rbb^d\times\Rbb^m}\hat c_n\diff \hat{\pi}_n &= & \int_{\Rbb^d\times\Rbb^m}\|\hat{\delta}_{\B}(z)-\eta\|^2\diff \hat{\pi}_n(z,\eta) \nonumber \\
                &=& \int_{\Rbb^d}\|\hat{\delta}_{\B}(z)\|^2\diff \bar F_n(z) + \int_{\Rbb^m}\|\eta\|^2\diff \hat G_n(\eta) - 2\int_{\Rbb^d}\int_{\Rbb^m}\left\langle\eta,\hat{\delta}_{\B}(z)\right\rangle\diff \hat{\pi}_n(\eta\,|\,z)\diff \bar F_n(z) \nonumber \\
                &=& \int_{\Rbb^d}\|\hat{\delta}_{\B}(z)\|^2\diff \bar F_n(z) + \int_{\Rbb^m}\|\eta\|^2\diff \hat G_n(\eta)  - 2\int_{\Rbb^d}\left\langle\hat{\delta}_{\distrcstr}(z),\hat{\delta}_{\B}(z)\right\rangle\diff \bar F_n(z) 
        \end{eqnarray}
        from the fact that we have $\hat{\delta}_{\distrcstr}(z) = \int \eta \,\diff \hat{\pi}_n(\eta\,|\,z)$ by construction.
        Second, we get
        \begin{eqnarray}\label{eqn:distr-cstr-3}
                \int_{\Rbb^d\times\Rbb^m} c_G\diff \bar{\pi}_n  &= & \int_{\Rbb^d\times\Rbb^m}\|\delta_{\B}(z)-\eta\|^2\diff \bar \pi_n(z,\eta) \nonumber \\
                &= & \int_{\Rbb^d}\|\delta_{\B}(z)\|^2\diff \bar F_n(z) + \int_{\Rbb^m}\|\eta\|^2\diff \tilde G_n(\eta)  
                - 2\int_{\Rbb^d}\left\langle\delta_{\B}(z),\delta_{\distrcstr}(z)\right\rangle\diff \bar F_n(z), \quad \qquad 
        \end{eqnarray}
         from the fact that we have ${\delta}_{\distrcstr}(z) = \int \eta \,\diff \bar{\pi}_n(\eta\,|\,z)$ holding $\bar F_n$ almost surely, which requires a small argument:
         defining $\bar{\delta}_{\distrcstr}(z) := \int \eta \,\diff \bar{\pi}_n(\eta\,|\,z)$, we may use the contraction property of conditional expectation to compute:
         \begin{align*}
             &\int_{\Rbb^d}\|\bar{\delta}_{\distrcstr}(z)-\delta_{\distrcstr}(z)\|^2\diff \bar F_n(z) \\
             &= \int_{\Rbb^d}\|\bar{\delta}_{\distrcstr}(z)\|^2\diff \bar F_n(z) +\int_{\Rbb^d}\|\delta_{\distrcstr}(z)\|^2\diff \bar F_n(z) -2\int_{\Rbb^d}\langle\delta_{\distrcstr}(z),\bar\delta_{\distrcstr}(z)\rangle\diff \bar F_n(z) \\
            &= \int_{\Rbb^d}\|\bar{\delta}_{\distrcstr}(z)\|^2\diff \bar \pi_n(z,\eta) +\int_{\Rbb^d}\|\delta_{\distrcstr}(z)\|^2\diff \bar \pi_n(z,\eta) -2\int_{\Rbb^d}\langle\delta_{\distrcstr}(z),\bar\delta_{\distrcstr}(z)\rangle\diff \bar \pi_n(z,\eta) \\
             &\le \int_{\Rbb^d}\|\eta\|^2\diff \bar \pi_n(z)+\int_{\Rbb^d}\|\delta_{\distrcstr}(z)\|^2\diff \bar \pi_n(z) -2\int_{\Rbb^d}\langle\delta_{\distrcstr}(z),\eta\rangle\diff \bar \pi_n(z) \\
             &= \int_{\Rbb^d}\|\eta-\delta_{\distrcstr}(z)\|^2\diff \bar \pi_n(z) \\
             &= 0.
         \end{align*}
         Thus, equations~\eqref{eqn:distr-cstr-3} and \eqref{eqn:distr-cstr-2} are both true.
        By subtracting \eqref{eqn:distr-cstr-3} from \eqref{eqn:distr-cstr-2}, and plugging the result into \eqref{eqn:gen-cstr-1} and rearranging, we conclude:
        \begin{equation}\label{eqn:gen-cstr-4}
            \begin{split}
                \frac{1}{2L}&\int_{\Rbb^d}\left\|\hat \delta_{\distrcstr}(z)-\delta_{\distrcstr}(z)\right\|^2\diff \bar F_n(z) \\
                &\le\int_{\Rbb^m}\left(\phi^{\ast}(\eta)-\frac{1}{2}\|\eta\|^2\right)\diff \hat G_n(\eta) - \int_{\Rbb^m}\left(\phi^{\ast}(\eta)-\frac{1}{2}\|\eta\|^2\right)\diff \tilde G_n(\eta) \\
                & \qquad - 2\int_{\Rbb^d}\left\langle\delta_{\B}(z)-\hat{\delta}_{\B}(z),\delta_{\distrcstr}(z)\right\rangle\diff \bar F_n(z) +\int_{\Rbb^d}\left\|\delta_{\B}(z)\right\|^2\diff\bar F_n(z) - \int_{\Rbb^d}\big\|\hat{\delta}_{\B}(z)\big\|^2\diff\bar F_n(z) \\
                &\qquad +\frac{1}{2}\left(\int_{\Rbb^d\times\Rbb^m}\hat c_n\diff \hat{\pi}_n - \int_{\Rbb^d\times\Rbb^m}c_G\diff \hat{\pi}_n\right)  + \frac{1}{2}\left(\int_{\Rbb^d\times\Rbb^m}c_G\diff \hat{\pi}_n -\int_{\Rbb^d\times\Rbb^m} c_G\diff \bar{\pi}_n\right).
            \end{split}
        \end{equation}
        Next, let $\rho_n$ be an optimal coupling of $\bar G_n$ and $\hat G_n$, and define $\lambda_n := (\nabla \phi^{\ast},\id)_{\#}\rho_n$, so that $\lambda_n$ is a coupling (not necessarily optimal) of $(\delta_{\B})_{\#}\bar F_n$ and $\hat{G}_n$.
        Then we can compute:
        \begin{equation}\label{eqn:gen-cstr-5}
        \begin{split}
            W_2^2\left((\delta_{\B})_{\#}\bar F_n,\hat{G}_n\right) &\le \int_{\Rbb^m\times\Rbb^m}\|\eta-\hat{\eta}\|^2\diff \lambda_n(\eta,\hat{\eta}) \\
            &= \int_{\Rbb^m\times\Rbb^m}\|\nabla \phi^{\ast}(\bar\eta)-\hat{\eta}\|^2\diff \rho_n(\bar{\eta},\hat{\eta}) \\
            &= \int_{\Rbb^m\times\Rbb^m}\|\nabla \phi^{\ast}(\bar\eta)-\bar{\eta}\|^2\diff \tilde{G}_n(\bar{\eta}) + \int_{\Rbb^m\times\Rbb^m}\|\bar\eta-\hat{\eta}\|^2\diff \rho_n(\bar{\eta},\hat{\eta}) \\
            &\qquad + 2\int_{\Rbb^m\times\Rbb^m}\left\langle\nabla\phi^{\ast}(\bar\eta)-\bar{\eta}, \bar\eta-\hat{\eta}\right\rangle\diff \rho_n(\bar{\eta},\hat{\eta}).
            \end{split}
        \end{equation}
        Note that the first term on the right side is exactly
        \begin{equation*}
            \int_{\Rbb^m}\|\nabla \phi^{\ast}(\bar\eta)-\bar{\eta}\|^2\diff \tilde{G}_n(\bar{\eta}) = \int_{\Rbb^m}\|\delta_{\B}(z)-\nabla\phi(\delta_{\B}(z))\|^2\diff \bar{F}_n(z) = W_2^2((\delta_{\B})_{\#}\bar F_n,\tilde G_n)
        \end{equation*}
        because $\nabla \phi$ is an optimal transport map, and that the second term is $W_2^2(\tilde{G}_n,\hat{G}_n)$ by definition.
        Therefore, \eqref{eqn:gen-cstr-5} becomes
        \begin{equation}\label{eqn:gen-cstr-6}
        \begin{split}
            W_2^2\left((\delta_{\B})_{\#}\bar F_n,\hat{G}_n\right) &\le \int_{\Rbb^m\times\Rbb^m}\|\eta-\hat{\eta}\|^2\diff \lambda_n(\eta,\hat{\eta}) \\
            &=W_2^2\left((\delta_{\B})_{\#}\bar{F}_n,\bar{G}_n\right) + W_2^2(\tilde{G}_n,\hat{G}_n) \\
            & \qquad + 2\int_{\Rbb^m\times\Rbb^m}\left\langle\nabla\phi^{\ast}(\bar\eta)-\bar{\eta}, \bar\eta-\hat{\eta}\right\rangle\diff \rho_n(\bar{\eta},\hat{\eta}).
            \end{split}
        \end{equation}
        For the cross term, expand the square and use the strong convexity from \eqref{eqn:L} to get
        \begin{equation}\label{eqn:gen-cstr-7}
            \begin{split}
                &2\int_{\Rbb^m\times\Rbb^m}\left\langle\nabla\phi^{\ast}(\bar \eta),\bar{\eta}-\hat{\eta}\right\rangle \diff\rho_n(\bar{\eta},\hat{\eta}) \\
                &\le2\int_{\Rbb^m\times\Rbb^m}\left(\phi^{\ast}(\bar{\eta})-\phi^{\ast}(\hat{\eta}) + \frac{1}{2\lambda}\|\bar{\eta}-\hat{\eta}\|^2\right)\diff\rho_n(\bar{\eta},\hat{\eta}) \\
                &=2\int_{\Rbb^m}\phi^{\ast}(\bar{\eta})\diff\tilde{G}_n(\bar{\eta}) - 2\int_{\Rbb^m}\phi^{\ast}(\hat{\eta})\diff\hat{G}_n(\bar{\eta}) + \frac{1}{\lambda}W_2^2(\tilde{G}_n,\hat{G}_n)
            \end{split}
        \end{equation}
        and
        \begin{equation}\label{eqn:gen-cstr-8}
            \begin{split}
            &2\int_{\Rbb^m\times \Rbb^m}\left\langle -\bar{\eta},\bar{\eta}-\hat{\eta}\right\rangle \diff \rho_n(\bar{\eta},\hat{\eta}) \\
            &=\int_{\Rbb^m\times \Rbb^m}\left(\|\hat{\eta}\|^2 - \|\hat{\eta}-\bar{\eta}\|^2-\|  \bar{\eta}\|^2\right)\diff \rho_n(\bar{\eta},\hat{\eta}) \\
            &=\int_{\Rbb^m}\|\hat{\eta}\|^2\diff \hat{G}_n(\hat{\eta}) - \int_{\Rbb^m}\|\bar{\eta}\|^2\diff\tilde{G}_n(\bar{\eta}) - W_2^2(\tilde{G}_n,\hat{G}_n),
            \end{split}
        \end{equation}
        and then plug in \eqref{eqn:gen-cstr-7} and \eqref{eqn:gen-cstr-8} into \eqref{eqn:gen-cstr-6}, yielding
        \begin{equation}\label{eqn:gen-cstr-9}
            \begin{split}
            W_2^2\left((\delta_{\B})_{\#}\bar{F}_n,\hat{G}_n\right) &\le W_2^2\left((\delta_{\B})_{\#}\bar{F}_n, \tilde{G}_n\right) + \frac{1}{\lambda}W_2^2(\tilde{G}_n,\hat{G}_n) \\
            &\quad + \int_{\Rbb^m}\left(2\phi^{\ast}(\bar{\eta})-\|\bar{\eta}\|^2\right)\diff\tilde{G}_n(\bar{\eta}) - \int_{\Rbb^m}\left(2\phi^{\ast}(\hat{\eta})-\|\hat{\eta}\|^2\right)\diff\hat{G}_n(\bar{\eta}).
            \end{split}
        \end{equation}
        Since
        \begin{equation*}
            \int_{\Rbb^d\times \Rbb^m}c_G\diff \hat{\pi}_n = W_2^2\left((\delta_{\B})_{\#}\bar F_n,\hat{G}_n\right)\qquad \textnormal{ and }\qquad \int_{\Rbb^d\times \Rbb^m}c_G\diff \bar{\pi}_n = W_2^2\left((\delta_{\B})_{\#}\bar F_n,\tilde{G}_n\right),
        \end{equation*}
        we may plug in~\eqref{eqn:gen-cstr-9} into~\eqref{eqn:gen-cstr-4} to get:
        \begin{equation*}
            \begin{split}
                \frac{1}{2L}&\int_{\Rbb^d}\left\|\hat \delta_{\distrcstr}(z)-\delta_{\distrcstr}(z)\right\|^2\diff \bar F_n(z) \\
                &\hspace{-0.3in} \le \; \frac{1}{2\lambda}W_2^2(\hat{G}_n,\tilde{G}_n) - 2\int_{\Rbb^d}\left\langle\delta_{\B}(z)-\hat{\delta}_{\B}(z),\delta_{\distrcstr}(z)\right\rangle\diff \bar F_n(z)\\
                &+\int_{\Rbb^d}\left\|\delta_{\B}(z)\right\|^2\diff\bar F_n(z) - \int_{\Rbb^d}\left\|\hat{\delta}_{\B}(z)\right\|^2\diff\bar F_n(z)  +\frac{1}{2}\left(\int_{\Rbb^d\times\Rbb^m}\hat c_n\diff \hat{\pi}_n - \int_{\Rbb^d\times\Rbb^m}c_G\diff \hat{\pi}_n\right). \\
            \end{split}
        \end{equation*}
        Now, expand the squares in the definitions of the last two terms to get:
        \begin{equation}\label{eqn:gen-cstr-10}
            \begin{split}
                \frac{1}{2L}\int_{\Rbb^d}\left\|\hat \delta_{\distrcstr}(z)-\delta_{\distrcstr}(z)\right\|^2\diff \bar F_n(z) &\le \frac{1}{2\lambda}W_2^2(\hat{G}_n,\tilde{G}_n) \\
                &- \int_{\Rbb^d}\left\langle\delta_{\B}(z)-\hat{\delta}_{\B}(z),\delta_{\distrcstr}(z)\right\rangle\diff \bar F_n(z)\\
                &+\frac{1}{2}\int_{\Rbb^d}\left\|\delta_{\B}(z)\right\|^2\diff\bar F_n(z) - \frac{1}{2}\int_{\Rbb^d}\left\|\hat{\delta}_{\B}(z)\right\|^2\diff\bar F_n(z).
            \end{split}
        \end{equation}
        It is useful to compare \eqref{eqn:gen-cstr-10} with \citet[Theorem~7]{SlawskiSen} which derives a similar bound, but with only the first term;
        the additional terms are due to the fact that $\hat{\delta}_{\distrcstr}$ comes from an optimal transport problem whose cost function has been estimated.
        
        In the rest of the proof, we will show that each of the terms on the right side of~\eqref{eqn:gen-cstr-10} is $O_{\Pbb}(\alpha_n^{\sfrac{1}{2}}\vee \beta_n\vee \gamma_n)$.
        The first term is the easiest to handle since by the triangle inequality we can bound
        \begin{equation*}
            W_2^2(\hat G_n,\tilde G_n) \le 2W_2^2(\hat G_n,G)+2W_2^2(\tilde G_n,G),
        \end{equation*}
        and note by assumption~\eqref{eqn:CM} and the fact that $\tilde{G}_n$ and $\bar{G}_n$ have the same distribution that these terms on the right side are $O_{\Pbb}(\beta_n)$ and $O_{\Pbb}(\gamma_n)$, respectively.
        For the second term, use Cauchy-Schwarz to bound
        \begin{align*}
            &\int_{\Rbb^d}\left\langle\delta_{\B}(z)-\hat{\delta}_{\B}(z),\delta_{\distrcstr}(z)\right\rangle\diff \bar F_n(z)\\
            &\le\sqrt{\frac{1}{n}\sum_{i=1}^{n}\|\delta_{\B}(Z_i)-\hat{\delta}_{\B}(Z_i)\|^2}\sqrt{\frac{1}{n}\sum_{i=1}^{n}\|\delta_{\distrcstr}(Z_i)\|^2}.
        \end{align*}
        and note that the first factor is $O_{\Pbb}(\alpha_n^{\sfrac{1}{2}})$ by \eqref{eqn:EBQ}, and the second factor is $O_{\Pbb}(1)$ because of \eqref{eqn:2M} and the law of large numbers.
        For the third term, recall the elementary inequality
        \begin{equation*}
            \left|\|x\|^2 - \|y\|^2\right| \le \|x-y\|^2 + 2\|x\|\cdot\|x-y\|
        \end{equation*}
        for all $x,y\in\Rbb^m$.
        Taking $x=\delta_{\B}(Z_i)$ and $y=\hat{\delta}_{\B}(Z_i)$, summing, and applying Cauchy-Schwarz yields
        \begin{align*}
        &\int_{\Rbb^d}\|\delta_{\B}(Z_i)\|^2 \diff \bar F_n(z)-\int_{\Rbb^d}\|\hat{\delta}_{\B}(Z_i)\|^2 \diff \bar F_n(z)\\
        &=\sum_{i=1}^{n}\left(\|\delta_{\B}(Z_i)\|^2-\|\hat{\delta}_{\B}(Z_i)\|^2 \right)\\
        &\le \frac{1}{n}\sum_{i=1}^{n}\left(\|\delta_{\B}(Z_i)-\hat{\delta}_{\B}(Z_i)\|^2 + 2\|\delta_{\B}(Z_i)\|\cdot\|\delta_{\B}(Z_i)-\hat{\delta}_{\B}(Z_i)\|\right) \\
        &= \frac{1}{n}\sum_{i=1}^{n}\|\delta_{\B}(Z_i)-\hat{\delta}_{\B}(Z_i)\|^2 + \frac{2}{n}\sum_{i=1}^{n}\|\delta_{\B}(Z_i)\|\cdot\|\delta_{\B}(Z_i)-\hat{\delta}_{\B}(Z_i)\| \\
        &\le \frac{1}{n}\sum_{i=1}^{n}\|\delta_{\B}(Z_i)-\hat{\delta}_{\B}(Z_i)\|^2 + 2\cdot\sqrt{\frac{1}{n}\sum_{i=1}^{n}\|\delta_{\B}(Z_i)-\hat{\delta}_{\B}(Z_i)\|^2}\sqrt{\frac{1}{n}\sum_{i=1}^{n}\|\delta_{\B}(Z_i)\|^2}.
        \end{align*}
        Note that the first term on the right is $O_{\Pbb}(\alpha_n)$ by assumption~\eqref{eqn:EBQ}.
        Within the second term, the first factor is $O_{\Pbb}(\alpha_n^{\sfrac{1}{2}})$ by assumption~\eqref{eqn:EBQ}, and the second factor is $O_{\Pbb}(1)$ by assumption~\eqref{eqn:2M} and the law of large numbers.
    \end{proof}

    \subsection{Proofs from Subsection~\ref{subsec:emp-gencstr}}

    Lastly, we give the proof of Theorem~\ref{thm:gen-cstr-emp} on the rate of convergence for the general-constrained EB denoiser $\hat{\delta}_{\gencstr}$ from Algorithm~\ref{alg:distr-cstr}.

    An important step in our proof is the following technical result which states that limiting feasible couplings can be approximated by prelimiting feasible couplings.
    To state it, let us write $\Gamma_K(F;G,\Psi)$ for the set of $\pi\in \Gamma(F;G,\Psi)$ satisfying $\Supp(\pi(\Rbb^d\times \diff\eta))\subseteq K$, and similar for $\Gamma_K(\bar{F}_n;\hat{G}_n,\Psi)$.
    \begin{lemma}\label{lem:prelim-cvg-gen}
        Under assumptions \eqref{eqn:CM}, \eqref{eqn:CS}, \eqref{eqn:int} and \eqref{eqn:CC}, for every $\pi\in\Gamma_K(F;G,\Psi)$ there exists $\pi_n\in \Gamma_K(\bar F_n; \hat G_n,\Psi)$ for each $n\in\Nbb$ such that $W_2(\pi_n,\pi)=o_{\Pbb}(1)$.
    \end{lemma}

    \begin{proof}
        First, take $\pi\in\Gamma_K(F;G,\Psi)$ and get $H_0\in\Pcal(\Rbb^m)$ such that we have $\pi\in\Gamma(F;H_0)$, and let us construct a suitable exponential family from these pieces; this will require some basic notions which can be found in \citet[Chapter~3]{WainwrightJordan}.
        That is, let us define $A:\Rbb^k\to\Rbb\cup\{\infty\}$ via
        \begin{equation*}
            A(\alpha):= \log \int_{\Rbb^m}\exp\left(\sum_{\ell=1}^{k}\alpha_{\ell}\psi_{\ell}(\eta)\right)\diff H_0(\eta),
        \end{equation*}
        and, for each $\alpha\in D(H_0,\Psi)$ (the natural parameter space for this exponential family, i.e., the set of $\alpha$ for which $A(\alpha)<\infty$) we define the probability measure $H_{\alpha}\in \Pcal(\Rbb^m)$ via its Radon-Nikodym derivative
        \begin{equation*}
            \frac{\diff H_{\alpha}}{\diff H_0}(\eta) = \exp\left(\sum_{\ell=1}^{k}\alpha_{\ell}\psi_{\ell}(\eta)-A(\alpha)\right).
        \end{equation*}
        It follows that $\{H_{\alpha}\}_{\alpha}$ is an exponential family;
        in fact, the linear independence of $\psi_1,\ldots, \psi_k$ implies that $\{H_{\alpha}\}_{\alpha}$ is minimal (but not necessarily complete).
        Also, observe that $H_0$ indeed corresponds to the member of this family with $\alpha=0$.
    By assumption \eqref{eqn:CC} and \eqref{eqn:int}, we have
    \begin{equation*}
        \begin{pmatrix}
            \int_{\Rbb^m}\psi_1\diff \hat G_n \\ \vdots \\\int_{\Rbb^m}\psi_k\diff \hat G_n
        \end{pmatrix} \to \begin{pmatrix}
            \int_{\Rbb^m}\psi_1\diff G \\ \vdots \\\int_{\Rbb^m}\psi_k\diff G
        \end{pmatrix} \in\mathcal{M}^{\circ}(\Psi)
    \end{equation*}
    as $n\to\infty$, so it follows  that we have
    \begin{equation*}
        \begin{pmatrix}
            \int_{\Rbb^m}\psi_1\diff \hat G_n \\ \vdots \\\int_{\Rbb^m}\psi_k\diff \hat G_n
        \end{pmatrix} \in \mathcal{M}^{\circ}(\Psi)
    \end{equation*}
    for all sufficiently large $n\in\Nbb$.
    This implies that, for such $n\in\Nbb$, the parameter
    \begin{equation*}
        \alpha_n := \nabla A^{\ast}\left(\begin{pmatrix}
            \int_{\Rbb^m}\psi_1\diff \hat G_n \\ \vdots \\\int_{\Rbb^m}\psi_k\diff \hat G_n
        \end{pmatrix}\right)
    \end{equation*}
    satisfies
    \begin{equation*}
        \begin{pmatrix}
            \int_{\Rbb^m}\psi_1\diff H_{\alpha_n} \\ \vdots \\\int_{\Rbb^m}\psi_k\diff H_{\alpha_n}
        \end{pmatrix} = \begin{pmatrix}
            \int_{\Rbb^m}\psi_1\diff \hat G_n \\ \vdots \\\int_{\Rbb^m}\psi_k\diff \hat G_n
        \end{pmatrix}.
    \end{equation*}
    Of course, we also have
    \begin{equation*}
        \begin{pmatrix}
            \int_{\Rbb^m}\psi_1\diff H_{0} \\ \vdots \\\int_{\Rbb^m}\psi_k\diff H_{0}
        \end{pmatrix} = \begin{pmatrix}
            \int_{\Rbb^m}\psi_1\diff G \\ \vdots \\\int_{\Rbb^m}\psi_k\diff G
        \end{pmatrix}
    \end{equation*}
    by construction.
    Since $A$ being strictly convex implies that $A^{\ast}$ is continuously differentiable, we conclude $\alpha_n\to 0$ as $n\to\infty$.
    Also, the relative entropy of $H_{\alpha_n}$ from $H_0$ is just the $A$-Bregman divergence of $\alpha_n$ from $0$, so combining $\alpha_n\to 0$ and Pinsker's inequality shows that we have $\|H_{\alpha_n}- H_{0}\|_{\TV} \to 0$ as $n\to\infty$.
    Since $\{H_{\alpha}\}_{\alpha}$ have common compact support by \eqref{eqn:CS}, this further implies $W_2(H_{\alpha_n},H_0)\to 0$.
    Now use $W_2(\bar F_n,F)\to 0$ and the gluing lemma \citet[p. 11]{Villani} to construct $\pi_n\in\Gamma(\bar F_n;H_{\alpha_n})$ for all $n\in\Nbb$ such that we have $W_2(\pi_n,\pi)\to 0$.
    By construction, this means we have $\pi_n\in\Gamma_K(\bar F_n;\hat G_n,\Psi)$, so the proof is complete.
    \end{proof}

    \begin{proof}[Proof of Theorem~\ref{thm:gen-cstr-emp}]
        We claim (as stated in the discussion in Section~\ref{subsec:emp-gencstr}) that we have $W_2(\hat{H}_n,H)\to 0$ in probability as $n\to\infty$.
        If this is true, then we may apply the exact same argument from the proof of Theorem~\ref{thm:emp-distrcstr} so we are done.

        In fact, we will prove a slightly stronger claim.
        For convenience, write $\hat{\pi}_n$ for $\hat{\pi}_{\gencstr}$, in order to emphasize the dependence on $n\in\Nbb$.
        Then, we claim that we have $W_2(\hat{\pi}_{n},\pi_{\gencstr})\to 0$ in probability as $n\to\infty$.
        Note that the assumptions of Theorem~\ref{thm:gen-cstr-num} are satisfied, so $\pi_{\gencstr}$ is the unique solution to \eqref{eqn:pop-distrcstr-Monge}; thus, it suffices to show that $\{\hat \pi_{n}\}_{n\in\Nbb}$ is $W_2$-precompact and that every subsequential $W_2$-limit is a solution to \eqref{eqn:pop-distrcstr-Monge}.
        To see pre-compactness, simply note that the marginal distributions of $\{\hat \pi_n\}_{n\in\Nbb}$ are $\{\bar F_n\}_{n\in\Nbb}$ and $\{\hat H_n\}_{n\in\Nbb}$; the former converges in $W_2$ to $F$ hence it is $W_2$-pre-compact, and the latter is $W_2$-pre-compact by combining~\eqref{eqn:quadratic-growth} and~\eqref{eqn:UI-emp} with the fact that pre-compactness in $W_2$ is equivalent to tightness plus uniform integrability of the second moment..
    
        Next, suppose that $\{n_k\}_{k\in\Nbb}$ and $\pi_{\infty}\in\Pcal(\Rbb^d\times\Rbb^m)$ have $W_2(\hat{\pi}_{\gencstr,n_k},\pi_{\infty})\to 0$, and let us show that $\pi_{\infty}$ is optimal for \eqref{eqn:pop-distrcstr-Monge}.
        To see that $\pi_{\infty}$ is feasible for \eqref{eqn:pop-distrcstr-Monge}, take arbitrary $1\le\ell\le k$.
        Then, use~\eqref{eqn:UI-emp}, the fact that $\hat{H}_n\in\Gamma(\bar F_{n};\hat{G}_n,\Psi)$ for all $n\in\Nbb$, and the fact that $W_2$ metrizes weak convergence (recall that $\psi_1,\ldots, \psi_k$ are assumed continuous) to get:

        \begin{align*}
        \int_{\Rbb^d\times K}\psi_{\ell}(\eta)\diff\pi_{\infty}(z,\eta) &= \lim_{n\to\infty}\int_{\Rbb^d\times K}\psi_{\ell}(\eta)\diff\hat{\pi}_{n}(z,\eta) \\
        &= \lim_{n\to\infty}\int_{K}\psi_{\ell}(\eta)\diff\hat{G}_{n}(\eta) = \int_{K}\psi_{\ell}(\eta)\diff G(\eta).
        \end{align*}
        Also, the first coordinate of $\pi_{\infty}$ is the $W_2$-limit of the first coordinates $\bar{F}_n$ of $\hat{\pi}_n$, which is $F$.
        Thus, $\pi_{\infty}$ is feasible.
        
        Next, we develop a useful estimate.
        We claim that any $\{\pi_n\}_{n\in\Nbb}$ with $\pi_n\in\Gamma(\bar F_n;\hat{H}_n)$ must satisfy
        \begin{equation}\label{eqn:integral-estimate-H}
            \left|\int_{\Rbb^d\times K}c_G\,\diff \pi_{n}-\int_{\Rbb^d\times K}\hat c_{n}\diff \pi_{n}\right| \to 0.
        \end{equation}
        as $n\to\infty$.
        To show this, we calculate:
        \begin{align*}
            &\left|\int_{\Rbb^d\times K}c_{G}\diff \pi_{n}-\int_{\Rbb^d\times K}\hat c_{n}\diff \pi_{n}\right| \\
            &=\left|\int_{\Rbb^d\times K}\left(\|\delta_{\B}(z)-\eta\|^2-\|\hat \delta_{\B}(z)-\eta\|^2\right)\diff  \pi_{n}(z,\eta)\right| \\
            &=\left|\int_{\Rbb^d\times K}\left(\|\delta_{\B}(z)\|^2-\|\hat{\delta}_{\B}(z)\|^2-2\left\langle\delta_{\B}(z)-\hat{\delta}_{\B}(z),\eta\right\rangle\right)\diff \pi_{n}(z,\eta)\right| \\
            &\le\left|\int_{\Rbb^d\times K}\left(\|\delta_{\B}(z)\|^2-\|\hat{\delta}_{\B}(z)\|^2\right)\diff \pi_{n}(z,\eta)\right| + 2\left|\int_{\Rbb^d\times K}\left\langle\delta_{\B}(z)-\hat{\delta}_{\B}(z),\eta\right\rangle\diff \pi_{n}(z,\eta)\right| \\
            &\le \int_{\Rbb^d\times K}\left|\|\delta_{\B}(z)\|^2-\|\hat{\delta}_{\B}(z)\|^2\right|\diff \pi_{n}(z,\eta) + 2\int_{\Rbb^d\times K}\left|\left\langle\delta_{\B}(z)-\hat{\delta}_{\B}(z),\eta\right\rangle\right|\diff \pi_{n}(z,\eta),
        \end{align*}
        and we separately further the bound on each of these terms.
        For the first term, we notice that the integral only depends on the marginal distribution of $z$ under $\pi_{n}$, which equals $\bar F_n$ by construction.
        Also, we have the elementary inequality
        \begin{equation*}
            \left|\|x\|^2 - \|y\|^2\right| \le \|x-y\|^2 + 2\|x\|\cdot\|x-y\|
        \end{equation*}
        for all $x,y\in\Rbb^m$.
        So, taking $x=\delta_{\B}(Z_i)$ and $y=\hat{\delta}_{\B}(Z_i)$, summing, and applying Jensen's inequality yields
        \begin{align*}
        &\int_{\Rbb^d\times K}\left|\|\delta_{\B}(z)\|^2-\|\hat{\delta}_{\B}(z)\|^2\right|\diff \pi_{n}(z,\eta) \\
        &\le \frac{1}{n}\sum_{i=1}^{n}\left|\|\delta_{\B}(Z_i)\|^2-\|\hat{\delta}_{\B}(Z_i)\|^2\right| \\
        &\le \frac{1}{n}\sum_{i=1}^{n}\left(\|\delta_{\B}(Z_i)-\hat{\delta}_{\B}(Z_i)\|^2 + 2\|\delta_{\B}(Z_i)\|\cdot\|\delta_{\B}(Z_i)-\hat{\delta}_{\B}(Z_i)\|\right) \\
        &= \frac{1}{n}\sum_{i=1}^{n}\|\delta_{\B}(Z_i)-\hat{\delta}_{\B}(Z_i)\|^2 + \frac{2}{n}\sum_{i=1}^{n}\|\delta_{\B}(Z_i)\|\cdot\|\delta_{\B}(Z_i)-\hat{\delta}_{\B}(Z_i)\| \\
        &\le \frac{1}{n}\sum_{i=1}^{n}\|\delta_{\B}(Z_i)-\hat{\delta}_{\B}(Z_i)\|^2 + 2\cdot\sqrt{\frac{1}{n}\sum_{i=1}^{n}\|\delta_{\B}(Z_i)-\hat{\delta}_{\B}(Z_i)\|^2}\sqrt{\frac{1}{n}\sum_{i=1}^{n}\|\delta_{\B}(Z_i)\|^2}.
        \end{align*}
    The right side above vanishes as $n\to\infty$ because of the assumption~\eqref{eqn:EBQ}, assumption~\eqref{eqn:quadratic-growth}, and the law of large numbers.
        For the second term above, use Cauchy-Schwarz twice to get
        \begin{align*}
            &\int_{\Rbb^d\times K}\left|\left\langle\delta_{\B}(z)-\hat{\delta}_{\B}(z),\eta\right\rangle\right|\diff \pi_{n}(z,\eta) \\
            &\le \int_{\Rbb^d\times K}\left\|\delta_{\B}(z)-\hat{\delta}_{\B}(z)\right\|\cdot\|\eta\|\diff \pi_{n}(z,\eta) \\
            &\le \sqrt{\int_{\Rbb^d\times K}\left\|\delta_{\B}(z)-\hat{\delta}_{\B}(z)\right\|\diff \pi_{n}(z,\eta)}\sqrt{\int_{\Rbb^d\times K}\|\eta\|^2\diff \pi_n(z,\eta)} \\
            &= \sqrt{\frac{1}{n}\sum_{i=1}^{n}\left\|\delta_{\B}(Z_i)-\hat{\delta}_{\B}(Z_i)\right\|^2}\sqrt{\int_{K}\|\eta\|^2\diff\hat  H_n(\eta)}.
        \end{align*}
        As before, the first term vanishes because of \eqref{eqn:EBQ}, and we claim that the second term is $O_{\Pbb}(1)$. 
        
        To show this, use that the product coupling $\bar{F}_n\otimes\hat{G}_n$ is feasible for the linear program in Algorithm~\ref{alg:distr-cstr}, so optimality of $\hat{\pi}_n$ yields
    \begin{equation*}
        \int_{\Rbb^d\times\Rbb^m}\hat{c}_n\diff \hat{\pi}_n \le \int_{\Rbb^d\times\Rbb^m}\hat{c}_n\diff (\bar{F}_n\otimes\hat{G}_n) \le 2\int_{\Rbb^m}\|\eta\|^2\diff \hat{G}_n(\eta) + 2\left\|\hat{\delta}_{\B}(\mathbf{Z})\right\|_n^2.
    \end{equation*}
    Then use the elementary inequality $\|\eta\|^2\le 2\|\eta-\hat{\delta}_{\B}(z)\|^2+2\|\hat{\delta}_{\B}(z)\|^2$ and recall that the first marginal of $\hat{\pi}_n$ is $\bar{F}_n$ to further obtain
    \begin{equation*}
        \int_{\Rbb^m}\|\eta\|^2\diff \hat{H}_n(\eta) \le 2\int_{\Rbb^d\times\Rbb^m}\hat{c}_n\diff\hat{\pi}_n + 2\left\|\hat{\delta}_{\B}(\mathbf{Z})\right\|_n^2 \le 4\int_{\Rbb^m}\|\eta\|^2\diff \hat{G}_n(\eta) + 6\left\|\hat{\delta}_{\B}(\mathbf{Z})\right\|_n^2.
    \end{equation*}
    The first term on the right is $O_{\Pbb}(1)$ since \eqref{eqn:CM} gives $W_2(\hat{G}_n,G)\to 0$ in probability, and that convergence in $W_2$ implies convergence of second moments \citep[Theorem~6.9]{Villani}; the second term is also $O_{\Pbb}(1)$, since we have $\|\hat{\delta}_{\B}(\mathbf{Z})\|_n \le \|\delta_{\B}(\mathbf{Z})\|_n + \|\hat{\delta}_{\B}(\mathbf{Z})-\delta_{\B}(\mathbf{Z})\|_n,$ which are $O_{\Pbb}(1)$ and $o_{\Pbb}(1)$ by \eqref{eqn:2M}  and \eqref{eqn:EBQ}, respectively.
    This shows \eqref{eqn:integral-estimate-H}, as needed.

    Lastly, we show that $\pi_{\infty}$ is in fact optimal for \eqref{eqn:pop-distrcstr-Monge}.
        To do this, take an arbitrary $\pi\in\Gamma(F;H)$ and use Lemma~\ref{lem:prelim-cvg-gen} to get some $\{\pi_k\}_{k\in\Nbb}$ with $\pi_k\in\Gamma(\bar F_{n_k};  \hat H_{n_k})$ and $W_2(\pi_k,\pi)\to 0$ as $k\to\infty$.
        Then use the Portmanteau lemma, the estimate~\eqref{eqn:integral-estimate-H}, the optimality of $\{\hat \pi_{n_k}\}_{k\in\Nbb}$, and the estimate~\eqref{eqn:integral-estimate-H} again, to get:
        \begin{align*}
            \int_{\Rbb^d\times K}c_G\,\diff \pi_{\infty} &\le \liminf_{k\to\infty}\int_{\Rbb^d\times K}c_G\,\diff \hat{\pi}_{n_k} \\
            &\le \liminf_{k\to\infty}\int_{\Rbb^d\times K}\hat{c}_{n_k}\,\diff \hat{\pi}_{n_k} \\
            &\le \liminf_{k\to\infty}\int_{\Rbb^d\times K}\hat{c}_{n_k}\,\diff \pi_{k} \\
            &\le \liminf_{k\to\infty}\int_{\Rbb^d\times K}c_G\,\diff \pi_{k}
        \end{align*}
        In fact, we have
        \begin{equation*}
            \lim_{k\to\infty}\int_{\Rbb^d\times K}c_G\,\diff \pi_{k} = \int_{\Rbb^d\times K}c_G\,\diff \pi
        \end{equation*}
        by $W_2$ convergence, since $c_G$ is uniformly integrable under $\{\pi_k\}_{k\in\Nbb}$.
        Thus $\pi_{\infty}$ is optimal, hence $\pi_{\infty} =\pi_{\gencstr}$.
        This completes the proof.
    \end{proof}

\end{document}